\documentclass[a4paper,11pt,reqno,noindent]{amsart}
\usepackage[centertags]{amsmath}
\usepackage{mathtools} 
\usepackage{amsfonts,amssymb,amsthm,dsfont,cases,amscd,esint,enumerate}
\usepackage[T1]{fontenc}
\usepackage[english,frenchb]{babel}
\usepackage[applemac]{inputenc}
\usepackage{newlfont}
\usepackage{color}
\usepackage[body={15cm,21.5cm},centering]{geometry} 
\usepackage{fancyhdr}
\usepackage{enumerate}

\theoremstyle{plain}
\newtheorem{theor10}{Theorem}
\newenvironment{theor1}
  {\pushQED{\qed}\begin{theor10}}
  {\popQED\end{theor10}}
\newtheorem{prop10}[theor10]{Proposition}
\newenvironment{prop1}
  {\pushQED{\qed}\begin{prop10}}
  {\popQED\end{prop10}}
\newtheorem{cor10}[theor10]{Corollary}
\newenvironment{cor1}
  {\pushQED{\qed}\begin{cor10}}
  {\popQED\end{cor10}}
\newtheorem{lem0}{Lemma}[section]
\newenvironment{lem}
  {\pushQED{\qed}\begin{lem0}}
  {\popQED\end{lem0}}
\newtheorem{theor0}[lem0]{Theorem}
\newenvironment{theor}
  {\pushQED{\qed}\begin{theor0}}
  {\popQED\end{theor0}}
\newtheorem{prop0}[lem0]{Proposition}
\newenvironment{prop}
  {\pushQED{\qed}\begin{prop0}}
  {\popQED\end{prop0}}
\newtheorem{cor0}[lem0]{Corollary}

\theoremstyle{definition}
\newtheorem{defin0}[lem0]{Definition}
\newenvironment{defin}
  {\pushQED{\qed}\begin{defin0}}
  {\popQED\end{defin0}}
\newtheorem{rems0}[lem0]{Remarks}
\newenvironment{rems}
  {\pushQED{\qed}\begin{rems0}}
  {\popQED\end{rems0}}
\newtheorem{rem0}[lem0]{Remark}
\newenvironment{rem}
  {\pushQED{\qed}\begin{rem0}}
  {\popQED\end{rem0}}
 
\theoremstyle{plain}
\newtheorem{conj0}[theor10]{Conjecture}

\newtheorem*{conjn0*}{\assumptionnumber}
  \providecommand{\assumptionnumber}{}
  \makeatletter
  \newenvironment{conjn0}[2]
   {\renewcommand{\assumptionnumber}{Conjecture (#1) {\normalfont--- #2}}
    \begin{conjn0*}
    \protected@edef\@currentlabel{{\normalfont(#1)}}}
   {\end{conjn0*}}
  \makeatother
\newenvironment{conjn}
  {\pushQED{\qed}\begin{conjn0}}
  {\popQED\end{conjn0}}
  
\numberwithin{equation}{section}
\mathchardef\emptyset="001F

\newcommand{\ac}{{\operatorname{ac}}}
\newcommand{\pp}{{\operatorname{pp}}}
\newcommand{\sg}{{\operatorname{sc}}}

\newcommand{\N}{\mathbb N}
\newcommand{\e}{\varepsilon}
\newcommand{\R}{\mathbb R}
\newcommand{\Z}{\mathbb Z}
\newcommand{\C}{\mathbb C}
\newcommand{\T}{\mathbb T}
\newcommand{\Dm}{\mathbb D}
\newcommand{\Cc}{\mathcal C}
\newcommand{\Pc}{\mathcal P}
\newcommand{\Rc}{\mathcal R}

\newcommand{\Hf}{\mathfrak H}
\newcommand{\Bc}{\mathcal B}
\newcommand{\F}{\mathcal F}
\newcommand{\Lc}{\mathcal L}
\newcommand{\Vc}{\mathcal V}

\newcommand{\Dc}{\mathcal D}
\newcommand{\Hc}{\mathcal H}
\newcommand{\Op}{\operatorname{Op}}
\newcommand{\pv}{\operatorname{p.v.}}
\newcommand{\st}{{\operatorname{st}}}
\newcommand{\Leb}{\mathcal{L}eb}
\newcommand{\A}{\mathcal A}

\newcommand{\Id}{\operatorname{Id}}
\newcommand{\p}{\mathbb{P}}
\newcommand{\sgn}{\operatorname{sgn}}
\newcommand{\supess}{\operatorname{sup\,ess}}
\newcommand{\infess}{\operatorname{inf\,ess}}

\newcommand{\conv}{{\operatorname{conv}}}
\newcommand{\loc}{{\operatorname{loc}}}
\newcommand{\comp}{{\operatorname{comp}}}
\newcommand{\Sym}{{\operatorname{sym}}}

\newcommand{\per}{{\operatorname{per}}}
\newcommand{\adh}{{\operatorname{adh\,}}}
\newcommand{\dom}{{\operatorname{dom}}}
\newcommand{\Ld}{\operatorname{L}}
\newcommand{\supp}{\operatorname{supp}}
\newcommand{\diam}{\operatorname{diam}}
\newcommand{\inter}{\operatorname{int}}

\newcommand{\step}[1]{\noindent \textit{Step} #1.}
\newcommand{\substep}[1]{\noindent \textit{Substep} #1.}
\newcommand{\cov}[2]{\operatorname{Cov}\left[{#1};{#2}\right]}
\newcommand{\var}[1]{\operatorname{Var}\left[{#1}\right]}
\newcommand{\pr}[1]{\mathbb{P}\left[{#1}\right]}
\newcommand{\prm}[1]{\mathbb{P}\big[{#1}\big]}
\newcommand{\E}{\mathbb{E}}
\newcommand{\expec}[1]{\mathbb{E}\left[{#1}\right]}
\newcommand{\expecm}[1]{\mathbb{E}\big[{#1}\big]}
\newcommand{\ee}{e}
\def\dbar{\,{{\mathchar'26\mkern-12mud}}}

\usepackage[symbol]{footmisc}

\usepackage[colorlinks,citecolor=black,urlcolor=black]{hyperref}

\title[A new spectral analysis of random Schr\"odinger operators]{A new spectral analysis of stationary random Schr\"odinger operators}

\begin{document}
\selectlanguage{english}
\maketitle

\begin{center}
{\footnotesize MITIA DUERINCKX${}^{1,2}$\!\footnote{mitia.duerinckx@u-psud.fr} AND CHRISTOPHER SHIRLEY${}^1$\!\footnote{christopher.shirley@universite-paris-saclay.fr}\\$ $\\
{\it ${}^1$Université Paris-Saclay, CNRS, Laboratoire de Mathématiques d'Orsay, 91405~Orsay, France\\
${}^2$Universit\'e Libre de Bruxelles, Département de Mathématique, 1050~Brussels, Belgium}
}
\end{center}

\begin{abstract}
Motivated by the long-time transport properties of quantum waves in weakly disordered media,
the present work puts random Schrödinger operators into a new spectral perspective. Based on a stationary random version of a Floquet type fibration, we reduce the description of the quantum dynamics to a fibered family of abstract spectral perturbation problems on the underlying probability space.
We state a natural resonance conjecture for these fibered operators: in contrast with periodic and quasiperiodic settings, this would entail that Bloch waves do not exist as extended states, but rather as resonant modes, and this would justify the expected exponential decay of time correlations.
Although this resonance conjecture remains open, we develop new tools for spectral analysis on the probability space, and in particular we show how ideas from Malliavin calculus lead to rigorous Mourre type results: we obtain an approximate dynamical resonance result and the first spectral proof of the decay of time correlations on the kinetic timescale.
This spectral approach suggests a whole new way of circumventing perturbative expansions and renormalization techniques.
\end{abstract}

%%%%%%%%%%%%%%%%%%%%%
%%%%%%%%%%%%%%%%%%%%%
%%%%%%%%%%%%%%%%%%%%%

\section{Introduction}
\subsection{General overview}
We consider random Schrödinger operators of the form
\[H_{\lambda,\omega}:=-\triangle+\lambda V_\omega\]
on $\Ld^2(\R^d)$, where $V_\omega:\R^d\to\R$ is a realization of a ``stationary'' (that is, statistically translation-invariant) random potential $V$, constructed on a given probability space $(\Omega,\p)$,
and we study the properties of the corresponding Schrödinger flow on $\R^d$,
\[i\partial_t u_{\lambda,\omega}= H_{\lambda,\omega} u_{\lambda,\omega},\qquad u_{\lambda,\omega}|_{t=0}=u^\circ,\]
with initial data $u^\circ\in\Ld^2(\R^d)$.
This well-travelled equation models the motion of an electron in a disordered medium described by the potential $V_\omega$, where the coupling constant~$\lambda>0$ stands for the strength of the disorder.

\medskip
For comparison, let us first recall transport properties in the simpler case of periodic or quasiperiodic disorder.
If $V$ is periodic, the energy transport is well-known to remain purely ballistic as for the free flow $u_{\lambda=0}$, cf.~\cite{Asch-Knauf-98}.
The proof relies on the absolute continuity of the spectrum of the periodic Schrödinger operator on~$\Ld^2(\R^d)$, and more precisely on the existence of so-called Bloch waves, which are extended states
constructed by means of standard perturbation theory as deformations of Fourier modes, $x \mapsto e^{ik\cdot x} \varphi_{k,\lambda}(x)$ with $\varphi_{k,\lambda}$ periodic.
In case of a quasiperiodic potential~$V$,
the problem is more involved and depends on the strength of the disorder: energy transport is expected to remain ballistic only at weak coupling $0<\lambda\ll1$ or at high energies. This is rigorously established in dimension~$d=1$~\cite{Zhao-16}. In higher dimensions $d>1$, for a specific class of quasiperiodic potentials, it was recently shown that there exist initial data at high energies that indeed display ballistic transport~\cite{Karpeshina-Lee-Shterenberg-Stolz-15,Karpeshina-Parnovski-Shterenberg-20} (the analysis of the discrete setting is more complete~\cite{Bourgain-07}). Inspired by the periodic case, the proof relies on the existence of corresponding Bloch waves as extended states of the form $x\mapsto e^{ik\cdot x}\varphi_{k,\lambda}(x)$ with $\varphi_{k,\lambda}$ quasiperiodic, but their construction is more intricate since standard perturbation theory no longer applies.
In~\cite{DGS1}, we provide a simple method to construct ``approximate'' Bloch waves and deduce ballistic transport for all data at least up to ``very long'' timescales both at weak coupling and at high energies.
These results in the periodic and quasiperiodic settings show how Bloch waves are crucial tools to infer transport properties of the Schrödinger flow.

\medskip
The present work is concerned with the more general stationary random setting in the weak coupling regime \mbox{$0<\lambda\ll1$}.
In case of a random potential $V$ with short-range correlations, in stark contrast with the periodic and quasiperiodic cases, a celebrated conjecture by Anderson~\cite{Anderson-58} states that in dimension $d>2$ every initial condition can be almost surely decomposed into two parts: a low-energy part that remains dynamically localized and a bulk-energy part that propagates diffusively.
Despite the great recent achievements of rigorous perturbation theory in some asymptotic time regimes, e.g.~\cite{Spohn-77,Erdos-Yau-00,Erdos-Salmhofer-Yau-1,Erdos-Salmhofer-Yau-2,Chen-Komorowski-Ryzhik-16}, successfully describing the emergence of irreversible diffusion from the reversible Schrödinger dynamics,
the full justification of this quantum diffusion phenomenon remains a major open problem in mathematical physics~\cite{Simon-00,Erdos-rev}.
More precisely, the ensemble-averaged Wigner transform of the quantum wave $u_\lambda$ is known to converge to the solution of a linear Boltzmann equation on the kinetic timescale $t\sim\lambda^{-2}$, and of a heat equation on longer times, but the justification is limited to a perturbative time regime $t\ll\lambda^{-2-\eta}$ for some small $\eta>0$.
A simplified question concerns the behavior of time correlations in form of the averaged wavefunction $\expec{u_\lambda^t}$,
which is expected to display exponential time decay: more precisely, in Fourier space, on the kinetic timescale $t\sim\lambda^{-2}$, as $\lambda\downarrow0$,
\begin{equation}\label{eq:exp-decay}
\big|\expecm{\widehat u_\lambda^{t=\lambda^{-2}s}(k)}\big|\,\sim\,e^{-s\alpha_k}\,|\widehat u^\circ(k)|,
\end{equation}
where the decay rate $\alpha_k>0$
would coincide with the total scattering cross section in the corresponding Boltzmann equation,
and where corrections are added on longer times.
A proof of this exponential decay on the kinetic timescale is given in~\cite{Chen-Komorowski-Ryzhik-16} based on a perturbative expansion of a Feynman-Kac type formula. The perturbative analysis of~\cite{Erdos-Yau-00,Erdos-Salmhofer-Yau-1,Erdos-Salmhofer-Yau-2} would further yield an improved result, but still restricted to limited timescales.

\medskip
Motivated by these open questions, rather than trying to improve on perturbative expansions and renormalization techniques, we aim at developping an alternative spectral approach to describe the long-time behavior of the system beyond perturbative timescales.
More precisely, we take inspiration from the periodic and quasiperiodic cases, although the behavior radically differs from the present random setting,
and we investigate the role of a corresponding notion of Bloch waves.
It appears that these Bloch waves are no longer extended states associated with absolutely continuous spectrum: they are expected to be only defined in a weak distributional sense in probability and to play the role of resonant modes associated with some kind of ``continuous resonant spectrum''.
Exploiting ideas from Malliavin calculus, we manage to appeal to perturbative Mourre's theory, cf.~Theorem~\ref{th:Mourre-pert}, which leads to the construction of approximate dynamical resonances and constitutes the first spectral proof of~\eqref{eq:exp-decay}, cf.~Corollary~\ref{th:mourre-partial}.
Non-perturbative refinements to reach longer times are postponed to future works, as well as the investigation of other possible dynamical consequences in closer connection with quantum diffusion.

\subsection{Summary of our approach and results}
We briefly describe the framework of our new approach to Schrödinger operators $H_{\lambda,\omega}=-\triangle+\lambda V_\omega$.
First, we change the point of view and rather consider the operator \mbox{$H_\lambda:=-\triangle+\lambda V$} as acting on the augmented Hilbert space $\Ld^2(\R^d\times\Omega)$, then studying the corresponding Schrödinger flow on~$\Ld^2(\R^d\times\Omega)$,
\begin{equation}\label{eq:schr}
i\partial_t u_{\lambda}= H_{\lambda} u_{\lambda},\qquad
u_{\lambda}|_{t=0}=u^\circ,
\end{equation}
with deterministic initial data $u^\circ\in\Ld^2(\R^d)\subset\Ld^2(\R^d\times\Omega)$. This can be viewed as including stochastic averaging conveniently into the functional setup; see also~\cite{Pillet-85,Schenker-Kang}.
(Note that $H_\lambda$ on $\Ld^2(\R^d\times\Omega)$ has absolutely continuous spectrum as a consequence of Wegner estimates when~$H_{\lambda,\omega}$ on~$\Ld^2(\R^d)$ has almost sure pure point spectrum.)

\medskip
As we have shown and already used in~\cite{DGS1}, see~Section~\ref{sec:Floquet--Bloch} below for details (also~\cite{BG-16,Schenker-Kang}), the operator $H_\lambda$ on $\Ld^2(\R^d\times\Omega)$ can be decomposed via a Fourier-type transformation as a direct integral of fibered operators $\{H_{k,\lambda}^\st\}_k$ acting on the elementary space $\Ld^2(\Omega)$, which is viewed as the space of stationary random fields,
\begin{equation}\label{eq:fibration}
\big(H_\lambda,\Ld^2(\R^d\times\Omega)\big)=\int_{\R^d}^\oplus\big(H_{k,\lambda}^\st+|k|^2,\Ld^2(\Omega)\big)\,\mathfrak e_k\,\dbar k,\qquad\mathfrak e_k(x):=e^{ik\cdot x}.
\end{equation}
The (centered) fibered operators take the form
\begin{eqnarray*}
H_{k,\lambda}^\st&:=&
H_{k,0}^\st+\lambda V,\\
H_{k,0}^\st&:=&-(\nabla^\st+ik)\cdot(\nabla^\st+ik)-|k|^2\,=\,-\triangle^\st-2ik\cdot\nabla^\st,
\end{eqnarray*}
where $\nabla^\st$ and $\triangle^\st$ denote the stationary gradient and Laplacian on $\Ld^2(\Omega)$; see Section~\ref{sec:stat-calc} for proper definitions.
In particular, the Schrödinger flow $u_\lambda$ is decomposed as
\begin{equation}\label{eq:flow-decomp0}
u_\lambda^t(x,\omega)\,=\,\int_{\R^d}\widehat u^\circ(k)\,e^{ik\cdot x-it|k|^2}\big(e^{-itH^\st_{k,\lambda}}1\big)(x,\omega)\,\dbar k,
\end{equation}
in terms of the fibered evolutions $\{e^{-itH^\st_{k,\lambda}}1\}_k$ on $\Ld^2(\Omega)$.
This partial diagonalisation via Fourier is henceforth referred to as the stationary Floquet--Bloch fibration, in analogy with the well-known corresponding construction in the periodic setting,
e.g.~\cite{Kuchment-93,Kuchment-16}.

\medskip
At vanishing disorder $\lambda=0$, as the constant function $1\in\Ld^2(\Omega)$ is an eigenfunction with eigenvalue~$0$ for the unperturbed fibered operators $\{H_{k,0}^\st\}_k$, the associated spectral measure coincides with the Dirac mass at $0$, and the decomposition~\eqref{eq:flow-decomp0} then reduces to the usual Fourier diagonalisation of the free flow,
\begin{align}\label{eq:free-flow}
u_{\lambda=0}^t(x)\,=\,\int_{\R^d}\widehat u^\circ(k)\,e^{ik\cdot x-it|k|^2}\,\dbar k.
\end{align}
When the disorder is turned on but small, $0<\lambda\ll1$, the description of the Schrödinger flow $u_\lambda$ is reduced to a family of (hopefully simpler) fibered perturbation problems for the spectral measures.
In case of a periodic potential~$V$ (that is,~$\Omega=\T^d$, cf.~Remark~\ref{rem:stat}(a)), the fibered operators $\{H_{k,0}^\st\}_k$ have compact resolvent in view of the Rellich theorem, and thus discrete spectrum. The eigenvalue at~$0$ is then typically simple and isolated, which allows to apply standard perturbation methods, e.g.~\cite{Kuchment-93,Asch-Knauf-98}, showing that it is perturbed into isolated eigenvalues $\{z_{k,\lambda}=\lambda^2z_{\lambda}(k)\}_k$.
In other words, Fourier modes $\{\mathfrak e_k\}_{k}$ are perturbed into so-called periodic Bloch waves that diagonalize the Schrödinger flow and are associated with perturbed generalized eigenvalues of the form $\{|k|^2+\lambda^2z_\lambda(k)\}_{k}$. This entails that the flow is approximately conjugated to the free flow~\eqref{eq:free-flow} in the sense of
\begin{align*}
u_\lambda^t(x)\,=\,\int_{\R^d}\widehat u^\circ(k)\,e^{ik\cdot x-it(|k|^2+\lambda^2z_{\lambda}(k))}\,\dbar k\,+O(\lambda),\qquad z_{\lambda}\in C^\infty(\R^d;\R),
\end{align*}
and in particular the energy transport remains ballistic forever. The same conclusion holds in fact for all $\lambda$, cf.~\cite{Asch-Knauf-98}.
In case of a quasiperiodic potential $V$ (cf.\@ Remark~\ref{rem:stat}(a)), the situation is expected to be similar at weak coupling, but the existence of corresponding Bloch waves is a more subtle question: quasiperiodic fibered operators $\{H_{k,0}^\st\}_k$ are degenerate elliptic operators, for which compactness fails, and the simple eigenvalue at $0$ is no longer isolated but embeds in dense pure point spectrum, so that no standard perturbation theory applies;
see~\cite{Karpeshina-Lee-Shterenberg-Stolz-15,Karpeshina-Parnovski-Shterenberg-20,DGS1}.

\medskip
In case of a random potential $V$ with short-range correlations, the situation differs drastically in link with the expected diffusive behavior.
We show that~$0$ is typically the only eigenvalue of the fibered operators~$\{H_{k,0}^\st\}_k$, is simple, and embeds in absolutely continuous spectrum, cf.~Proposition~\ref{th:main0} below.
According to Fermi's Golden Rule, whenever the disorder is turned on, this embedded eigenvalue is then expected to dissolve in the continuous spectrum, cf.~Proposition~\ref{prop:Fermi}, and to turn into a complex resonance at
\[z_{k,\lambda}=\lambda^2z_\lambda(k)=\lambda^2\expec{V(i0-H_{k,0}^\st)^{-1}V}+O_k(\lambda^3),\]
in the lower complex half-plane.
In particular, this provides a spectral explanation why approximate Bloch wave analysis leading to ballistic transport as in~\cite{DGS1} breaks down on the kinetic timescale $t\sim\lambda^{-2}$. For the averaged wavefunction, this leads to expect
\begin{gather*}
\expec{u_\lambda^t(x)}\,=\,\int_{\R^d}\widehat u^\circ(k)\,e^{ik\cdot x-it(|k|^2+\lambda^2z_{\lambda}(k))}\,\dbar k\,+O(\lambda),\qquad
z_{\lambda}\in C^\infty(\R^d;\C\setminus\R),
\end{gather*}
which would indeed agree with the exponential decay~\eqref{eq:exp-decay} on the kinetic timescale, with $\Im z_{\lambda}(k)\sim-\alpha_k<0$, and a finer resonance analysis would yield a more accurate expansion.
From a spectral perspective, fibered resonances are transferred via the fibration~\eqref{eq:fibration} to kind of a ``continuous resonant spectrum'' for the full operator $H_\lambda$ on~$\Ld^2(\R^d\times\Omega)$, cf.~Remark~\ref{rem:res-spec} below.

\medskip
General spectral tools are however dramatically missing to rigorously study these fibered perturbation problems on $\Ld^2(\Omega)$, in particular due to the lack of any relative compactness of the perturbation.
We start by performing
a detailed study of rudimentary spectral properties of fibered operators, cf.~Propositions~\ref{th:main0}--\ref{prop:Fermi},
emphasizing the strong dependence on the structure of the underlying probability space $(\Omega,\p)$.
Next, we appeal to Mourre's theory~\cite{Mourre-80,ABMG-96} as a rigorous approach to fibered perturbation problems.
More precisely, we start by constructing Mourre conjugates for the unperturbed operators~$\{H^\st_{k,0}\}_k$. The construction requires a surprisingly nontrivial work and relies on a deep use of Malliavin calculus, which constitutes the core of our contribution, cf.~Section~\ref{sec:Mourre}.
This construction is however not compatible with the perturbation $V$ in the sense that $\lambda V$ cannot be considered as a small perturbation of $\{H_{k,0}^\st\}_k$ in the sense of Mourre's theory, in link with the infinite dimensionality of the probability space.
For this reason, we only manage to apply perturbative Mourre's theory under a suitable (weak) truncation, cf.~Theorem~\ref{th:Mourre-pert}.
As a direct consequence, the decay law~\eqref{eq:exp-decay} is recovered at least on the kinetic timescale, cf.~Corollary~\ref{th:mourre-partial}.
Finally, we give a relevant formulation of resonance conjectures for fibered operators, cf.~Conjectures~\ref{C1} and~\ref{C2},
which are motivated by our partial results
and are shown to imply the expected decay law~\eqref{eq:exp-decay} to finer accuracy on all timescales, cf.~Corollary~\ref{th:res-spectr}.
These conjectures are further illustrated in Section~\ref{sec:toy}, where we display a toy model that shares various properties of Schrödinger operators and allows for a rigorous resonance analysis.
Although these conjectures are left open, the present work sheds a new light on the study of random Schrödinger operators, in particular providing the first spectral proof of~\eqref{eq:exp-decay}; our results will be strengthened in future works and hopefully serve as a starting point for a new line of research in the field.

\medskip
\subsection*{Notation}
\begin{enumerate}[$\bullet$]
\item We denote by $C\ge1$ any constant that only depends on the space dimension $d$ and on the law of the random potential~$V$.
We use the notation $\lesssim$ (resp.~$\gtrsim$) for $\le C\times$ (resp.\@ $\ge\frac1C\times$) up to such a multiplicative constant $C$. We write $\simeq$ when both $\lesssim$ and $\gtrsim$ hold. We add subscripts to $C,\lesssim,\gtrsim,\simeq$ to indicate dependence on other parameters. We denote by $O(K)$ any quantity that is bounded by $CK$.
\item We denote by $\widehat f(k):=\F f(k):=\int_{\R^d}e^{-ik\cdot x}f(x)\,dx$ the usual Fourier transform of a smooth function~$f$ on $\R^d$. The inverse Fourier transform is then given by $f(x)=\int_{\R^d}e^{ik\cdot x}\widehat f(k)\,\dbar k$ in terms of the rescaled Lebesgue measure $\dbar k:=(2\pi)^{-d}dk$.
\item The ball centered at $x$ and of radius $r$ in $\R^d$ is denoted by $B_r(x)$, and we write for abbreviation $B(x):=B_1(x)$, $B_r:=B_r(0)$, and $B:=B_1(0)$. Without ambiguity, we occasionally also denote by $B$ the unit ball at the origin in the complex plane $\C$.
\item For a set $E\subset\R^d$ we denote by $\conv(E)$ its convex envelope, by $\inter(E)$ its interior, and by $\adh(E)$ its closure.
\item We denote by $\Bc(\R^k)$ the set of Borel subsets of $\R^k$, and for $E\subset\R^k$ we let $\Pc(E)$ denote the set of Borel probability measures on $E$.
\item For a vector space $X$, we write $X^{\otimes p}$ for its $p$-fold tensor product, and $X^{\odot p}$ for its $p$-fold symmetric tensor product.
\item For $a,b\in\R$, we write $a\wedge b:=\min\{a,b\}$ and $a\vee b:=\max\{a,b\}$.
\end{enumerate}

%%%%%%%%%%%%%%%%%%%%%
%%%%%%%%%%%%%%%%%%%%%
%%%%%%%%%%%%%%%%%%%%%

\medskip
\section{Main results}
This section is devoted to a brief description of our main results, while proofs and detailed statements are postponed to the next sections.

\subsection{Framework}
We refer to Section~\ref{sec:Floquet--Bloch} for a suitable definition of stationarity as statistical invariance under spatial translations.
Throughout, the stationary random potential~$V$ is assumed real-valued and centered, $\expec{V}=0$.
As we show, fine spectral properties crucially depend on the structure of the underlying probability space. We therefore mainly focus on Gaussian or Poisson settings, where Malliavin calculus is available and provides a useful Fock space decomposition of $\Ld^2(\Omega)$.
More precisely, we consider the following:
\begin{enumerate}[\quad$\bullet$]
\item {\it Gaussian setting:} $V=b(V_0)$ for some Borel function $b:\R\to\R$ and some stationary centered Gaussian random field $V_0$ with bounded covariance function
\[\qquad\Cc_0(x):=\expec{V_0(x+y)V_0(y)}.\]
Equivalently, the field $V_0$ can be represented as
\begin{equation}\label{eq:white-rep}
\qquad V_0(x)=\int_{\R^d}\Cc_0^\circ(x+y)\,dZ(y),
\end{equation}
where $dZ$ is a standard Gaussian white noise on $\R^d$ and where the kernel $\Cc_0^\circ$ is the convolution square root of the covariance function, $\Cc_0=\Cc_0^\circ\ast\Cc_0^\circ$.
\item {\it Poisson setting:} $V=b(V_0)$ for some Borel function $b:\R\to\R$ and some $V_0$ of the form
\begin{equation}\label{eq:Poiss}
\qquad V_0(x)=\sum_{y\in\Pc_0}\Cc_0^\circ(x+y),
\end{equation}
where $\Pc_0$ is a standard Poisson point process on $\R^d$ and where $\Cc_0^\circ$ is the single-site potential.
\end{enumerate}
We say that the random potential $V$ is short-range if it has integrable decay of correlations:
in the above settings, this amounts to choosing $\Cc_0^\circ\in\Ld^1(\R^d)$.

\medskip
For shortness, in the sequel, we shall mainly restrict to the Gaussian setting, although the same results can be transferred mutatis mutandis in the Poisson case (using the corresponding version of Malliavin calculus, e.g.~\cite{Peccati-Reitzner}).
For simplicity, we occasionally further restrict to a random potential $V=V_0$ that is itself Gaussian: although unbounded, such potentials have a simpler action on the Fock space decomposition of $\Ld^2(\Omega)$.

\subsection{Basic spectral theory of fibered operators}\label{sec:spectrum0}
We refer to Section~\ref{sec:Floquet--Bloch} for the construction of the stationary Floquet--Bloch fibration~\eqref{eq:fibration}.
Next, we start with a detailed spectral analysis of the unperturbed operators
\[H^\st_{k,0}:=-\triangle^\st-2ik\cdot\nabla^\st,\qquad k\in\R^d.\]
Although the stationary Laplacian $-\triangle^\st$ is a natural operator on $\Ld^2(\Omega)$ and has been introduced in various settings (e.g.\@ in the context of stochastic homogenization~\cite{PapaVara,JKO94}), its spectral properties have never been elucidated before, and we close this gap here.
Note that some preliminary remarks on its spectrum have been made in~\cite[Section~3.1]{Cances-Lahbabi-Lewin-13}, see also~\cite[Section~2.C]{Lahbabi-13}, namely that it is discrete if $\Omega$ is a finite set, that there is in general no spectral gap above~$0$ in contrast with the periodic setting, and that it coincides with $[0,\infty)$ in case of an i.i.d.\@ structure.
Interestingly, the spectrum depends crucially on the structure of the underlying probability space $(\Omega,\p)$, as precisely formulated in Section~\ref{sec:spectrum-Hk0} below in terms of a notion of ``spectrum'' of the probability space.
In the model Gaussian setting, our result takes on the following simple guise.

\begin{prop1}[Spectral decomposition of $H_{k,0}^\st$]\label{th:main0}
Given a stationary Gaussian field $V_0$ on $\R^d$ with covariance function $\Cc_0$, denote by~$\widehat\Cc_0$ the (nonnegative measure) Fourier transform of~$\Cc_0$, and assume that the probability space $(\Omega,\p)$ is endowed with the $\sigma$-algebra  $\sigma(\{V_0(x)\}_{x\in\R^d})$ generated by $V_0$.
\begin{enumerate}[(i)]
\item If $\Cc_0$ is not periodic in any direction, then $\sigma(H_{k,0}^\st)=[-|k|^2,\infty)$.
\item If the measure $\widehat\Cc_0$ is absolutely continuous (in particular, if $\Cc_0$ is integrable),
then the eigenvalue at $0$ is simple (with eigenspace $\C$) and
\[\sigma_{\pp}(H_{k,0}^\st)\,=\,\{0\},\quad\sigma_\sg(H_{k,0}^\st)\,=\,\varnothing,\quad\sigma(H_{k,0}^\st)\,=\,\sigma_{\ac}(H_{k,0}^\st)\,=\,[-|k|^2,\infty).\qedhere\]
\end{enumerate}
\end{prop1}

We turn to the perturbed fibered operators $\{H^\st_{k,\lambda}\}_k$ and start with a characterization of their spectrum.
While we focus here on the Gaussian setting, a more general statement is given in Section~\ref{sec:spectrum-fibered}.
In case of an unbounded potential~$V$, the essential self-adjointness of the perturbed operators $\{H_{k,\lambda}^\st\}_k$ is already a delicate issue, for which an (almost optimal) criterion is included in Appendix~\ref{sec:self-adj}, requiring $V\in\Ld^p(\Omega)$ for some $p>\frac d2$, in line with the corresponding celebrated self-adjointness problem for Schrödinger operators on~$\Ld^2(\R^d)$ with singular potentials, cf.~\cite{Kato-72,Faris-Lavine-74,Kato}; see also~Proposition~\ref{prop:nelson} for the simpler case when $V=V_0$ is itself Gaussian.

\begin{prop1}[Spectrum of $H_{k,\lambda}^\st$]\label{th:pert-G-sp}
Consider the Gaussian setting $V=b(V_0)$, where $V_0$ is a stationary Gaussian field, and assume that $V_0$ is nondegenerate and that $\expec{|V|^p}<\infty$ holds for some $p>\frac d2$.
Then $H_{k,\lambda}^\st$ is essentially self-adjoint on~$H^2\cap\Ld^\infty(\Omega)$ and
\[\sigma(H_{k,\lambda}^\st)=[-|k|^2+\lambda\infess b,\,\infty).\qedhere\]
\end{prop1}

The nature of the spectrum of the perturbed operators $\{H_{k,\lambda}^\st\}_k$ is a more involved question and is a main concern in the sequel.
In view of the fibration~\eqref{eq:flow-decomp0}, the perturbation of the eigenvalue at $0$ for the fibered operators $\{H_{k,0}^\st\}_k$ is of particular interest.
According to Fermi's Golden Rule, e.g.~\cite[Section~XII.6]{Reed-Simon-78},
this eigenvalue embedded in continuous spectrum is expected to dissolve when the perturbation is turned on.
The simplest rigorous version of this key conjecture is as follows.
It is based on observing that the formula for the second derivative of a hypothetic branch of eigenvalues at $\lambda=0$ would be a complex number, cf.~Section~\ref{sec:instab}.

\begin{prop1}[Instability of the bound state]\label{prop:Fermi}
Let $k\in\R^d\setminus\{0\}$, let $V$ be a stationary random field, denote by $\widehat\Cc$ the (nonnegative measure) Fourier transform of its covariance function,
and assume that $\widehat\Cc$ does not vanish identically on the sphere $\partial B_{|k|}(-k)$, in the sense that $\lim_{\e\downarrow0}\frac1\e\widehat\Cc\big(B_{|k|+\e}(-k)\setminus B_{|k|-\e}(-k)\big)\,>\,0$.
Then there exists no $C^2$ branch
\[[0,\delta)\to\R\times\Ld^2(\Omega):\lambda\mapsto(E_{k,\lambda},\psi_{k,\lambda})\]
with
\[H_{k,\lambda}^\st\psi_{k,\lambda}=E_{k,\lambda}\psi_{k,\lambda},\qquad(E_{k,\lambda},\psi_{k,\lambda})|_{\lambda=0}=(0,1).\qedhere\]
\end{prop1}

This basic instability result is however quite weak: for $k\in\R^d\setminus\{0\}$, the operator~$H_{k,\lambda}^\st$ is in fact expected to have purely absolutely continuous spectrum in a neighborhood of~$0$ for $0<\lambda\ll1$.
In addition, in view of the resonance interpretation of Fermi's Golden Rule, which originates in the work of Weisskopf and Wigner~\cite{W-Wigner-30}, the perturbed eigenvalue is expected to turn into a complex resonance.
Relevant conjectures are formulated in Section~\ref{sec:conj} below.

\subsection{Perturbative Mourre's commutator approach}\label{sec:intro-Mourre}

The perturbation problem for an eigenvalue embedded in continuous spectrum, in link with Fermi's Golden Rule and resonances, is an active topic of research in spectral theory.
Various general approaches have been successfully developed, see e.g.~\cite{Zworski-book} and references therein, but none seems to be available in our probabilistic setting: a key difficulty is that the random perturbation~$V$ is never compact with respect to the unperturbed operators $\{H_{k,0}^\st\}_k$ on $\Ld^2(\Omega)$ (unless it is degenerate).
This calls for the development of robust techniques for the spectral analysis of stationary operators on the probability space.
In the present contribution, we appeal to Mourre's commutator theory~\cite{Mourre-80,ABMG-96}, cf.~Section~\ref{sec:Mourre-intro}, which is reputedly flexible and requires no compactness. Although not allowing to deduce the existence of resonances in any strong form, Mourre's theory would ensure similar dynamical consequences, e.g.~\cite{Orth-90,Hunziker-90,CGH-06}.

\medskip
More precisely, we start with the construction of a natural group of dilations~$\{U_t^\st\}_{t\in\R}$ on~$\Ld^2(\Omega)$ in the model Gaussian setting, cf.~Section~\ref{sec:oper-L2}:
heuristically, it amounts to dilating the underlying white noise in the representation~\eqref{eq:white-rep}, which constitutes a unitary group since dilations preserve the law of the white noise.
The generator $A^\st$ of this group is then checked to be a conjugate operator for the stationary Laplacian~$-\triangle^\st$ in the sense of Mourre's theory, cf.~Proposition~\ref{prop:commut}(i). In Section~\ref{sec:Mourre-fiber}, by means of suitable deformations, we further construct corresponding conjugates for the whole family of fibered operators~$\{H_{k,0}^\st\}_k$. This appears to be surprisingly more involved than for $k=0$, in link with the infinite dimensionality of the probability space: our proof makes a deep use of the Fock space structure of $\Ld^2(\Omega)$ as provided by Malliavin calculus, thus emphasizing the interplay between spectral theory and the functional structure of the probabilistic setup.
Next, we turn to perturbed operators~$\{H_{k,\lambda}^\st\}_k$. It appears that the perturbation~$\lambda V$ is not compatible in the sense of Mourre's theory, cf.~Proposition~\ref{prop:commut}(iii), again in link with the infinite dimensionality of the probability space.
Perturbative Mourre's theory can therefore not be applied
unless we introduce a suitable (weak) Wiener truncation.

\begin{theor1}[Perturbative Mourre's theory up to truncation]\label{th:Mourre-pert}
Let $V=V_0$ be a stationary Gaussian field with covariance function $\Cc_0\in C^\infty_c(\R^d)$,
let~$\Lc$ denote the Ornstein--Uhlenbeck operator for the associated Malliavin calculus, cf.~Section~\ref{chap:Mall}, and for a given constant~\mbox{$L_0>0$} consider the truncation $Q_\lambda:=\mathds1_{[0,(L_0\lambda)^{-2}]}(\Lc)$ onto Wiener chaoses of order $\le(L_0\lambda)^{-2}$.
Then, for any $k\in\R^d$, there exists a self-adjoint operator $C_k^\st$ on~$\Ld^2(\Omega)$ and an explicit core~$\Pc(\Omega)$, cf.~\eqref{eq:def-Pc}, such that the following properties hold:
\begin{enumerate}[(i)]
\item
For all $\e>0$, the truncated operator $Q_\lambda H_{k,0}^\st Q_\lambda$ satisfies a Mourre relation on the interval $J_\e:=[\e-\frac34|k|^2,\infty)$ with respect to $C_k^\st$.
More precisely, its domain is invariant under the unitary group generated by $C_k^\st$,
and its commutator with $\frac1iC_k^\st$
is well-defined and essentially self-adjoint on $\Pc(\Omega)$, is $Q_\lambda H_{k,0}^\st Q_\lambda$-bounded, and satisfies the following lower bound,
\[\qquad\mathds1_{J_\e}(Q_\lambda H_{k,0}^\st Q_\lambda)\,[Q_\lambda H_{k,0}^\st Q_\lambda,\tfrac1iC_k^\st]\,\mathds1_{J_\e}(Q_\lambda H_{k,0}^\st Q_\lambda)\,\ge\, \e \mathds1_{J_\e}(Q_\lambda H_{k,0}^\st Q_\lambda) -\tfrac34|k|^2\E.\]
\item The truncated perturbation $Q_\lambda \lambda V Q_\lambda$ is compatible with respect to $C_k^\st$ in the sense that its iterated commutators with $\frac1iC_k^\st$ are well-defined on~$\Pc(\Omega)$ and bounded by~$O(L_0^{-1})$.
\end{enumerate}
In particular, the truncated perturbed operator $Q_\lambda H_{k,\lambda}^\st Q_\lambda$ satisfies a corresponding Mourre relation on $[\e+CL_0^{-1}-\frac34|k|^2,\infty)$ with respect to $C_k^\st$.
\end{theor1}

\smallskip
Based on this perturbative Mourre result, an approximate dynamical resonance analysis can be developed for truncated fibered operators in the spirit of~\cite{Orth-90,Hunziker-90,CGH-06} and leads to the exponential time decay of the corresponding averaged wavefunction. Further noting that the truncation error is easily estimated on the kinetic timescale,
we can get rid of the truncation and rigorously deduce the validity of the exponential decay law~\eqref{eq:exp-decay} as stated below; the proof is postponed to Section~\ref{sec:consequ-Mourre}.

\begin{cor1}[Exponential decay law on kinetic timescale]\label{th:mourre-partial}
Let $u^\circ\in\Ld^2(\R^d)$ have compactly supported Fourier transform,
let $V=V_0$ be a stationary Gaussian field with covariance function {$\Cc_0\in C^\infty_c(\R^d)$}, and define $\alpha_k,\beta_k\in\R$ by
\[\alpha_k+i\beta_k:=\lim_{\e\downarrow0}\tfrac1i\expecm{V(H_{k,0}^\st-i\e)^{-1}V},\]
that is, more explicitly,
\begin{eqnarray}\label{eq:def-alphbet}
\alpha_k&:=&
\pi\int_{\R^d}\widehat\Cc_0(y-k)\,\delta(|y|^2-|k|^2)\,\dbar y\,=\,
\frac{\pi}{2(2\pi)^{d}|k|}\int_{\partial B_{|k|}(-k)}\widehat\Cc_0~~>0,\\
\beta_k&:=&-(2\pi)^{-d}\,\pv\int_{-|k|^2}^\infty\frac1{r}\bigg(\frac1{2\sqrt{|k|^2+r}}\int_{\partial B_{\sqrt{|k|^2+r}}(-k)}\widehat\Cc_0\bigg)\,dr,\nonumber
\end{eqnarray}
Then there exists $s_0\simeq1$ such that the Schrödinger flow satisfies for all $0\le s\le s_0$,
\[\lim_{\lambda\downarrow0}\,\sup_{x\in\R^d}\bigg|\expecm{u_\lambda^{\lambda^{-2}s}(x)}\,-\,\int_{\R^d}\widehat u^\circ(k)\,e^{ik\cdot x-i\lambda^{-2}s|k|^2}e^{-s(\alpha_k+i\beta_k)}\,\dbar k\bigg|\,=\,0.\qedhere\]
\end{cor1}

Although our truncation argument could be compared with the truncation of the Dyson series in the perturbative analysis of~\cite{Spohn-77,Erdos-Yau-00,Erdos-Salmhofer-Yau-1,Erdos-Salmhofer-Yau-2}, it only requires to estimate a truncation error, which is often a simpler matter, while the truncated evolution is intrinsically analyzed by means of Mourre's theory, avoiding any Feynman diagram analysis or any renormalization to handle the truncated Dyson series.
In addition, formal computations indicate that the truncation of the evolution at time~$t$ on Wiener chaoses of order $\le K$ should be accurate provided $K\gg\lambda^2t$.
Since our truncation $Q_\lambda$ in Theorem~\ref{th:Mourre-pert} amounts to projecting onto chaoses of order $\le(L_0\lambda)^{-2}$, which is a particularly high order compared to truncations of Dyson series in~\cite{Spohn-77,Erdos-Yau-00,Erdos-Salmhofer-Yau-1,Erdos-Salmhofer-Yau-2}, the accuracy in Corollary~\ref{th:mourre-partial} should thus follow in fact up to times $t\ll\lambda^{-4}$.
Non-perturbative approaches to fibered resonances and accuracy on even longer timescales are postponed to future works.

\subsection{Exact resonance conjectures}\label{sec:conj}
The above results provide partial indications that the eigenvalue at $0$ of the fibered operators $\{H_{k,0}^\st\}_k$ should dissolve in the continuous spectrum upon perturbation and turn into complex resonances.
In particular, in agreement with Fermi's Golden Rule, Corollary~\ref{th:mourre-partial} is consistent with resonances at
\[z_{k,\lambda}=\lambda^2(\beta_k-i\alpha_k)+o_k(\lambda^2),\qquad k\in\R^d\setminus\{0\}.\]
As resonance theories have never been constructed for operators on the probability space, we formulate relevant conjectures that will be investigated rigorously in future works.
To emphasize the relevance of our formulation, we further consider in Section~\ref{sec:toy} an illustrative toy model that shares many spectral features of Schr\"{o}dinger operators and for which resonances are explicitly shown to exist in a similar sense, cf.~Theorem~\ref{th:toy}(iii).

\medskip
According to the usual definition, the operator $H_{k,\lambda}$ has a resonance at $z_{k,\lambda}$ in the lower complex half-plane if the resolvent $(H_{k,\lambda}^\st-z)^{-1}$ on the upper half-plane $\Im z>0$, when viewed in a suitably weakened topology, extends to a meromorphic family of operators indexed by all $z\in\C$ (or at least in a complex neighborhood), and if this family admits a simple pole at~$z=z_{k,\lambda}$.
In the usual case of operators on $\Ld^2(\R^d)$, the suitable weakening of the topology typically consists of viewing the resolvent as a family of linear operators~$C^\infty_c(\R^d)\to \mathcal D'(\R^d)$ rather than $\Ld^2(\R^d)\to\Ld^2(\R^d)$.
In the present setting on the probability space, the role of $C^\infty_c(\R^d)\subset\Ld^2(\R^d)$ can be played for instance by
the dense linear subspace~$\Pc(\Omega)\subset\Ld^2(\Omega)$ of $V$-polynomials,
\begin{equation}\label{eq:def-Pc}
\Pc(\Omega)\,:=\,\bigg\{\sum_{j=1}^na_j\prod_{l=1}^{m_j}V(x_{lj})\,:\,n\ge1,\,a_j\in\C,\,m_j\ge0,\,x_{lj}\in\R^d\bigg\},
\end{equation}
and the dual $\Dc'(\R^d)$ is then replaced by the dual space $\Pc'(\Omega)$ of continuous linear functionals on~$\Pc(\Omega)$.
In these terms, we formulate the following resonance conjecture.
The linear functionals $\Psi_{k,\lambda}^+$ and $\Psi_{k,\lambda}^-$ below are referred to as the resonant and co-resonant states, respectively.
Since the imaginary part of the expected branch of resonances $\Im z_{k,\lambda}=-\lambda^2\alpha_k+O_k(\lambda^3)$ vanishes to leading order both as $k\to0$ (in dimension $d>2$) and as $|k|\uparrow\infty$, cf.~formula~\eqref{eq:def-alphbet}, we henceforth restrict to $k$ in a compact set away from~$0$.

\begin{conjn}{LRC}{Local resonance conjecture}\label{C1}$ $\\
Given a compact set $K\subset\R^d\setminus\{0\}$, there are $\lambda_0,M>0$ such that for all $k\in K$ and $0\le\lambda<\lambda_0$ the resolvent $z\mapsto(H_{k,\lambda}^\st-z)^{-1}$ defined on $\Im z>0$ as a family of operators $\Pc(\Omega)\to\Pc'(\Omega)$ can be extended meromorphically to the whole complex ball $|z|\le\frac1M$ with a unique simple pole. In other words, there exist continuous collections $\{z_{k,\lambda}\}_{k,\lambda}\subset\C$ and $\{\Psi_{k,\lambda}^+\}_{k,\lambda},\{\Psi_{k,\lambda}^-\}_{k,\lambda}\subset\Pc'(\Omega)$ such that for all $\phi,\phi'\in\Pc(\Omega)$ we can write for $\Im z>0$,
\begin{equation}\label{eq:dec-res}
\big\langle\phi',(H_{k,\lambda}^\st-z)^{-1}\phi\big\rangle_{\Ld^2(\Omega)}\,=\,\frac{\overline{\langle\Psi^+_{k,\lambda},\phi'\rangle_{\Pc'(\Omega),\Pc(\Omega)}}\,\langle\Psi_{k,\lambda}^-,\phi\rangle_{\Pc'(\Omega),\Pc(\Omega)}}{z_{k,\lambda}-z}+\zeta_{k,\lambda}^{\phi',\phi}(z),
\end{equation}
where the remainder~$\zeta_{k,\lambda}^{\phi',\phi}$ is holomorphic on the set $\{z:\Im z>0\}\bigcup\frac1MB$ and has continuous dependence on $k,\lambda$.
\end{conjn}

Next, we state a global version of this resonance conjecture in the case of an unbounded potential $V$ with $\sigma(H_{k,\lambda}^\st)=\R$ (see e.g.~Proposition~\ref{th:pert-G-sp}).
A direct computation shows that the spectral measure of $H_{k,0}^\st$ associated with~$V$ is typically supported on the whole half-axis $[-|k|^2,\infty)$ and is only $\frac{d-2}2$-times differentiable at~$-|k|^2$ in dimension $d>2$, cf.\@ proof of Lemma~\ref{th:spectrum-gen};
this suggests that the band on which the meromorphic extension of the resolvent exists must shrink as $\lambda\downarrow0$ close to $z=-|k|^2$.

\begin{conjn}{GRC}{Global resonance conjecture}\label{C2}$ $\\
The same decomposition~\eqref{eq:dec-res} holds with a remainder~$\zeta_{k,\lambda}^{\phi',\phi}$ that is holomorphic on the set $\{z:\Im z>-\frac1M\lambda^\rho\}\bigcup \frac1MB$ for some exponent $\rho<2$, has continuous dependence on $k,\lambda$, and satisfies a uniform bound of the form
\[\sup_{|\Im z|\le\frac1{2M}\lambda^\rho}|\zeta_{k,\lambda}^{\phi',\phi}(z)|\,\le\,\lambda^{-M}.\qedhere\]
\end{conjn}

\begin{rem}[Continuous resonant spectrum]\label{rem:res-spec}
When integrated along the Floquet--Bloch fibration~\eqref{eq:fibration}, in dimension $d>2$, the conjectured fibered resonances would yield a hammock-shaped set in the lower complex half-plane, connecting some point~$O(\lambda^2)$ on the real axis to $+\infty$,
\[\qquad\Sigma_\lambda:=\{|k|^2+z_{k,\lambda}:k\in\R^d\}\approx\{|k|^2+\lambda^2(\beta_k-i\alpha_k):k\in\R^d\}.\]
This set is increasingly thinner at infinity and can reach a thickness $O(\lambda^2)$ in the middle, but it reduces to a curve for instance when the covariance is radial.
This set can be viewed as a kind of ``continuous resonant spectrum'' for the Schr\"{o}dinger operator $H_\lambda=-\triangle+\lambda V$ on $\Ld^2(\R^d\times\Omega)$.
While to the best of our knowledge such a notion has never been introduced in the literature, it is made rigorous for the illustrative toy model that we introduce in Section~\ref{sec:toy}, cf.~Theorem~\ref{th:toy}(iv).
\end{rem}

We show that the above conjectures imply the expected exponential decay law~\eqref{eq:exp-decay} for the averaged wavefunction to finer accuracy, thus providing a strong improvement and a very first workaround for the available perturbative methods~\cite{Spohn-77,Erdos-Yau-00,Erdos-Salmhofer-Yau-1,Erdos-Salmhofer-Yau-2,Chen-Komorowski-Ryzhik-16}.
Under Conjecture~\ref{C1} an accurate description of the decay law is deduced only for times $t\ll\lambda^{-2}|\!\log\lambda|$, but accuracy is reached on all timescales under Conjecture~\ref{C2}.
The result is expressed as a resonant-mode expansion of the Schrödinger flow in the weak sense of $\Pc'(\Omega)$, and the description of the averaged wavefunction follows as a particular case; the proof is quite standard and is postponed to Section~\ref{sec:expdec-conj-pr}.

\begin{cor1}[Consequences of resonance conjectures]\label{th:res-spectr}
Let $u^\circ\in\Ld^2(\R^d)$ have Fourier transform supported in the compact set $K\subset \R^d\setminus\{0\}$, with $\|u^\circ\|_{\Ld^2(\R^d)}=1$.
\begin{enumerate}[(i)]
\item Under Conjecture~\emph{\ref{C1}}, there holds in $\Ld^\infty(\R^d;\Pc'(\Omega))$, for all $0\le\lambda<\lambda_0$,
\begin{equation}\label{eq:expand-LRC}
\qquad u_\lambda^t(x)\,=\,\int_{\R^d}\widehat u^\circ(k)\,e^{ik\cdot x-it(|k|^2+z_{k,\lambda})}\,\langle\Psi_{k,\lambda}^-,1\rangle_{\Pc'(\Omega),\Pc(\Omega)}\,\Psi_{k,\lambda}^{+;x}\,\dbar k\,+\,O_{K,M}(\lambda),
\end{equation}
where we have set $\langle\Psi_{k,\lambda}^{+;x},\phi\rangle_{\Pc'(\Omega),\Pc(\Omega)}:=\langle\Psi^+_{k,\lambda},\phi(-x,\cdot)\rangle_{\Pc'(\Omega),\Pc(\Omega)}$ for $\phi\in \Pc(\Omega)$.
More precisely, this means for all $\phi\in \Pc(\Omega)$,
\begin{multline*}
\qquad\sup_x\bigg|\expecm{\bar\phi\, u_\lambda^t(x)}
\,-\,\int_{\R^d}\widehat u^\circ(k)\,e^{ik\cdot x-it(|k|^2+z_{k,\lambda})}\overline{\langle\Psi^{+;x}_{k,\lambda},\phi\rangle_{\Pc'(\Omega),\Pc(\Omega)}}\langle\Psi_{k,\lambda}^-,1\rangle_{\Pc'(\Omega),\Pc(\Omega)}\,\dbar k\bigg|\,\\\lesssim_{\phi,K,M}\,\lambda\|u^\circ\|_{\Ld^2(\R^d)}.
\end{multline*}
In addition, the averaged wavefunction satisfies the following improved estimate,
\begin{equation*}
\qquad \sup_x\bigg|\expec{u_\lambda^t(x)}\,-\,\int_{\R^d}\widehat u^\circ(k)\,e^{ik\cdot x-it(|k|^2+z_{k,\lambda})}\,\dbar k\bigg|\,\lesssim_{K,M}\,\lambda^2.
\end{equation*}
\item Under Conjecture~\emph{\ref{C2}}, the same holds as in~(i) with the errors $O(\lambda)$ and~$O(\lambda^2)$ improved into $O(\lambda e^{-\frac t{8M}\lambda^\rho})$ and $O(\lambda^2 e^{-\frac t{8M}\lambda^\rho})$, respectively.
\end{enumerate}
Moreover, for $k\in K$ and $0<\lambda<\lambda_0$, the restriction of the spectrum of $H_{k,\lambda}^\st$ to $(-\frac1M,\frac1M)$ is absolutely continuous under Conjecture~\ref{C1}, and the whole spectrum is absolutely continuous under Conjecture~\ref{C2}.
\end{cor1}

In order to make the above resonant-mode expansion~\eqref{eq:expand-LRC} more striking, we note that resonances and resonant states can be computed explicitly in form of a perturbative Rayleigh--Schrödinger series.
In particular, in agreement with Corollary~\ref{th:mourre-partial}, the resonance is checked to coincide to leading order with $\lambda^2(\beta_k-i\alpha_k)$; the proof is included in Section~\ref{sec:RS}.

\begin{prop1}[Approximate computation of resonances]\label{prop:pert-res}
If Conjecture~\emph{\ref{C1}} holds and if for all $k\in K$ and $\phi,\phi'\in \Pc(\Omega)$ the map
\begin{equation}\label{eq:map-reson}
[0,\lambda_0)\to\C\times\Pc'(\Omega)\times\Pc'(\Omega)\times\Ld^\infty_\loc(\tfrac1MB):\lambda\mapsto(z_{k,\lambda},\Psi_{k,\lambda}^+,\Psi_{k,\lambda}^-,\zeta_{k,\lambda}^{\phi',\phi})
\end{equation}
is of class $C^2$, then up to a gauge transformation
there hold as $\lambda\downarrow0$, for all $k\in K$,
\begin{eqnarray*}
z_{k,\lambda}&=&\lambda^2(\beta_k-i\alpha_k)+o_k(\lambda^2),\\
\Psi_{k,\lambda}^\pm&=&1+\lambda\Phi^{1,\pm}_k+\lambda^2\Phi^{2,\pm}_k+o_k(\lambda^2),
\end{eqnarray*}
where $\alpha_k,\beta_k$ are defined in~\eqref{eq:def-alphbet} and where $\Phi^{1,\pm}_k$ and $\Phi^{2,\pm}_k$ are given by
\begin{eqnarray*}
\Phi_k^{1,\pm}&:=&-(H_{k,0}^\st\mp i0)^{-1}V,\\
\Phi_k^{2,\pm}&:=&(H_{k,0}^\st\mp i0)^{-1}\Pi V(H_{k,0}^\st\mp i0)^{-1}V,
\end{eqnarray*}
in terms of the projection $\Pi:=\Id-\E$ onto $\Ld^2(\Omega)\ominus\C$. More precisely, the latter formulas are understood as follows, for all $\phi\in\Pc(\Omega)$,
\begin{align*}
\langle\Phi_k^{1,\pm},\phi\rangle_{\Pc'(\Omega),\Pc(\Omega)}&~:=~-\lim_{\e\downarrow0}\expec{V(H_{k,0}^\st\pm i\e)^{-1}\phi},\\
\langle\Phi_k^{2,\pm},\phi\rangle_{\Pc'(\Omega),\Pc(\Omega)}&~:=~\lim_{\e\downarrow0}\expec{V(H_{k,0}^\st\pm i\e)^{-1}V(H_{k,0}^\st\pm i\e)^{-1}\Pi\phi}.\qedhere
\end{align*}
\end{prop1}

\begin{rem}[Full Rayleigh--Schrödinger series for resonances]\label{rem:reson}
The proof of the above is easily pursued to any order. For $n\ge1$ and $\phi\in \Pc(\Omega)$, if the map~\eqref{eq:map-reson} is of class~$C^n$, then there hold as $\lambda\downarrow0$, for all $k\in K$,
\vspace{-0.1cm}
\begin{equation}\label{eq:RS-series}
z_{k,\lambda}=\sum_{m=1}^{n-1}\lambda^{m+1}\nu_k^m+o_k(\lambda^n),\qquad\Psi_{k,\lambda}^\pm=\sum_{m=0}^n\lambda^m\Phi_k^{m,\pm}+o_k(\lambda^n),
\end{equation}
where the coefficients are explicitly defined and can be checked to coincide with those of the formal Rayleigh--Schrödinger series for the perturbation of a bound state. This asymptotic series makes no sense in $\Ld^2(\Omega)$ (in link with the dissolution of the bound state), but can be constructed in the weak sense of $\Pc'(\Omega)$; this partially answers in our setting a question raised in~\cite[p.179]{Hunziker-90}.
More precisely, for all $\phi\in \Pc(\Omega)$, we can write
\begin{gather*}
\nu_k^m:=\langle\Phi_k^{m,-},V\rangle_{\Pc'(\Omega),\Pc(\Omega)}=\lim_{\e\downarrow0}\expec{V\overline{\phi_k^{m,\e,-}}},\\
\langle\Phi_k^{m,\pm},\phi\rangle_{\Pc'(\Omega),\Pc(\Omega)}:=\lim_{\e\downarrow0}\expec{\phi\,\overline{\phi_k^{m,\e,\pm}}},
\end{gather*}
where the limits indeed exist and
where for all $\e>0$ the sequence $(\phi_k^{m,\e,\pm})_m\subset\Ld^2(\Omega)$ is defined iteratively as follows: we set $\phi_{k}^{0,\e,\pm}=1$ and for all $m\ge0$ we define $\phi_k^{m+1,\e,\pm}$ as the unique solution of the regularized Rayleigh--Schrödinger recurrence equation,
\[\qquad\big(H_{k,0}^\st\mp i\e\big)\phi_{k}^{m+1,\e,\pm}=-V\phi_{k}^{m,\e,\pm}+\sum_{l=0}^m\expec{V\phi_{k}^{l,\e,\pm}}\phi_{k}^{m-l,\e,\pm}.\]
The Rayleigh--Schrödinger series~\eqref{eq:RS-series} is not known to be summable, hence cannot be used to actually {\it construct} resonances, which constitutes a reputed difficulty in this problem; see also~\cite{Spohn-77,Erdos-Yau-00}.
\end{rem}

%%%%%%%%%%%%%%%%%%%%%
%%%%%%%%%%%%%%%%%%%%%
%%%%%%%%%%%%%%%%%%%%%

\medskip
\section{Stationary random setting and Floquet--Bloch fibration}\label{sec:Floquet--Bloch}

In this section, we give a suitable definition of stationarity (or statistical translation-invariance) and we define the associated stationary differential calculus on the probability space, which was first introduced in~\cite{PapaVara} and plays a key role in the context of stochastic homogenization theory, e.g.~\cite[Section~7]{JKO94}.
Next, we generalize the periodic Floquet--Bloch theory to this stationary setting, establishing in particular~\eqref{eq:fibration} and~\eqref{eq:flow-decomp0}.

\subsection{Stationary setting}\label{sec:stat-setting}
Given a reference (complete) probability space $(\Omega,\p)$, we start by recalling the classical notion of stationarity.
In particular, a Gaussian field $V_0$, that is, a family $V_0=\{V_0(x,\cdot)\}_{x\in\R^d}$ of Gaussian random variables, is an example of a stationary measurable random field if the variables $\{V_0(x,\cdot)\}_{x\in\R^d}$ have the same expectation and have covariance $K_0(x,y):=\cov{V_0(x,\cdot)}{V_0(y,\cdot)}$ of the form $K_0(x,y)=\Cc_0(x-y)$ with $\Cc_0$ continuous at the origin.

\begin{defin}\label{defin:stat1}
A {\it random field} on $\R^d$ is a map $\phi:\R^d\times\Omega\to\C$ such that for all $x\in\R^d$ the random variable $\phi(x,\cdot):\Omega\to\C$ is measurable.
It is said to be {\it stationary} if its finite-dimensional law is shift-invariant, that is, if for any finite set $E\subset\R^d$ the law of $\{\phi(x+y,\cdot)\}_{x\in E}$ does not depend on the shift $y\in\R^d$.
In addition, it is said to be {\it measurable} if the map $\phi:\R^d\times\Omega\to\R$ is jointly measurable. (In view of a result due to von Neumann~\cite{Neumann-32}, which can be viewed as a stochastic version of Lusin's theorem, joint measurability is equivalent to requiring that for almost all $x$ and for all $\delta>0$ there holds $\p\{\omega\in\Omega:|\phi(x+y,\omega)-\phi(x,\omega)|>\delta\}\to0$ as $y\to0$.)
\end{defin}

This basic notion of stationarity is usefully reformulated in terms of a measure-preserving action on the probability space, which draws the link with the theory of dynamical systems and ergodic theory.

\begin{defin}\label{defin:stat2}
A {\it measurable action} $\tau:=(\tau_x)_{x\in\R^d}$ of the group $(\R^d,+)$ on~$(\Omega,\p)$ is a collection of measurable maps $\tau_x:\Omega\to\Omega$ that satisfy
\begin{itemize}
\item $\tau_x\circ\tau_y=\tau_{x+y}$ for all $x,y\in\R^d$;
\item $\pr{\tau_xA}=\pr{A}$ for all $x\in\R^d$ and measurable $A\subset\Omega$;
\item the map $\R^d\times\Omega\to\Omega:(x,\omega)\mapsto\tau_x\omega$ is jointly measurable.
\end{itemize}
A random field $\phi:\R^d\times\Omega\to\C$ is said to be {\it $\tau$-stationary} if there exists a measurable map $\phi_\circ:\Omega\to\C$ such that $\phi(x,\omega)=\phi_\circ(\tau_{-x}\omega)$ for all $x,\omega$.
\end{defin}

This second definition yields a bijection between random variables $\phi_\circ:\Omega\to\C$ and $\tau$-stationary random fields $\phi:\R^d\times\Omega\to\C$.
The random field $\phi$ is referred to as the {\it $\tau$-stationary extension} of~$\phi_\circ$.
In addition, given $\phi_\circ\in\Ld^p(\Omega)$ with $p\ge1$, since there holds $\expec{\int_K|\phi|^p}=|K|\,\expec{|\phi_\circ|^p}$ for any compact $K\subset\R^d$, we deduce that the realization $\phi(\cdot,\omega)$ belongs to $\Ld^p_\loc(\R^d)$ for almost all~$\omega$. The Banach space $\Ld^p(\Omega)$ can thus be identified with the subspace of $\tau$-stationary random fields in $\Ld^p_\loc(\R^d;\Ld^p(\Omega))$.

\medskip
While the notion of $\tau$-stationarity in the sense of Definition~\ref{defin:stat2} obviously implies measurability and stationarity in the sense of Definition~\ref{defin:stat1}, the following asserts that both are in fact essentially equivalent.

\begin{lem}
Let $\phi$ be a stationary measurable random field defined on~$(\Omega,\p)$ in the sense of Definition~\ref{defin:stat1}. Then there exist a probability space $(\Omega',\p')$, endowed with a measurable action $\tau$, and a $\tau$-stationary random field $\phi'$ defined on $(\Omega',\p')$ in the sense of Definition~\ref{defin:stat2} such that $\phi$ and $\phi'$ have the same finite-dimensional law.
This extends to a correspondence between $\sigma(\phi)$-measurable random variables on $\Omega$ and random variables on $\Omega'$.
\end{lem}

\begin{proof}
The proof is a variant of e.g.~\cite[Section~16.1]{KS-07}.
Let $\Omega'$ denote the set of measurable functions $\R^d\to\C$, endowed with the cylindrical $\sigma$-algebra $\F'$, and consider the map $H:\Omega\to\Omega':\omega\mapsto \phi(\cdot,\omega)$. This map is measurable and induces a probability measure $\p':=H_*\p$ on the measurable space $(\Omega',\F')$. Next, define $\tau_x:\Omega'\to\Omega'$ by $(\tau_x\omega')(y):=\omega'(y-x)$. As $\phi$ is jointly measurable and stationary, we find that $\tau$ is a measurable action. Finally, we set $\phi'_\circ(\omega'):=\omega'(0)$, with $\tau$-stationary extension $\phi'(x,\omega'):=\omega'(x)$, and the claim follows. We omit the details.
\end{proof}

Henceforth, we focus on the more convenient notion of $\tau$-stationarity in the sense of Definition~\ref{defin:stat2}: we implicitly assume that the reference probability space $(\Omega,\p)$ is endowed with a given measurable action $\tau$ and we assume that the random potential $V$ is $\tau$-stationary.
In the sequel, for abbreviation, $\tau$-stationarity is simply referred to as stationarity, and we abusively use the same notation for $\phi$ and~$\phi_\circ$ (in particular, for $V$ and $V_\circ$).

\begin{rems}\label{rem:stat}
\mbox{}
\begin{enumerate}[(a)]
\item
A standard construction~\cite{PapaVara} allows to view periodic and quasiperiodic functions (as well as almost periodic functions) as instances of stationary random fields (with correlations that do not decay at infinity). In the periodic setting, the probability space $(\Omega,\p)$ is chosen as the torus $\T^d$ endowed with the Lebesgue measure, the action $\tau$ is given by $\tau_{-x}\omega=\omega+x$ on $\T^d$, and we set $\phi(x,\omega)=\phi_\circ(\omega+x)$.
In the quasiperiodic setting, the probability space is chosen as a higher-dimensional torus $\T^M$ with $M>d$, endowed with the Lebesgue measure, the action $\tau$ is given by $\tau_{-x}\omega=\omega+Fx$ on $\T^M$ in terms of the winding matrix $F\in\R^{M\times d}$, and we set $\phi(x,\omega)=\phi_\circ(\omega+Fx)$. In both cases, the construction is viewed as introducing a uniform random shift.
\item
Any $\Z^d$-stationary random potential (that is, satisfying the stationarity assumption for an action of $(\Z^d,+)$ on $\Omega$) can also be seen as a stationary random potential in the above sense up to considering the random ensemble of shifts. Indeed, assume that $\tau':=(\tau_z')_{z\in\Z^d}$ is a measurable action of $(\Z^d,+)$ on a probability space $(\Omega',\p')$, and that $\phi$ is $\tau'$-stationary, that is, $\phi(x+z,\omega)=\phi(x,\tau'_{-z}\omega)$ for all $x\in\R^d$, $z\in\Z^d$, and $\omega\in\Omega'$. Endow $\Omega:=\Omega'\times [0,1)^d$ with the product measure $\p:=\p'\times\Leb$, where $\Leb$ denotes the Lebesgue measure on $[0,1)^d$, and define the action $\tau:=(\tau_x)_{x\in\R^d}$ of $(\R^d,+)$ on $\Omega=\Omega'\times [0,1)^d$ by
\[\tau_x(\omega,y):=\big(\tau'_{\lfloor x\rfloor}\omega~,~ y+x-\lfloor y+x\rfloor\big),\]
where $\lfloor x\rfloor=(\lfloor x_1\rfloor,\ldots,\lfloor x_d\rfloor)$ for $x\in\R^d$ and where $\lfloor a\rfloor$ denotes the largest integer $\le a$ for $a\in\R$.
The map $\psi(x,(\omega,y)):=\phi(x-y,\omega)$ then defines a $\tau$-stationary random field on $\R^d\times\Omega$.
\qedhere
\end{enumerate}
\end{rems}

\subsection{Stationary differential calculus}\label{sec:stat-calc}
A differential calculus is naturally developed on~$\Ld^2(\Omega)$ via the measurable action $\tau$ on $(\Omega,\p)$.
Indeed, while the subspace of stationary random fields in $\Ld^2_\loc(\R^d;\Ld^2(\Omega))$ is identified with the Hilbert space $\Ld^2(\Omega)$, the spatial weak gradient $\nabla$ on locally square integrable functions turns into a densely defined linear operator $\nabla^\st$ on~$\Ld^2(\Omega)$, which is referred to as the {\it stationary gradient}.
Equivalently, $\nabla^\st$ can be viewed as the infinitesimal generator of the group of isometries $\{T_x:\phi_\circ\mapsto \phi_\circ(\tau_{-x}\cdot)\}_{x\in\R^d}$ on~$\Ld^2(\Omega)$.
The adjoint is $(\nabla^\st)^*=-\nabla^\st$ and we denote by $-\triangle^\st=-\nabla^\st\cdot\nabla^\st$ the corresponding {\it stationary Laplacian}.
For all $s\ge0$, we define the (Hilbert) space $H^s(\Omega)$ as the space of all elements $\phi_\circ\in\Ld^2(\Omega)$ for which the stationary extension $\phi$ belongs to $H^s_\loc(\R^d;\Ld^2(\Omega))$, and we denote by $H^{-s}(\Omega)$ the dual of $H^s(\Omega)$. Note that $H^1(\Omega)$ coincides with the domain of~$\nabla^\st$, and that the stationary Laplacian $-\triangle^\st$ is self-adjoint on $H^2(\Omega)$.
We refer e.g. to~\cite[Section~7]{JKO94} for details.

\medskip
As opposed to the case of the periodic Laplacian on the torus, the stationary Laplacian~$-\triangle^\st$ on $\Ld^2(\Omega)$ typically has absolutely continuous spectrum and no spectral gap above $0$, cf.~Section~\ref{sec:spectrum-Hk0}. This entails that Poincaré's inequality does not hold on $H^1(\Omega)$ and that compact embeddings such as Rellich's theorem also fail. This lack of compactness is related to the fact that the gradient $\nabla^\st$ only contains information on a finite set of directions while $\Omega$ is typically an infinite product space.

\subsection{Stationary Floquet transform}\label{sec:stat-BF}
The usual periodic Floquet transform, e.g.~\cite{Kuchment-16}, is a reformulation of Fourier series: given a function $u\in\Ld^2(\R^d)$, its Floquet transform is (formally) defined by
\[\mathcal V^\circ_\per u(k,x)\,:=\,\sum_{n\in\Z^d}e^{-ik\cdot(x+n)}u(x+n),\]
which is periodic in $x$, so that the Fourier inversion formula takes the form
\[u(x)\,=\,\int_{2\pi\T^d}e^{ik\cdot x}\,\mathcal V^\circ_\per u(k,x)\,\dbar k,\]
thus leading to the following direct integral decomposition, e.g.~\cite[p.280]{Reed-Simon-78},
\[\Ld^2(\R^d)\,=\,\int_{2\pi\T^d}^\oplus\Ld^2(\T^d)\,\mathfrak e_k\,\dbar k,\qquad\mathfrak e_k(x):=e^{ik\cdot x}.\]
This decomposition allows for a simple adaptation to the augmented space $\Ld^2(\R^d\times\T^d)$:
given $u\in\Ld^2(\R^d\times\T^d)$, its Floquet transform is defined by
\[\mathcal V_\per u(k,q)\,:=\,\int_{\R^d}e^{-ik\cdot y}u(y,q-y)\,dy,\]
which is periodic in $q$, so that the Fourier inversion formula takes the form
\[u(x,q)\,=\,\int_{\R^d}e^{ik\cdot x}\,\mathcal V_\per u(k,x+q)\,\dbar k,\]
and leads to the direct integral decomposition
\[\Ld^2(\R^d\times\T^d)=\int_{\R^d}^\oplus\Ld^2(\T^d)\,\mathfrak e_k\,\dbar k.\]
We may now mimick this construction in the general stationary random setting: given $u\in\Ld^2(\R^d\times\Omega)$, its stationary Floquet transform is defined by
\[\mathcal V_\st u(k,\omega)\,:=\,\int_{\R^d}e^{-ik\cdot y}u(y,\tau_{y}\omega)\,dy,\]
so that the Fourier inversion formula takes the form
\[u(x,\omega)\,=\,\int_{\R^d}e^{ik\cdot x}\,\mathcal V_\st u(k,\tau_{-x}\omega)\,\dbar k,\]
and leads to the direct integral decomposition
\begin{equation}\label{eq:fibration-re}
\Ld^2(\R^d\times\Omega)=\int_{\R^d}^\oplus\Ld^2(\Omega)\,\mathfrak e_k\,\dbar k.
\end{equation}
This stationary Floquet transform was first introduced in~\cite[Section~3.2]{Schenker-Kang}; see also~\cite{BG-16,DGS1}. Some key properties are collected in the following.

\begin{samepage}
\begin{lem}\label{lem:floquet}
The {stationary Floquet transform} $\mathcal V_\st$ is a unitary operator on $\Ld^2(\R^d\times\Omega)$, and satisfies
\begin{enumerate}[(i)]
\item $\mathcal V_\st g=\widehat g$ for all $g\in\Ld^2(\R^d)\subset\Ld^2(\R^d\times\Omega)$;
\item $\mathcal V_\st(\phi u)=\phi_\circ\mathcal V_\st u$ for all $\phi_\circ\in\Ld^2(\Omega)$ and $u\in\Ld^2(\R^d\times\Omega)$ with $\phi u\in \Ld^2(\R^d\times\Omega)$.\qedhere
\end{enumerate}
\end{lem}
\end{samepage}

\subsection{Stationary Floquet--Bloch fibration}\label{sec:FB-fibr}
In view of~\eqref{eq:fibration-re},
the stationary Floquet transform $\mathcal V_\st$ decomposes differential operators with stationary random coefficients (such as the Schr\"odinger operator $H_\lambda$) into a direct integral of ``elementary'' fibered operators on the stationary space~$\Ld^2(\Omega)$.
First, the Laplacian $-\triangle$ on $\Ld^2(\R^d\times\Omega)$ is self-adjoint on $H^2(\R^d;\Ld^2(\Omega))$ and is mapped by $\mathcal V_\st$ on
\begin{equation}\label{eq:lap-k}
\mathcal V_\st[(-\triangle)u](k,\omega)\,=\,-(\nabla^\st+ik)\cdot(\nabla^\st+ik)\mathcal V_\st u(k,\omega)
\,=\,(H^\st_{k,0}+|k|^2)\mathcal V_\st u(k,\omega),
\end{equation}
in terms of the (centered) fibered Laplacian
\[H^\st_{k,0}\,:=\,-(\nabla^\st+ik)\cdot(\nabla^\st+ik)-|k|^2=-\triangle^\st-2ik\cdot\nabla^\st,\]
As the stationary Laplacian $-\triangle^\st$ is self-adjoint on $H^2(\Omega)$ and as $-2ik\cdot \nabla^\st$ is an infinitesimal perturbation, the Kato-Rellich theorem ensures that this fibered Laplacian $H_{k,0}^\st$ is also self-adjoint on $H^2(\Omega)$, and the centering ensures that constant functions belong to its kernel.
Next, if the stationary random potential~$V$ is uniformly bounded (the unbounded case is postponed to Appendix~\ref{sec:self-adj}), it defines a bounded multiplication operator on $\Ld^2(\R^d\times\Omega)$ and the corresponding Schr\"{o}dinger operator $H_{\lambda}=-\triangle+\lambda V$ is thus self-adjoint on $H^2(\R^d;\Ld^2(\Omega))$.
Combining~\eqref{eq:lap-k} with Lemma~\ref{lem:floquet}(ii), we find
\begin{equation}\label{e.dir-decomp}
\mathcal V_\st[H_\lambda f](k,\omega)=(H^\st_{k,\lambda}+|k|^2)\mathcal V_\st f(k,\omega),
\end{equation}
in terms of the (centered) fibered Schrödinger operator
\[H^\st_{k,\lambda}\,:=\,H_{k,0}^\st+\lambda V,\]
which is self-adjoint on $H^2(\Omega)$.
Using direct integral representation, e.g.~\cite[p.280]{Reed-Simon-78}, we may reformulate the above as
\begin{equation}\label{eq:fibration-re+}
\big(H_\lambda,\Ld^2(\R^d\times\Omega)\big)\,=\,\int_{\R^d}^{\oplus} \big(H_{k,\lambda}^\st+|k|^2,\Ld^2(\Omega)\big)\,\mathfrak e_k\,\dbar k.
\end{equation}
This decomposition of the Schr\"{o}dinger operator yields a stationary version of the so-called Bloch wave decomposition of the Schr\"{o}dinger flow: given a deterministic initial condition $u^\circ\in \Ld^2(\R^d)\subset\Ld^2(\R^d\times\Omega)$, appealing to~\eqref{e.dir-decomp} and to Lemma~\ref{lem:floquet}(i),
\begin{eqnarray*}
u_\lambda^t(x,\omega)\,=\,\big(e^{-itH_\lambda}u^\circ\big)(x,\omega)
&=&\int_{\R^d}e^{ik\cdot x}\,\mathcal V_\st\big[e^{-itH_\lambda}u^\circ\big](k,\tau_{-x}\omega)\,\dbar k\\
&=&\int_{\R^d}e^{ik\cdot x-it|k|^2}\,\big(e^{-itH_{k,\lambda}^\st}\mathcal V_\st u^\circ(k,\cdot)\big)(\tau_{-x}\omega)\,\dbar k\\
&=&\int_{\R^d}\widehat u^\circ(k)\,e^{ik\cdot x-it|k|^2}\,\big(e^{-itH_{k,\lambda}^\st}1\big)(\tau_{-x}\omega)\,\dbar k,
\end{eqnarray*}
that is,~\eqref{eq:flow-decomp0}. Alternatively,
in terms of the $\Ld^2(\Omega)$-valued spectral measure $\mu_{k,\lambda}^{1}$ of $H^\st_{k,\lambda}$ associated with the constant function $1$,
\begin{eqnarray*}
u_\lambda^t(x,\omega)
&=&\int_{\R^d}\widehat u^\circ(k)\int_\R e^{ik\cdot x-it(|k|^2+\kappa)}\,d\mu_{k,\lambda}^1(\kappa)(\tau_{-x}\omega)\,\dbar k.
\end{eqnarray*}
For vanishing disorder $\lambda=0$ the spectral measures take the form $d\mu_{k,0}^{1}=d\delta_0$ and we recover the Fourier diagonalization of the free Schrödinger flow, cf.~\eqref{eq:free-flow}, while for $\lambda>0$ each Fourier mode $e^{ik\cdot x}$ is deformed into a ``Bloch measure'' $e^{ik\cdot x}d\mu_{k,\lambda}^{1}$.
In the periodic setting the measure $\mu_{k,\lambda}^{1}$ is known to be discrete, leading to the deformation of the plane wave $e^{ik\cdot x}$ into a superposition of so-called Bloch waves, cf.~\cite{Kuchment-93,Asch-Knauf-98}. The picture is very different in the random setting as $\mu_{k,\lambda}^{1}$ is rather expected to be absolutely continuous.

%%%%%%%%%%%%%%%%%%%%%
%%%%%%%%%%%%%%%%%%%%%
%%%%%%%%%%%%%%%%%%%%%

\medskip
\section{Basic spectral theory of fibered operators}\label{sec:spectrum}
This section is devoted to the proof of Propositions~\ref{th:main0}, \ref{th:pert-G-sp}, and~\ref{prop:Fermi}.
We consider general (non-Gaussian) stationary random potentials~$V$ and discuss the fine dependence on the probabilistic structure. Note that our results could also be adapted to the random perturbation of a periodic Schrödinger operator,
in which case fibered operators take the form~\mbox{$-\triangle^\st_k+V_\per+\lambda V$},
where the periodic potential~$V_\per$ models the underlying crystalline structure.

\subsection{Unperturbed fibered operators}\label{sec:spectrum-Hk0}
We give a full account of the spectral properties of the unperturbed operators $\{H_{k,0}^\st\}_k$ on $\Ld^2(\Omega)$.
We start with some general definitions.
For $\phi\in\Ld^2(\Omega)$, we denote its covariance function by $\Cc^\phi(x):=\expecm{\bar\phi(\cdot)\,\phi(\tau_{-x}\cdot)}$, which belongs to $\Ld^\infty(\R^d)$ and is positive definite. By Bochner's theorem, the distributional Fourier transform $\widehat \Cc^\phi$ is then a nonnegative finite measure on $\R^d$ with total mass $\widehat \Cc^\phi(\R^d)=(2\pi)^d\|\phi\|_{\Ld^2(\Omega)}^2$, and is called the {spectral measure of $\phi$}. The set of all such spectral measures will play an important role in this section, so that we give it a name and notation.

\begin{defin}
The \emph{spectrum of the probability space} $(\Omega,\p)$ endowed with a given stationarity structure is defined as the subset
\[\widehat\Omega\,:=\,\big\{(2\pi)^{-d}\widehat \Cc^\phi:\phi\in \Ld^2(\Omega),\,\|\phi\|_{\Ld^2(\Omega)}=1\big\}\,\subset\,\Pc(\R^d).\qedhere\]
\end{defin}

We show that the spectrum of the unperturbed operators $\{H_{k,0}^\st\}_k$ can be completely characterized in terms of properties of $\widehat\Omega$.

\begin{lem}\label{th:spectrum-gen}
Let $V$ be a stationary random field and assume that the underlying probability space $(\Omega,\p)$ is endowed with the $\sigma$-algebra generated by $V$.
For $\phi\in\Ld^2(\Omega)$ and $k\in\R^d$ we denote by $\nu^\phi_k$ the probability measure on $\R^+$ defined by
\begin{align}\label{eq:def-nuk}
\nu^\phi_k([0,t])\,:=\,(2\pi)^{-d}\widehat \Cc^\phi\big(\overline{B_t(-k)}\big), \qquad t\ge0,
\end{align}
and we consider its Lebesgue decomposition
\[\nu^\phi_k=\nu^{\phi}_{k;\pp}+\nu^{\phi}_{k;\sg}+\nu^{\phi}_{k;\ac}\]
into pure point, singularly continuous, and absolutely continuous parts.
Then,
\begin{enumerate}[(i)]
\item The spectrum $\sigma(H_{k,0}^\st)$ of the operator $H_{k,0}^\st$ on $\Ld^2(\Omega)$ is included in $[-|k|^2,\infty)$ and there is an eigenvalue at $0$.
\item For $\ast=\pp,\sg,\ac$, there holds
\begin{gather*}
\sigma_{\ast}(H_{k,0}^\st)=\adh\bigg(\bigcup_{\widehat \Cc^\phi\in\widehat\Omega}h_k\big(\supp\nu^{\phi}_{k;\ast}\big)\bigg),\qquad h_k(t):=t^2-|k|^2.
\end{gather*}
\item The density of the absolutely continuous part of the spectral measure of $H_{k,0}^\st$ associated with $\phi$ takes the form
\[\frac{d\mu^{\phi,\phi}_{k,0;\ac}}{d\lambda}(\lambda)\,=\,\frac{\mathds1_{[-|k|^2,\infty)}(\lambda)}{2\sqrt{\lambda+|k|^2}}\,\frac{d\nu^{\phi}_{k;\ac}}{dt}\big(\sqrt{\lambda+|k|^2}\big).
\qedhere\]
\end{enumerate}
\end{lem}

\begin{proof}
First note that the Fourier symbol of $H_{k,0}^\st$ is given by $y\mapsto |y+k|^2-|k|^2\ge -|k|^2$, which easily implies that the operator $H_{k,0}^\st-E$ has bounded inverse on $\Ld^2(\Omega)$ for all $E<-|k|^2$. The spectrum of $H_{k,0}^\st$ is therefore included in $[-|k|^2,\infty)$, which already proves item~(i).
We now wish to determine the different types of spectrum. For that purpose it suffices to proceed to the Lebesgue decomposition of the spectral measure $\mu_{k,0}^{\phi,\phi}$ of $H_{k,0}^\st$ associated with any $\phi\in\Ld^2(\Omega)$.
We claim that this spectral measure is explicitly given by the following formula, for all $g\in C_b(\R)$,
\begin{align}\label{eq:form-mufk}
\int_\R g\,d\mu_{k,0}^{\phi,\phi} \,=\,\int_{\R^+}g(t^2-|k|^2)\,d\nu_k^\phi(t),
\end{align}
where $\nu_k^\phi$ is defined in the statement.
The conclusion directly follows from this claim since it yields for $\ast=\pp,\sg,\ac$,
\begin{align*}
\int_\R g\,d\mu^{\phi,\phi}_{k,0;\ast} \,=\,\int_\R g(t^2-|k|^2)\,d\nu^{\phi}_{k;\ast}(t),
\end{align*}
where we denote by $\mu_{k,0}^{\phi,\phi}=\mu^{\phi,\phi}_{k,0;\pp}+\mu^{\phi,\phi}_{k,0;\sg}+\mu^{\phi,\phi}_{k,0;\ac}$ the Lebesgue decomposition of $\mu_{k,0}^{\phi,\phi}$, and similarly for $\nu_k^\phi$.

\medskip\noindent
It remains to argue in favor of~\eqref{eq:form-mufk}. By density, it is enough to prove it for $g\in C_b^1(\R)$. Since the Fourier symbol of $H_{k,0}^\st$ is given by $y\mapsto|y+k|^2-|k|^2$, we compute in Fourier space,
\[\int_\R g\,d\mu_{k,0}^{\phi,\phi}=\expec{\bar\phi\, g(H_{k,0}^\st)\,\phi}=(2\pi)^{-d}\int_{\R^d} g(|y+k|^2-|k|^2)\,d\widehat \Cc^\phi(y),\]
and a radial change of variables then yields
\begin{eqnarray*}
\int_\R g\,d\mu_{k,0}^{\phi,\phi}&=&(2\pi)^{-d}\lim_{\e\downarrow0}\frac1{2\e}\int_{-\e}^\infty\int_{\R^d} \mathds1_{|y+k|-\e\le t<|y+k|+\e}\,g(|y+k|^2-|k|^2)\,d\widehat \Cc^\phi(y)\,dt\\
&=&(2\pi)^{-d}\lim_{\e\downarrow0}\frac1{2\e}\int_{0}^\infty g(t^2-|k|^2)\int_{\R^d} \mathds1_{|y+k|-\e\le t<|y+k|+\e}\,d\widehat \Cc^\phi(y)\,dt\\
&=&\lim_{\e\downarrow0}\int_{0}^\infty g(t^2-|k|^2)\,\frac{\nu_k^\phi((t-\e,t+\e])}{2\e}\,dt\\
&=&\int_{0}^\infty g(t^2-|k|^2)\,d\nu_k^\phi(t),
\end{eqnarray*}
that is, \eqref{eq:form-mufk}.
\end{proof}

In particular, the above result implies that the spectrum $\sigma(H_{k,0}^\st)$ can be of any type: for any measure $\mu\in\Pc([-|k|^2,\infty))$ with nontrivial pure point, singularly continuous, and absolutely continuous parts, we can construct a stationary Gaussian field $V$ such that the spectral measure $\mu_{k,0}^{V,V}$ coincides with $\mu$, which entails that the corresponding spectrum of $H_{k,0}^\st$ admits nontrivial pure point, singularly continuous, and absolutely continuous parts. Moreover, the eigenvalue at $0$ does not need to be simple in general.

\medskip
In most cases of interest, the picture is however much neater: the spectrum of the fibered operator $H_{k,0}^\st$ coincides with the whole interval $[-|k|^2,\infty)$ and is made of a simple eigenvalue at $0$ embedded in absolutely continuous spectrum. This is proven to hold below either under strong structural assumptions (e.g.\@ Gaussian structure) or under strong mixing assumptions (e.g.\@ exponential decay of correlations, or integrable $\alpha$-mixing).
We first recall some terminology:
For any diameter $D>0$ and distance $R>0$, we set
\begin{multline}\label{eq:def-alphamix}
\tilde\alpha(R,D;V):=\sup\Big\{\alpha\big(\sigma(\{V(x,\cdot)\}_{x\in S_1}),\sigma(\{V(x,\cdot)\}_{x\in S_2})\big):\\S_1,S_2\in\Bc(\R^d),\,\operatorname{dist}(S_1,S_2)\ge R,\,\diam(S_1),\diam(S_2)\le D\Big\},
\end{multline}
where Rosenblatt's $\alpha$-mixing coefficient is defined for any two sub-$\sigma$-algebras $\A_1,\A_2$ as
\[\alpha(\A_1,\A_2)\,:=\,\sup\Big\{\big|\pr{G_1\cap G_2}-\pr{G_1}\pr{G_2}\!\big|:G_1\in\A_1,\,G_2\in\A_2\Big\}.\]
The random field $V$ is said to be $\alpha$-mixing if for any $D<\infty$ there holds $\tilde\alpha(R,D;V)\to0$ as $R\uparrow\infty$.
We may now state the following criterion, which in particular implies Proposition~\ref{th:main0} when restricted to the Gaussian setting.

\begin{prop}\label{th:spectrum-neat}
Let $V$ be a (nonzero) stationary random field with covariance $\Cc$, and let the probability space $(\Omega,\p)$ be endowed with the $\sigma$-algebra generated by~$V$.
\begin{enumerate}[(i)]
\item Assume that one of the following two conditions holds,
\begin{enumerate}[\quad(i)]
\item[\emph{(C$_1$)}] $V$ is Gaussian and $\Cc$ is not periodic in any direction;
\item[\emph{(C$_2$)}] $\Cc$ has exponential decay, that is, $|\Cc(x)|\le Ce^{-\frac1C|x|}$ for all $x$.
\end{enumerate}
Then the spectrum $\sigma(H_{k,0}^\st)$ coincides with $[-|k|^2,\infty)$.
\item Assume that one of the following two conditions holds,
\begin{enumerate}[\quad(i)]
\item[\emph{(C$_3$)}] $V$ is Gaussian and the (nonnegative measure) Fourier transform $\widehat\Cc$ is absolutely continuous (this is the case for instance if $\Cc$ is integrable);
\item[\emph{(C$_4$)}]
$V$ is $\alpha$-mixing and satisfies $\int_{\R^d}\,\tilde\alpha(|x|,D;V)\,dx\,<\,\infty$ for all $D<\infty$.
\end{enumerate}
Then the eigenvalue at $0$ is simple (with eigenspace $\C$) and
\[\qquad\sigma_{\pp}(H_{k,0}^\st)\,=\,\{0\},\quad\sigma_{\sg}(H_{k,0}^\st)=\varnothing,\quad\sigma(H_{k,0}^\st)\,=\,\sigma_{\ac}(H_{k,0}^\st)\,=\,[-|k|^2,\infty).\qedhere\]
\end{enumerate}
\end{prop}

\begin{proof}
We split the proof into four steps, separately proving~(i) and~(ii) under conditions~(C$_1$), (C$_2$), (C$_3$), and~(C$_4$).

\medskip
\noindent\step1 Proof of~(i) under condition~{\rm (C$_1$)}.\\
Since $V$ is Gaussian and centered, a repeated use of Wick's formula yields for $n\ge1$,
\[\expec{V(0)^nV(x)^n}\,=\,\sum_{m=0}^{n}m!\,\binom{n}{m}^2\,\expecm{V^{n-m}}^2\,\Cc(x)^{m},\]
hence, taking Fourier transform,
\[\widehat\Cc^{(V^n)}\,=\,\sum_{m=0}^{n}m!\,\binom{n}{m}^2\,\expecm{V^{n-m}}^2\,(\widehat\Cc)^{\ast m},\]
where $(\cdot)^{\ast m}$ denotes the $m$th convolution power.
As all terms in the sum are nonnegative,
the support of $\widehat\Cc^{(V^n)}$ therefore contains the support of $(\widehat\Cc)^{\ast n}$, which coincides with the sum $\sum_{m=1}^n\supp\widehat\Cc$. We conclude
\begin{align}\label{eq:def-S-add}
S\,:=\,\bigcup_{n=1}^\infty\sum_{m=1}^n\supp\widehat\Cc~\subset~\adh\bigg(\bigcup_{\widehat\Cc^\phi\in\widehat\Omega}\supp\widehat\Cc^\phi\bigg)\,=:\,T.
\end{align}
As $\Cc$ is nonzero and even,
we note that $S$ is an additive subgroup of $\R^d$, so that its closure must be of the form $A+B$ for some linear subspace $A$ and some discrete additive subgroup~$B$.
Since $\widehat\Cc$ is supported in $S$, if $A$ is not the whole of $\R^d$, we would deduce that~$\Cc$ is periodic in some direction, which is excluded by assumption. We conclude $A=\R^d$, hence $T=\R^d$.
The definition~\eqref{eq:def-nuk} of $\nu_k^\phi$ then implies
\begin{align*}
\adh\bigg(\bigcup_{\widehat\Cc^\phi\in\widehat\Omega}\supp\nu_k^\phi\bigg)\,=\,[0,\infty),
\end{align*}
and $\sigma(H_{k,0}^\st)=[-|k|^2,\infty)$ follows from Lemma~\ref{th:spectrum-gen}.

\medskip
\noindent\step2 Proof of~(i) under condition~{\rm (C$_2$)}.\\
The exponential decay condition $|\Cc(x)|\le Ce^{-\frac1C|x|}$ entails that the Fourier transform $\widehat\Cc$ extends holomorphically to the complex strip $|\Im z|<\frac1C$, and hence its support coincides with the whole of $\R^d$.
It then follows from~\eqref{eq:def-nuk} that the support of $\nu_k^V$ coincides with the whole interval $[0,\infty)$, and therefore $\sigma(H_{k,0}^\st)=[-|k|^2,\infty)$ by Lemma~\ref{th:spectrum-gen}.

\medskip
\noindent\step3 Proof of~(ii) under condition~{\rm (C$_3$)}.\\
Recall the definition~\eqref{eq:def-Pc} of the set of $V$-polynomials,
\[\Pc(\Omega)\,:=\,\bigg\{\sum_{j=1}^n a_j\prod_{l=1}^{m_j}V(x_{lj},\cdot)\,:\,n\ge1,\, a_j\in\C,\,m_j\ge0,\,x_{lj}\in\R^d\bigg\},\]
and let $\Pc_0(\Omega)$ denote the subset of elements of $\Pc(\Omega)$ with vanishing expectation.
For $\phi\in \Pc_0(\Omega)$, since $V$ is Gaussian, Wick's formula allows to express $\Cc^\phi$ explicitly as a linear combination of products of translated copies of the covariance function $\Cc$, without constant term. As the Fourier transform $\widehat\Cc$ is assumed absolutely continuous and integrable,
we conclude that $\widehat\Cc^\phi$ is absolutely continuous and integrable as well for all $\phi\in \Pc_0(\Omega)$. Lemma~\ref{th:spectrum-gen} then implies that for $\phi\in \Pc_0(\Omega)$ the spectral measure $\mu_{k,0}^{\phi,\phi}$ is absolutely continuous. In other words, the absolutely continuous subspace for $H_{k,0}^\st$ contains $\Pc_0(\Omega)$.

\medskip\noindent
It remains to check that $\Pc_0(\Omega)$ is dense in $\Ld^2(\Omega)\ominus\C$.
Given $\phi\in\Ld^2(\Omega)$, by $\sigma(V)$-measurability, we may approximate $\phi$ by a sequence $\phi_n\to \phi$ in $\Ld^2(\Omega)$ of the form $\phi_n:=h_n(V(x_1,\cdot),\ldots,V(x_n,\cdot))$ for some Borel function $h_n$ on~$\R^n$ and some $(x_j)_j\subset\R^d$.
Truncating $V$ and smoothening the Borel functions $h_n$'s, we find $\phi'_n\to \phi$ in $\Ld^2(\Omega)$ of the form $\phi'_n:=h'_n(V_n(x_1,\cdot),\ldots,V_n(x_n,\cdot))$ for some $h_n'\in C^\infty_c(\R^n)$ and $V_n:=(V\wedge n)\vee(-n)$. For each $n$, Weierstrass' approximation theorem then allows to replace the smooth function $h_n'$ by a polynomial $p_n$ in $n$ variables. This proves that $\Pc(\Omega)$ is dense in $\Ld^2(\Omega)$, hence $\Pc_0(\Omega)$ is dense in $\Ld^2(\Omega)\ominus\C$.

\medskip
\step4 Proof of~(ii) under condition~{\rm (C$_4$)}.\\
Arguing as in Step~3, it is enough to prove that the spectral measure $\widehat\Cc^\phi$ is absolutely continuous for all $\phi\in \Ld^2(\Omega)\ominus\C$ of the form $\phi:=h(V(x_1,\cdot),\ldots,V(x_n,\cdot))$ with $n\ge1$, $h\in C^\infty_c(\R^n)$, and $(x_j)_{j=1}^n\subset\R^d$. Let $R:=\max_j|x_j|$.
Since $\phi(\tau_{-x}\cdot)$ is $\sigma(\{V(y,\cdot)\}_{y\in B_R(x)})$-measurable, the $\alpha$-mixing condition for $V$ yields,
cf.~\cite[Theorem~1.2.3]{Doukhan-94},
\[|\Cc^\phi(x)|=\big|\cov{\bar\phi(\cdot)}{\phi(\tau_{-x}\cdot)}\big|\le8\|h\|_{\Ld^\infty(\R^n)}^2\,\tilde\alpha\big((|x|-2R)\vee0,R;V\big).\]
The assumed integrability of the $\alpha$-mixing coefficient then yields $\Cc^\phi\in\Ld^1(\R^d)$, hence the nonnegative Fourier transform $\widehat\Cc^\phi$ is absolutely continuous and belongs to $\Ld^1\cap\Ld^\infty(\R^d)$.
\end{proof}

\subsection{Perturbed fibered operators}\label{sec:spectrum-fibered}
We turn to the perturbed fibered operators and show that the spectrum of  $H_{k,\lambda}^\st=H_{k,0}^\st+\lambda V$ typically coincides with $[-|k|^2+\lambda\infess V,\infty)$.
The precise statement below is however quite intricate and depends on the structure of level sets of $V$. This is to be compared with~\cite[Theorem~5.33]{Pastur-Figotin} for the almost sure spectrum of $H_{\lambda,\omega}$ on~$\Ld^2(\R^d)$. Combined with Theorem~\ref{th:app-sa}, this implies Proposition~\ref{th:pert-G-sp} in the Gaussian setting.

\begin{prop}[Spectrum of $H_{k,\lambda}^\st$]\label{prop:spectr-perturb}
Let $V$ be a stationary random field.
Define the following two closed subsets of $\R$,
\begin{eqnarray*}
\sigma_1(V)&:=&\big\{r\in\R\,:\, \pr{V\in [r-\epsilon,r+\epsilon]}>0~\forall \epsilon>0\big\},\\
\sigma_2(V)&:=&\big\{r\in\R\,:\, \pr{V(x,\cdot)\in [r-\epsilon,r+\epsilon]~\forall x\in B_R}>0~\forall \epsilon,R>0\big\}.
\end{eqnarray*}
Assume that $V$ satisfies the following weak mixing type condition: for all $r\in\sigma_2(V)$ and $\e,R>0$ the level set $V(\cdot,\omega)^{-1}([r-\epsilon,r+\epsilon])$ admits almost surely a bounded connected component containing a ball of radius $R$.
Then for all $k\in\R^d$ there holds
\begin{eqnarray}\label{eq:incl-spectrum}
\sigma(H_{k,0}^\st)+\sigma_2(V)&\subset&\sigma(H_{k,0}^\st+V)\nonumber\\
&\subset& \Big(\conv\big(\sigma(H_{k,0}^\st)\big)+\sigma_1(V)\Big) \bigcap \Big(\sigma(H_{k,0}^\st)+\conv\big(\sigma_1(V)\big)\Big).
\end{eqnarray}
In particular, in the Gaussian setting $V=b(V_0)$ with $V_0$ a nondegenerate stationary Gaussian field and with $b\in C(\R^d)$,
we find $\sigma_1(V)=\sigma_2(V)=[\infess b,\,\supess b]$, cf.~\cite[Theorem~5.34]{Pastur-Figotin}, hence $\sigma(H_{k,0}^\st+V)=[-|k|^2+\infess b,\,\infty)$.
\end{prop}

\begin{rem}
The set $\sigma_1(V)$ is known as the essential range of $V$ and coincides with the spectrum of $V$ as a multiplication operator on $\Ld^2(\Omega)$. The set $\sigma_2(V)$ is a closed subset of $\sigma_1(V)$ and can be much smaller: in the periodic case $\Omega=\T^d$, for instance, there holds $\sigma_2(V)=\varnothing$ unless $V$ is a constant.
\end{rem}

\begin{proof}[Proof of Proposition~\ref{prop:spectr-perturb}]
We split the proof into two steps, separately establishing the first and second inclusions in~\eqref{eq:incl-spectrum}.

\medskip
\step1 Second inclusion in~\eqref{eq:incl-spectrum}.

\noindent
We only prove that $\sigma(H_{k,0}^\st+V)\subset \sigma(H_{k,0}^\st)+\conv(\sigma_1(V))$, while the other inclusion is similar.
If $\conv(\sigma_1(V))=\R$, the inclusion is trivial. It remains to consider the cases when $\conv(\sigma_1(V))$ has the form $[a,\infty)$, $(-\infty,b]$, or $[a,b]$, with $a,b\in\R$. We focus on the case $\conv(\sigma_1(V))=[a,b]$, while the other cases are easier.
Without loss of generality we can assume $a=-b$, so that $b$ coincides with the (finite) operator norm of $V$. Let $E\notin\sigma(H_{k,0}^\st)+[-b,b]$. Since $E\notin \sigma(H_{k,0}^\st)$, we deduce that $H_{k,0}^\st-E$ is invertible and we compute
\[\|(H_{k,0}^\st-E)^{-1}V\|_{\Ld^2(\Omega)\to\Ld^2(\Omega)}<b\|V\|_{\Ld^2(\Omega)\to\Ld^2(\Omega)}=1.\]
Writing
\[H_{k,0}^\st+V-E\,=\,(H_{k,0}^\st-E)\big(\Id+(H_{k,0}^\st-E)^{-1}V\big),\]
and using Neumann series,
we conclude that $H_{k,0}^\st+V-E$ is invertible, which entails that~$E\notin\sigma(H_{k,0}^\st+V)$.

\medskip
\step2 First inclusion in~\eqref{eq:incl-spectrum}.

\nopagebreak\noindent
Given $r\in\sigma_2(V)$ and $E\ge-|k|^2$, we show that there exists a sequence $(\phi_n)_n\subset\Ld^2(\Omega)$ with $\|\phi_n\|_{\Ld^2(\Omega)}=1$ such that $(H_{k,0}^\st-E)\phi_n\to0$ and $(V-r)\phi_n\to0$ in $\Ld^2(\Omega)$, which entails $E+r\in\sigma(H_{k,0}^\st+V)$.
For $\e>0$, consider the open set $O_\e(\omega):=\inter( V(\cdot,\omega)^{-1}(r-\e,r+\e))$ and decompose it into its (at most countable) collection of connected components. Denote by $(O_\e^n(\omega))_n$ the subcollection of bounded connected components.
By assumption, this collection is almost surely nonempty.
For all $n$, we consider the balls included in $O_\e^n(\omega)$ with maximal radius. The maximum radius $R^n_\e(\omega)$ may be attained by different balls and we denote by $(x_\e^{n,m}(\omega))_m$ the collection of their centers. As this collection is a closed bounded set in $\R^d$, we may choose $x_\e^{n}(\omega)$ as the first element in lexicographic order. The set $\{x_\e^n(\omega)\}_n$ defines a (nonempty) stationary point process on~$\R^d$.
Now choose a smooth cut-off function $\chi$ with $\chi(x)=1$ for $|x|\le1$ and $\chi(x)=0$ for $|x|\ge2$, and choose $\xi\in \R^d$ with $|\xi+k|^2=E+|k|^2$. For $R>0$, we define the random variable
\begin{equation}\label{eq:def-varphiR}
\phi_{\e,R}(\omega)=\sum_{n}e^{i \xi\cdot x_\e^n(\omega)}\chi\big(\tfrac2Rx_\e^n(\omega)\big)\,\mathds1_{R_\e^n(\omega)\ge R}.
\end{equation}
\vspace{-0.3cm}\\
By assumption, the decimated stationary point process $\{x_{\e,R}^n(\omega)\}_n:=\{x_\e^n(\omega)\}_{n:R_\e^n(\omega)\ge R}$ is also nonempty and we denote by $\mu_{\e,R}>0$ its intensity.
Since the remaining points in this process are all separated by a distance at least $2R$, the sum~\eqref{eq:def-varphiR} defining $\phi_{\e,R}$ contains at most one non-zero term, and we find
\[2^{-d}\mu_{\e,R}|B_R|\,=\,\p\big[\exists n:x_{\e,R}^n\in B_{\frac12R}\big]\,\le\,\|\phi_{\e,R}\|_{\Ld^2(\Omega)}^2\,\le\,\p\big[\exists n:x_{\e,R}^n\in B_R\big]\,=\,\mu_{\e,R}|B_R|.\]
Next, we estimate
\[|(H_{k,0}^\st-E)\phi_{\e,R}|\,\le\,\tfrac4R\sum_n\Big(|\xi+k|\big|\nabla\chi\big(\tfrac2Rx_{\e,R}^n(\omega)\big)\big|+\tfrac1R\big|\triangle\chi\big(\tfrac2Rx_{\e,R}^n(\omega)\big)\big|\Big),\]
\vspace{-0.5cm}\\
hence,
\vspace{-0.1cm}
\begin{multline*}
\|(H_{k,0}^\st-E)\phi_{\e,R}\|_{\Ld^2(\Omega)}^2\,\lesssim_{k,E}\,R^{-2}\,\p\big[{\exists n:x_{\e,R}^n\in B_R\setminus B_{\frac12R}}\big]\\
\,=\,R^{-2}\mu_{\e,R}|B_R\setminus B_{\frac12R}|\,\lesssim\,R^{-2}\,\|\phi_{\e,R}\|_{\Ld^2(\Omega)}^2.
\end{multline*}
Finally, we compute $\|(V-r) \phi_{R,\e}\|_{\Ld^2(\Omega)}\le\epsilon \|\phi_{R,\e}\|_{\Ld^2(\Omega)}$, and the conclusion follows.
\end{proof}

\subsection{Instability of the bound state}\label{sec:instab}
While the spectrum of the perturbed fibered operators $\{H_{k,\lambda}^\st\}_k$ was easily characterized in the previous section, determining its nature is substantially more involved.
We recall the heuristic prediction from Fermi's Golden Rule, e.g.~\cite[Section~XII.6]{Reed-Simon-78}.
Given a perturbation $H+\lambda W$ of a self-adjoint operator $H$ on $\Hc$, if $H$ admits a simple eigenvalue at $E_0$ with normalized eigenvector~$\psi_0$, and if~$W$ satisfies
\begin{align}\label{eq:Fermi}
\lim_{\e\downarrow0}\,\Im\,\big\langle \bar P_0(W\psi_0),(H-E_0-i\e)^{-1}\bar P_0(W\psi_0)\big\rangle_{\Hc}\,>\,0,
\end{align}
where $\bar P_0$ denotes the orthogonal projection onto $\{\psi_0\}^\bot$,
then the eigenvalue at $E_0$ is expected to dissolve whenever the perturbation is turned on.
The simplest rigorous version of this statement is as follows.

\begin{lem}\label{lem:Fermi}
Let $H,W$ be two self-adjoint operators on a Hilbert space $\Hc$ and let $E_0$ be a simple eigenvalue of $H$ with normalized eigenvector $\psi_0$.
If for some $\delta>0$ there exists a branch $[0,\delta)\to\R\times\Hc:\lambda\mapsto(E_\lambda,\psi_\lambda)$ of class $C^2$ with
\[(H+\lambda W)\psi_\lambda=E_\lambda\psi_\lambda,\qquad(E_\lambda,\psi_\lambda)|_{\lambda=0}=(E_0,\psi_0),\]
then there holds
\begin{align}\label{eq:def-der2-E}
\tfrac{d^2}{d\lambda^2}E_\lambda|_{\lambda=0}\,=\,-2\lim_{\e\downarrow0}\big\langle \bar P_0(W\psi_0),(H-E_0-i\e)^{-1}\bar P_0(W\psi_0)\big\rangle_\Hc,
\end{align}
where $\bar P_0u:=u-\langle\psi_0,u\rangle_\Hc\psi_0$ is the orthogonal projection onto $\{\psi_0\}^\bot$.
In particular, if the right-hand side of~\eqref{eq:def-der2-E} is not real, then there exists no such branch $\lambda\mapsto(E_\lambda,\psi_\lambda)$. This is in particular the case whenever the spectral measure of $H$ associated with $\bar P_0(W\psi_0)$ is absolutely continuous in a neighborhood of $E_0$ and has non-vanishing density at $E_0$.
\end{lem}

\begin{proof}
Assume that there exists a $C^2$ branch $\lambda\mapsto(E_\lambda,\psi_\lambda)$ as in the statement and denote by $(E_0',\psi_0')$ and $(E_0'',\psi_0'')$ the first and second derivatives with respect to $\lambda$ at $\lambda=0$.
Differentiating the eigenvalue relation yields
\[(H-E_0)\psi'_0+W\psi_0=\psi_0E'_0.\]
Taking the scalar product with $\psi_0$, we find
\begin{align}\label{eq:E0'}
E'_0=\langle\psi_0,W\psi_0\rangle_\Hc,
\end{align}
hence
\begin{align*}
(H-E_0)\psi'_0=-\bar P_0(W\psi_0).
\end{align*}
This can be inverted in the form
\begin{align}\label{eq:psi0'-eqn}
\psi'_0=-(H-E_0-i\e)^{-1}\bar P_0(W\psi_0)-i\e(H-E_0-i\e)^{-1}\psi_0'.
\end{align}
Now differentiating the eigenvalue equation twice, we find
\[(H-E_0)\psi''_0+2W\psi'_0=E''_0\psi_0+2E'_0\psi'_0,\]
hence, injecting~\eqref{eq:E0'} and taking the scalar product with $\psi_0$,
\[E''_0=2\big\langle\bar P_0(W\psi_0),\psi'_0\big\rangle_\Hc.\]
Injecting~\eqref{eq:psi0'-eqn} then yields
\begin{equation*}
E''_0=-2\big\langle\bar P_0(W\psi_0),(H-E_0-i\e)^{-1}\bar P_0(W\psi_0)\big\rangle_\Hc
-2\big\langle i\e(H-E_0-i\e)^{-1}\bar P_0(W\psi_0),\psi_0'\big\rangle_\Hc.
\end{equation*}
Since $E_0$ is simple, we find $\mathds1_{\{E_0\}}(H)u=\langle\psi_0,u\rangle_\Hc\psi_0$, hence
\[\lim_{\e\downarrow0}\big\langle i\e(H-E_0-i\e)^{-1}\bar P_0(W\psi_0),\psi_0'\big\rangle_\Hc=\big\langle\mathds1_{\{E_0\}}(H)\bar P_0(W\psi_0),\psi_0'\big\rangle_\Hc=0,\]
and the conclusion follows.
\end{proof}

We apply this result to our setting $H_{k,\lambda}^\st=H_{k,0}^\st+\lambda V$ with $(E_0,\psi_0)=(0,1)$.
The quantity in~\eqref{eq:Fermi} takes the form, for~$k\in\R^d\setminus\{0\}$,
\begin{eqnarray*}
\alpha_k~:=~\lim_{\e\downarrow0}\,\Im\,\expec{V(H_{k,0}^\st-i\e)^{-1}V}&=&
\lim_{\e\downarrow0}\int_{\R^d}\frac{\e\,\widehat\Cc(y)}{(|y+k|^2-|k|^2)^2+\e^2}\dbar y\\
&=&(2\pi)^{-d}\frac{\pi}{2|k|}\lim_{\e\downarrow0}\frac1{2\e}\widehat\Cc\big(B_{|k|+\e}(-k)\setminus B_{|k|-\e}(-k)\big),
\end{eqnarray*}
and Proposition~\ref{prop:Fermi} follows.

%%%%%%%%%%%%%%%%%%%%%
%%%%%%%%%%%%%%%%%%%%%
%%%%%%%%%%%%%%%%%%%%%

\medskip
\section{Perturbative Mourre's commutator approach}\label{sec:Mourre}

This section is devoted to the use of Mourre's theory~\cite{Mourre-80,ABMG-96} to study fibered perturbation problems, in particular proving Theorem~\ref{th:Mourre-pert} and Corollary~\ref{th:mourre-partial}.
We focus on the short-range Gaussian setting, that is, $V=b(V_0)$ for some $b\in C^\infty(\R)$ and some stationary centered Gaussian field~$V_0$ with covariance function $\Cc_0\in\Ld^1\cap\Ld^\infty(\R^d)$, and without loss of generality we assume that the probability space $(\Omega,\p)$ is endowed with the $\sigma$-algebra generated by~$V_0$.

\subsection{Reminder on Mourre's theory}\label{sec:Mourre-intro}
We briefly recall the general purpose of Mourre's theory and its classical application to Schrödinger operators on $\Ld^2(\R^d)$; we refer to~\cite{Mourre-80,ABMG-96} for details.
A self-adjoint operator $H$ with domain $D(H)$ on a Hilbert space $\Hc$ is said to satisfy a \emph{Mourre relation} on an interval $J\subset\R$ with respect to a (self-adjoint) \emph{conjugate operator} $A$ with domain $D(A)\subset \Hc$ if there exists $C_0\ge1$ and a compact operator $K$ such that there holds in the sense of forms,
\begin{equation}\label{eq:Mourre0}
\mathds1_J(H)\,[H,iA]\,\mathds1_J(H)\,\ge\,\tfrac1{C_0}\mathds1_J(H)+K,
\end{equation}
where the commutator $[H,iA]$ is defined as a sesquilinear form on $D(H)\cap D(A)$. The Mourre relation~\eqref{eq:Mourre0} is said to be \emph{strict} if $K=0$.
For technical reasons, one typically requires that the domain of $H$ be invariant under the unitary group $\{e^{itA}\}_{t\in\R}$ generated by~$A$, that is,
\begin{equation}\label{eq:Mourre0-pres}
e^{itA}D(H)\subset D(H),\qquad\forall t\in\R,
\end{equation}
which in particular ensures that $D(H)\cap D(A)$ is dense in $D(H)$, and one further requires $[H,iA]$ to be $H$-bounded. In that case, the sesquilinear form $[H,iA]$ on $D(H)\cap D(A)$ automatically extends to the form of a unique $H$-bounded self-adjoint operator.

\medskip
In a semiclassical perspective, conjugate operators can be viewed as a quantum analogue of escape functions for Hamiltonian dynamical systems.
The main result of Mourre's theory~\cite{Mourre-80,ABMG-96,GGe-99} is that the relation~\eqref{eq:Mourre0} (together with additional regularity assumptions) entails that the eigenvalues of~$H$ in $J$ have finite multiplicity and that $H$ has no singular continuous spectrum in $J$.
In addition, a strict Mourre relation implies that the spectral measure is absolutely continuous on~$J$. This is actually a simple consequence of the virial theorem: if~$\lambda$ was an eigenvalue in~$J$ with normalized eigenvector $\psi$, then a strict Mourre inequality would formally yield
\[0=\langle\psi,[H,iA]\psi\rangle_\Hc\ge \tfrac1{C_0}\|\psi\|^2_\Hc,\]
a contradiction.
Alternatively, this short formal proof can be rewritten by noting that a strict Mourre relation implies ballistic transport for the flow $e^{itH}$ with respect to the conjugate operator $A$: for $\phi\in\mathds1_J(H)\Hc$ there holds
\[\partial_t\langle e^{itH}\phi,(-A)e^{itH}\phi\rangle_\Hc=\langle e^{itH}\phi,[H,iA]e^{itH}\phi\rangle_\Hc\ge \tfrac1{C_0}\|\phi\|_{\Hc}^2,\]
hence $\langle e^{itH}\phi,(-A)e^{itH}\phi\rangle_\Hc\ge \tfrac1{C_0}t\|\phi\|^2_\Hc+\langle \phi,(-A)\phi\rangle_\Hc$, thus prohibiting $\phi$ from being an eigenvector.
In addition to such spectral information, the Mourre relation~\eqref{eq:Mourre0} is further known to yield useful a priori estimates on boundary values of the resolvent in form of limiting absorption principles~\cite{Mourre-80,JMP-84}.

\medskip
We recall the standard construction of a Mourre conjugate operator for Schrödinger operators on $\Ld^2(\R^d)$, e.g.~\cite{Mourre-80}.
Considering the unitary group of dilations $U_tg:=e^{td/2}g(e^t\cdot)$ on $\Ld^2(\R^d)$, and noting that $U_{-t}(-\triangle)U_t=e^{2t}(-\triangle)$, we deduce by differentiation,
\[[-\triangle\,,\,iA]=2(-\triangle),\]
where $iA$ denotes the generator of dilations, that is, $A:=\frac1{2i}(x\cdot\nabla+\nabla\cdot x)$ on $\Ld^2(\R^d)$.
This implies that $-\triangle$ satisfies a strict Mourre inequality on $[\e,\infty)$ for all $\e>0$ with conjugate operator $A$,
\[\mathds1_{[\e,\infty)}(-\triangle)\,[-\triangle\,,\,iA]\,\mathds1_{[\e,\infty)}(-\triangle)\,\ge\,2\e\,\mathds1_{[\e,\infty)}(-\triangle).\]
In a semiclassical perspective, the conjugate operator $A$ corresponds to the escape function $(x,p)\mapsto x\cdot p$ for the free Hamiltonian $H(x,p)=\tfrac12|p|^2$.
Next, given a $\triangle$-bounded potential $W:\R^d\to\R$, we compute
\[[W,iA]=-x\cdot\nabla W,\]
so that the commutator $[W,iA]$ is bounded whenever the function $x\mapsto x\cdot\nabla W(x)$ is bounded.
For $\lambda$ small enough, this easily entails that the Schrödinger operator $-\triangle+\lambda W$ on $\Ld^2(\R^d)$ also satisfies a strict Mourre inequality on $[\e,\infty)$. This follows from the first general property below and illustrates the flexibility of Mourre's theory.

\begin{lem}\label{lem:Mourre-perturb}
Let $H$ be a self-adjoint operator on a Hilbert space $\Hc$, assume that $H$ satisfies a Mourre relation~\eqref{eq:Mourre0} on a bounded interval~$J\subset\R$ with respect to a conjugate operator~$A$, that the domain of $H$ is invariant under $\{e^{itA}\}_{t\in\R}$, cf.~\eqref{eq:Mourre0-pres}, and that $[H,iA]$ is $H$-bounded. Let also $W$ be symmetric and $H$-bounded.
\begin{enumerate}[(i)]
\item \emph{Mourre relation under perturbation:}\\
If the commutator $[W,iA]$ is $H$-bounded, then for all $J'\Subset J$ and $\lambda$ small enough the perturbed operator $H_\lambda:=H+\lambda W$ satisfies a Mourre relation on~$J'$ with respect to~$A$. In addition, if $H$ satisfies a strict Mourre relation, then $H_\lambda$ does too.
\item \emph{Strict relation on orthogonal complement of an eigenspace:}\\
If $H$ has an eigenvalue $E_0\in J$ with eigenprojector $P_0$, then there exists a neighborhood $J'\Subset J$ of $E_0$ such that the restriction $\bar H:=\bar P_0 H\bar P_0$ of $H$ to the range of $\bar P_0:=\Id-P_0$ satisfies a strict Mourre relation on $J'$ with conjugate operator $\bar A:=\bar P_0 A\bar P_0$.
\qedhere
\end{enumerate}
\end{lem}

\begin{proof}
We start with the proof of~(i). As the perturbation $W$ is $H$-bounded, the operator $H_\lambda$ has the same domain  as $H$ for $\lambda$ small enough in view of the Kato-Rellich theorem, hence by assumption its domain $D(H_\lambda)=D(H)$ is invariant under $\{e^{itA}\}_{t\in\R}$. Furthermore, the commutator $[H_\lambda,iA]$ is $H$-bounded, hence $H_\lambda$-bounded. Now, choose $h\in C^\infty_c(\R)$ with $\mathds1_{J'}\le h\le\mathds1_J$. Multiplying by $h(H)$ both sides of the Mourre relation for $H$ yields
\[h(H)[H,iA]h(H)\ge \frac1{C_0}h(H)+h(H)Kh(H).\]
As $[W,iA]$ is $H$-bounded, we deduce
\[h(H)[H_\lambda,iA]h(H)\ge \frac1{C_0}h(H)-C\lambda+h(H)Kh(H).\]
Noting that the $H$-boundedness of $W$ implies $\|h(H)-h(H_\lambda)\|\lesssim\lambda\|h'\|_{\Ld^\infty(\R)}$,
and further using the $H$-boundedness of $[H_\lambda,iA]$, we deduce
\begin{eqnarray*}
h(H_\lambda)[H_\lambda,iA]h(H_\lambda)&\ge& h(H)[H_\lambda,iA]h(H)-C\lambda\\
&\ge& \frac1{C_0}h(H)-2C\lambda+h(H)Kh(H)\\
&\ge& \frac1{C_0}h(H_\lambda)-2C\lambda+h(H)Kh(H).
\end{eqnarray*}
Multiplying both sides by $\mathds1_{J'}(H_\lambda)$, the conclusion~(i) follows for $\frac1{C_0}-2C\lambda\ge\frac1{2C_0}$.

\medskip\noindent
We turn to the proof of~(ii).
As $\bar P_0$ commutes with $H$, multiplying by $\bar P_0$ both sides of the Mourre relation for $H$ yields, on the range of $\bar P_0$,
\[\mathds1_J(\bar H)[\bar H,i\bar A]\mathds1_J(\bar H)\ge\frac1{C_0}\mathds1_J(\bar H)+\bar K,\]
in terms of $\bar H:=\bar P_0 H\bar P_0$, $\bar A:=\bar P_0A\bar P_0$, $\bar K:=\bar P_0K\bar P_0$. Multiplying both sides with $\mathds1_J(\bar H)$, the compact operator is replaced by $\mathds1_J(\bar H)\bar K\mathds1_J(\bar H)$. Since $\mathds1_J(\bar H)$ converges strongly to~$0$ on the range of $\bar P_0$ as $J\to\{E_0\}$, the conclusion~(ii) follows.
\end{proof}

Next, we state the following general result by Cattaneo, Graf, and Hunziker~\cite{CGH-06}, showing how Mourre's theory can be exploited to analyze the instability of embedded bound states in form of an approximate resonance theory; see also~\cite{Orth-90,Hunziker-90,Sofer-Weinstein-98,Merkli-Sigal-99,Costin-Soffer-01}. Although Mourre's theory does not allow to deduce the existence of resonances in any strong sense, it is shown to have essentially the same dynamical consequences.
The proof further allows for asymptotic expansions to finer accuracy in~$\lambda$, as well as for a description of the flow $e^{-iH_\lambda t}\mathds1_{J'}(H_\lambda)\psi_0$ projected on a whole class of ``smooth'' states rather than on~$\psi_0$ only, but such improvements are not pursued here.

\begin{theor}[Dynamical resonances from Mourre's theory;~\cite{CGH-06}]\label{th:CGH}
Let $H$ be a self-adjoint operator on a Hilbert space $\Hc$, let $W$ be symmetric and $H$-bounded, and consider the perturbation $H_\lambda:=H+\lambda W$. Let $E_0$ be a simple eigenvalue of $H$ with normalized eigenvector~$\psi_0$, and assume that the following properties hold:
\begin{enumerate}[\qquad$\bullet$]
\item There is a self-adjoint conjugate operator $A$ and a neighborhood $J\subset\R$ of $E_0$ such that $H$ satisfies a Mourre relation on $J$ with respect to $A$ in the sense of~\eqref{eq:Mourre0}.
In addition, the domain of $H$ is invariant under $\{e^{itA}\}_{t\in\R}$, cf.~\eqref{eq:Mourre0-pres}.
\item The iterated commutators $\operatorname{ad}^{k}_A(H)$ and $\operatorname{ad}^{k}_A(W)$ are $H$-bounded for $0\le k\le6$, where iterated commutators are defined by $\operatorname{ad}_A^{0}(H)=H$ and recursively $\operatorname{ad}_A^{k+1}(H)=[\operatorname{ad}_A^{k}(H),iA]$ for $k\ge0$.
\item Fermi's condition~\eqref{eq:Fermi} is satisfied, that is,
\[\lim_{\e\downarrow0}\Im\big\langle\bar P_0(W\psi_0),(\bar H-E_0-i\e)^{-1}\bar P_0(W\psi_0)\big\rangle_\Hc>0,\]
where $\bar P_0$ denotes the orthogonal projection onto $\{\psi_0\}^\bot$ and where we have set for abbreviation $\bar H:=\bar P_0 H\bar P_0$.
\end{enumerate}
Then there exists $\{z_\lambda\}_{\lambda>0}\subset\C$ with $\Im z_\lambda<0$ such that for all neighborhoods~$J'\Subset J$ of $E_0$ there holds for all $t\ge0$,
\[\Big|\big\langle\psi_0,e^{-iH_\lambda t}\mathds1_{J'}(H_\lambda)\psi_0\big\rangle_\Hc-e^{-iz_\lambda t}\Big|\lesssim_{J,J'}\lambda^2|\!\log\lambda|,\]
where the dynamical resonance $z_\lambda$ satisfies
\[z_\lambda=E_0+\lambda\langle\psi_0,W\psi_0\rangle_\Hc-\lambda^2\lim_{\e\downarrow0}\big\langle\bar P_0(W\psi_0),(\bar H-E_0-i\e)^{-1}\bar P_0(W\psi_0)\big\rangle_\Hc+o(\lambda^2).\qedhere\]
\end{theor}

\begin{proof}[Idea of the proof]
We include a brief summary of the proof for the reader's convenience, and refer to~\cite{CGH-06} for full details. Starting point is the following Feshbach-Schur complement formula for the resolvent, for $\Im z>0$, in terms of the restriction $\bar H_\lambda:=\bar P_0 H_\lambda\bar P_0$,
\begin{multline*}
\big\langle\psi_0,(z-H_\lambda)^{-1}\psi_0\big\rangle_\Hc\\
\,=\,\Big(z-E_0-\lambda \langle\psi_0,W\psi_0\rangle_\Hc-\lambda^2\big\langle \bar P_0(W\psi_0),(z-\bar H_\lambda)^{-1}\bar P_0(W\psi_0)\big\rangle_\Hc\Big)^{-1}.
\end{multline*}
Next, recall that Lemma~\ref{lem:Mourre-perturb} ensures that the restriction $\bar H_\lambda$ on the range of $\bar P_0$ satisfies a strict Mourre relation close to $E_0$. In view of~\cite{JMP-84}, together with the $H$-boundedness of iterated commutators $\operatorname{ad}^{k}_A(H)$ and $\operatorname{ad}^{k}_A(W)$ for $0\le k\le6$, this strict Mourre relation
implies the $C^4$-smoothness of boundary values on $J$ of the resolvent
\[z\mapsto\big\langle \bar P_0(W\psi_0),(z-\bar H_\lambda)^{-1}\bar P_0(W\psi_0)\big\rangle_\Hc,\qquad\Im z>0,\,\Re z\in J.\]
Inserting a Taylor expansion for the latter in the above Feshbach-Schur complement formula, we construct an approximate meromorphic extension for $z\mapsto\langle\psi_0,(z-H_\lambda)^{-1}\psi_0\rangle_\Hc$. The conclusion then follows from complex deformation techniques similarly as for true resonances as in Section~\ref{sec:expdec-conj-pr}.
\end{proof}

\begin{rem}\label{rem:CGH}
As it is clear from the proof, cf.~\cite{CGH-06}, we mention for later reference that a similar result holds if $H=H_\lambda^\circ$ and $W=W_\lambda^\circ$ further depend on $\lambda$. More precisely, assume for all $\lambda$ that $E_0$ is a simple eigenvalue of $H_\lambda^\circ$ with normalized eigenvector $\psi_0$ (independent of $\lambda$, say), that the restriction of $H_\lambda^\circ$ on the range of $\bar P_0$ satisfies a strict Mourre relation on a neighborhood $J$ of $E_0$ with conjugate operator $A$ and constant $C_0$ (independent of~$\lambda$), that the domain of $H_\lambda^\circ$ is invariant under $\{e^{itA}\}_{t\in\R}$, and that iterated commutators $\operatorname{ad}^k_A(H_\lambda^\circ)$ are $H_\lambda^\circ$-bounded by $C_0$ for $0\le k\le 6$.
Next, assume that the perturbation~$W_\lambda$ is bounded in the sense of $\|\langle A\rangle^6\bar P_0(W_\lambda\psi_0)\|_\Hc\le C_0$, and that iterated commutators $\operatorname{ad}_A^k(\lambda W_\lambda)$ are $H_\lambda^\circ$-bounded and small enough in the sense that for $0\le k\le 6$ and $\phi\in\Hc$,
\[\|\operatorname{ad}_A^k(\lambda W_\lambda)\phi\|_\Hc\le \frac1{C_1}(\|H_\lambda^\circ \phi\|_\Hc+\|\phi\|_\Hc),\]
for some large enough constant $C_1$ only depending on $C_0$.
Then the same result holds as in Theorem~\ref{th:CGH} above for the perturbed operator $H_\lambda=H_\lambda^\circ+\lambda W_\lambda^\circ$.
\end{rem}

\subsection{Reminder on Malliavin calculus}\label{chap:Mall}
We recall some notation and tools from Malliavin calculus for the fine analysis of nonlinear functionals of the underlying Gaussian field~$V_0$ with covariance function $\Cc_0\in\Ld^1\cap\Ld^\infty(\R^d)$; we refer to~\cite{Malliavin-97,Nualart,NP-book} for details.
We start by underlining the Hilbert structure associated with the Gaussian field $V_0$.
The random variables $\Vc_0(\zeta):=\int_{\R^d} V_0 \zeta$ with $\zeta\in C^\infty_c(\R^d)$ are centered Gaussians with covariance
\[\langle\Vc_0(\zeta'),\Vc_0(\zeta)\rangle_{\Ld^2(\Omega)}\,=\,
\iint_{\R^d\times\R^d} \Cc_0(x-y)\,\overline{\zeta'(x)}\,\zeta(y)\,dxdy.\]
We consider the completion of $C^\infty_c(\R^d)$ endowed with the (semi)norm
\[\|\zeta\|_{\Hf}:=\langle\zeta,\zeta\rangle_\Hf,\qquad\langle\zeta',\zeta\rangle_\Hf:=\iint_{\R^d\times\R^d} \Cc_0(x-y)\,\overline{\zeta'(x)}\,\zeta(y)\,dxdy,\]
and we denote by $\Hf$ the quotient of this completed space with respect to the kernel of~\mbox{$\|\cdot\|_\Hf$}.
The normed space $\Hf$ is a separable Hilbert space
and the random field $V_0$ satisfies
the isometry relation
\[\langle\Vc_0(\zeta'),\Vc_0(\zeta)\rangle_{\Ld^2(\Omega)}=\langle\zeta',\zeta\rangle_\Hf.\]
The map $\Vc_0:\zeta\mapsto \Vc_0(\zeta)$ then extends as a linear isometric embedding $\Hf\to\Ld^2(\Omega)$ and constitutes a so-called isonormal Gaussian process over $\Hf$.
The structure of $\Hf$ is conveniently characterized as follows:
as $\Cc_0\in\Ld^1(\R^d)$, the (nonnegative measure) Fourier transform $\widehat\Cc_0$ is absolutely continuous,
hence the square root $\widehat\Cc_0^\circ:=(\widehat\Cc_0)^{1/2}$ belongs to~$\Ld^2(\R^d)$ and the linear map
\begin{equation}\label{eq:K-def}
K:C_c^\infty(\R^d)\to\Ld^2(\R^d):\zeta\mapsto \Cc_0^\circ\ast\zeta
\end{equation}
extends into a unitary transformation $K:\Hf\to\Ld^2(\R^d)$.
Note that for all $x$ the Dirac mass~$\delta_x$ is (a representative of) an element of~$\Hf$ with $K\delta_x=\Cc_0^\circ(\cdot-x)$.
By definition, the linear isometric embedding $\Ld^2(\R^d)\to\Ld^2(\Omega):u\mapsto \Vc_0(K^{-1}u)$ is a white noise.

\medskip
As a model dense subspace of $\Ld^2(\Omega)$, instead of considering the linear subspace $\Pc(\Omega)$ of $V_0$-polynomials, cf.~\eqref{eq:def-Pc}, we define the following slightly more convenient subspace,
\begin{multline}\label{eq:def-RcO}
\Rc(\Omega):=\Big\{g\big(\Vc_0(\zeta_1),\ldots,\Vc_0(\zeta_n)\big)\,:\,n\in\N,\,g:\R^n\to\C\text{ polynomial},\\
\,\zeta_1,\ldots,\zeta_n\in C^\infty_c(\R^d;\R)\Big\}.
\end{multline}
Recall that we implicitly assume that the underlying probability space~$(\Omega,\p)$ is endowed with the minimal $\sigma$-algebra generated by $V_0$, thus ensuring that $\Rc(\Omega)$ is indeed dense in $\Ld^2(\Omega)$.
This allows to define operators and prove properties on the simpler subspace $\Rc(\Omega)$ in a concrete way before extending them to $\Ld^2(\Omega)$ by density.

\medskip
For a random variable $\phi\in \Rc(\Omega)$, say $\phi=g(\Vc_0(\zeta_1),\ldots,\Vc_0(\zeta_n))$,
we define its \emph{Malliavin derivative} $D\phi\in \Ld^2(\Omega;\Hf)$ as
\begin{align*}
D\phi\,:=\,\sum_{j=1}^n\zeta_{j}\,(\partial_{j}g)\big(\Vc_0(\zeta_1),\ldots,\Vc_0(\zeta_n)\big),
\end{align*}
and similarly, for all $p\ge1$, its $p$th {Malliavin derivative} $D^p\phi\in \Ld^2(\Omega;\Hf^{\otimes p})$ is given by
\begin{align*}
D^p\phi:=\sum_{1\le j_1,\ldots,j_p\le n}\big(\zeta_{j_1}\otimes\ldots\otimes\zeta_{j_p}\big)\big(\partial_{j_1\ldots j_p}^pg\big)\!\big(\Vc_0(\zeta_1),\ldots,\Vc_0(\zeta_n)\big).
\end{align*}
Note that by definition this belongs to the symmetric tensor product, $D^p\phi\in\Ld^2(\Omega;\Hf^{\odot p})$.
These operators on $\Rc(\Omega)$ are closable on $\Ld^2(\Omega)$.
We then set
\[\|X\|_{\Dm^{p,2}(\Omega)}^2:=\expec{|X|^2}+\sum_{j=1}^p\expec{\|D^jX\|_{\Hf^{\otimes j}}^2},\]
we define the \emph{Malliavin-Sobolev space} $\Dm^{p,2}(\Omega)$ as the closure of $\Rc(\Omega)$ for this norm, and we extend the $p$th Malliavin derivative $D^p$ by density to this space.

\medskip
Next, we define the corresponding \emph{divergence operator} $D^*$ as the adjoint of the Malliavin derivative $D$, and similarly, for all $p\ge1$, the $p$th-order divergence operator $(D^*)^p$ as the adjoint of $D^p$. In other words, this is defined by the following integration by parts formula, for all $\phi'\in \dom (D^*)^p\subset\Ld^2(\Omega;\Hf^{\otimes p})$ and $\phi\in\Rc(\Omega)$,
\[\langle\phi,(D^*)^p\phi'\rangle_{\Ld^2(\Omega)}=\expec{\langle D^p\phi,\phi'\rangle_{\Hf^{\otimes p}}}.\]
The so-called Meyer inequalities ensure that the $p$th divergence operator $(D^*)^p$ extends as a bounded operator $\Dm^{m,2}(\Omega;\Hf^{\otimes p})\to\Dm^{m-p,2}(\Omega)$ for all $m,n\ge p$, hence in particular its domain contains~$\Dm^{p,2}(\Omega;\Hf^{\otimes p})$.
For $\phi\in\Rc(\Omega)$ and $\zeta\in\Hf$, a direct computation yields
\begin{equation*}
D^*(\zeta\phi)\,=\,\Vc_0(\zeta)\phi-\langle\zeta,D\phi\rangle_\Hf.
\end{equation*}
Due to this relation, with in particular $D^*\zeta=\Vc_0(\zeta)$, the divergence operator $D^*$ is sometimes referred to as the \emph{Skorokhod integral}; see also the notion of multiple integrals below.

\medskip
With the Malliavin derivative and the divergence operator at hand, we may construct the corresponding \emph{Ornstein--Uhlenbeck operator} (or infinite-dimensional Laplacian)
\begin{equation}\label{eq:Orn-Uhl}
\Lc:=D^* D,
\end{equation}
as a self-adjoint operator acting on $\Ld^2(\Omega)$ with domain $\Dm^{2,2}(\Omega)$.
The spectrum of $\Lc$
is~$\sigma(\Lc)=\N$ and its kernel coincides with constants.
In particular, the following Poincaré inequality holds: for all $\phi\in \Dm^{1,2}(\Omega)$ with $\expec{\phi}=0$,
\[\|\phi\|_{\Ld^2(\Omega)}^2\,\le\,\expec{\bar\phi\Lc \phi}\,=\,\expec{\|D\phi\|_{\Hf}^2}.\]
This ensures the invertibility of the restriction of $\Lc$ to $\Ld^2(\Omega)\ominus\C$,
and allows to define a pseudo-inverse $\Lc^{-1}:=\Pi\Lc^{-1}\Pi$ on $\Ld^2(\Omega)$ in terms of the projection $\Pi:=\Id-\E$.

\medskip
We turn to a spectral decomposition of $\Lc$.
For that purpose, for $p\ge0$, we first define the $p$th \emph{multiple integral} $I_p$ as the bounded linear operator $\Hf^{\odot p}\to\Ld^2(\Omega)$ given by the restriction of the $p$th divergence operator, that is, $I_p(\zeta):=(D^*)^p\zeta$ for all $\zeta\in\Hf^{\odot p}$. Alternatively, $I_p$ can be characterized as follows: for all $\zeta\in C^\infty_c(\R^d;\R)$ with $\|\zeta\|_{\Hf}=1$ there holds
\[I_p(\zeta^{\otimes p})=H_p(\Vc_0(\zeta)),\]
where $H_p$ denotes the $p$th Hermite polynomial, that is, $H_p(t):=e^{\frac12t^2}(-\frac{d}{dt})^pe^{-\frac12t^2}$.
The image of $I_p$ is known as the $p$th \emph{Wiener chaos} $\Hc_p\subset \Ld^2(\Omega)$.
Properties of Hermite polynomials easily imply the following orthogonality property: for all $p,q\ge0$ and $\zeta\in\Hf^{\odot p},\zeta'\in\Hf^{\odot q}$,
\begin{equation}\label{eq:orth-Ip}
\langle I_q(\zeta'),I_p(\zeta)\rangle_{\Ld^2(\Omega)}\,=\,\delta_{pq}\,p!\,\langle \zeta',\zeta\rangle_{\Hf^{\otimes p}},
\end{equation}
which in particular entails that $I_p$ is a unitary transformation $\Hf^{\odot p}\to\Hc_p$, where the symmetric tensor product~$\Hf^{\odot p}$ is endowed with the norm
\begin{equation*}
\|\zeta\|_{\Hf^{\odot p}}:=\sqrt{p!}\,\|\zeta\|_{\Hf^{\otimes p}}.
\end{equation*}
In view of~\eqref{eq:K-def}, recall that $\Hf^{\odot p}$ is further isometric to $\Ld^2(\R^d)^{\odot p}=\Ld^2_\Sym((\R^d)^p)$, endowed with the norm
\begin{equation*}
\|u_p\|_{\Ld^2_\Sym((\R^d)^p)}:=\sqrt{p!}\,\|u_p\|_{\Ld^2((\R^d)^p)},
\end{equation*}
so that we are led to the following unitary transformations,
\begin{equation}\label{eq:isom-L2-Wiener}
\begin{array}{ccccc}
\Ld^2_\Sym((\R^d)^p)&\xrightarrow\sim&\Hf^{\odot p}&\xrightarrow\sim&\Hc^p\\
u_p&\mapsto& (K^{-1})^{\otimes p}u_p&\mapsto& I_p((K^{-1})^{\otimes p}u_p),
\end{array}
\end{equation}
and we write for abbreviation
\[J_p(u_p):=I_p((K^{-1})^{\otimes p}u_p).\]
As a consequence of the orthogonality property~\eqref{eq:orth-Ip}, the following \emph{Wiener chaos expansion} holds in form of a (bosonic) Fock space decomposition,
\begin{equation}\label{eq:wiener-chaos}
\Ld^2(\Omega)\,=\,\bigoplus_{p=0}^\infty\Hc_p\,\cong\,\bigoplus_{p=0}^\infty\Ld^2_\Sym((\R^d)^p).
\end{equation}
More precisely, for all $\phi\in\Ld^2(\Omega)$, we can expand
\[\phi=\sum_{p=0}^\infty I_p(\phi_p)=\sum_{p=0}^\infty J_p(u_p),\]
for some unique collection of kernels $\phi_p\in\Hf^{\odot p}$ or $u_p\in\Ld^2_\Sym((\R^d)^p)$, where the expansion is converging in $\Ld^2(\Omega)$. The Stroock formula asserts
\[\phi_p=\frac1{p!}\expec{D^pF},\qquad\text{provided $\phi\in\Dm^{p,2}(\Omega)$}.\]
It can be checked that the $p$th Wiener chaos $\Hc_p$ coincides with the eigenspace of the Ornstein--Uhlenbeck operator $\Lc$ associated with the eigenvalue $p$, so that the Wiener chaos expansion~\eqref{eq:wiener-chaos} coincides with the spectral decomposition of $\Lc$.

\medskip
Intuitively, higher chaoses can be viewed as characterizing higher complexity of randomness. In our study of random Schrödinger operators, the use of Wiener chaos decomposition is reminiscent of cumulant expansions for interacting particle systems, e.g.~\cite{BGSRS-20,D-19}.

\subsection{A new class of operators on $\Ld^2(\Omega)$}\label{sec:oper-L2}
This section is devoted to a general construction allowing to transfer operators on $\Ld^2(\R^d)$ into operators on $\Ld^2(\Omega)$, which will be a key tool in the sequel and is analogous to second quantization in quantum field theory.
Given a bounded operator $T$ on $\Ld^2(\R^d)$, for all $p\ge0$, we denote by $\Op^\circ_p(T)$ the bounded operator on $\Ld^2(\R^d)^{\odot p}=\Ld^2_\Sym((\R^d)^p)$ given by
\[\Op^\circ_p(T)\, g^{\otimes p}=\sum_{j=0}^{p-1}g^{\otimes j}\otimes Tg\otimes g^{\otimes (p-j-1)},\qquad g\in\Ld^2(\R^d).\]
Via the isomorphism~\eqref{eq:isom-L2-Wiener},
we can then construct a {bounded} operator $\Op_p(T)$ on the $p$th Wiener chaos $\Hc_p$ via
\[\Op_p(T)J_p(u_p):=J_p\big(\,\overline{\Op_p^\circ(T^*)\,\overline{u_p}}\,\big),\qquad u_p\in \Ld^2_\Sym((\R^d)^p),\]
where $T^*$ is the adjoint of $T$ on $\Ld^2(\R^d)$.
In particular, on the first chaos, this definition formally yields $\Op_1(T)\int_{\R^d}V_0\zeta=\int_{\R^d}(TK^{-1}V_0)(K\zeta)$.
Via the Wiener chaos decomposition~\eqref{eq:wiener-chaos}, we then let $\Op(T)$ denote the densely defined operator on $\Ld^2(\Omega)$ given by the direct sum
\[\Op(T)=\bigoplus_{p=0}^\infty\Op_p(T).\]
As $\Op(T)$ is obviously $\Lc$-bounded for bounded $T$,
the map $\Op$ provides a linear embedding $\Bc(\Ld^2(\R^d);\Ld^2(\R^d))\to\Bc(\Dm^{2,2}(\Omega);\Ld^2(\Omega))$, but this is however not a group homomorphism as in particular $\Op(\Id)=\Lc$.

\medskip
If~$T$ is bounded and self-adjoint, then $\Op(T)$ is self-adjoint on $\Dm^{2,2}(\Omega)$. More generally, if $T$ is unbounded on $\Ld^2(\R^d)$ and essentially self-adjoint on some subset $\Cc$, then $\Op^\circ_p(T)$ defines an essentially self-adjoint operator on $\cal C^{\odot p}$, hence $\Op_p(T)$ is
essentially self-adjoint on $J_p \Cc^{\odot p}$, cf.~\cite[Theorem VIII.33]{Reed-Simon-72}, and in turn $\Op(T)$ defines an
essentially self-adjoint operator on
\begin{displaymath}
\bigg\{\phi=\sum_{p=0}^\infty J_p(u_p)\,:\, u_p\in \cal C^{\odot p}~\forall p,\text{ and }u_p=0\text{ for $p$ large enough}\bigg\}.
\end{displaymath}
In particular, noting that the definition~\eqref{eq:def-RcO} of $\Rc(\Omega)$ can be reformulated as
\[\Rc(\Omega)\,=\,\bigg\{\phi=\sum_{p=0}^\infty J_p(u_p)\,:\, u_p\in C^\infty_c(\R^d)^{\odot p}~\forall p,\text{ and }u_p=0\text{ for $p$ large enough}\bigg\},\]
we deduce that if $T$ is essentially self-adjoint on $C^\infty_c(\R^d)$, then $\Op(T)$ is essentially self-adjoint on $\Rc(\Omega)$. Similarly, if $T$ leaves $C^\infty_c(\R^d)$ invariant, then $\Op(T)$ leaves $\Rc(\Omega)$ invariant. Also note that the operators $\Op(T)$ and $\cal L$ strongly commute since $\Lc$ acts as $p\Id$ on~$\Hc_p$ and since $\Op(T)$ preserves the chaos decomposition.

\medskip
Next, the following shows that the stationary gradient $\nabla^\st$ corresponds to the spatial gradient $\nabla$ via this embedding $\Op$.

\begin{lem}\label{lem:nabla-Op}
There holds $\nabla^\st=\Op(\nabla)$ on $\Ld^2(\Omega)$. In particular, $\nabla^\st$ preserves the chaos decomposition and commutes strongly with $\cal L$.
\end{lem}

\begin{proof}
Given $\zeta\in C^\infty_c(\R^d;\R)$, we compute
\[\nabla^\st I_p(\zeta^{\otimes p})=\nabla^\st H_p(\Vc_0(\zeta))=H_p'(\Vc_0(\zeta))\nabla^\st\Vc_0(\zeta).\]
Noting that $\nabla^\st\Vc_0(\zeta)=-\Vc_0(\nabla\zeta)$ and recalling that Hermite polynomials satisfy $H_p'=pH_{p-1}$, we deduce
\begin{equation}\label{eq:partial-res-nabla-Ip}
\nabla^\st I_p(\zeta^{\otimes p})=-p\Vc_0(\nabla\zeta)H_{p-1}(\Vc_0(\zeta))=-pI_1(\nabla\zeta)I_{p-1}(\zeta^{\otimes(p-1)}).
\end{equation}
Next, we appeal to the following useful product formula for multiple integrals (see e.g.~\cite[Section~2.7.3]{NP-book} for a more general statement):
for all $q\ge1$ and $\zeta_0,\zeta_1\in C^\infty_c(\R^d;\R)$,
\begin{equation}\label{eq:prod-form}
I_1(\zeta_1)I_q(\zeta_0^{\otimes q})=I_{q+1}(\zeta_1\widetilde\otimes \zeta_0^{\otimes q})+qI_{q-1}(\zeta_1\widetilde\otimes_1\zeta_0^{\otimes q}),
\end{equation}
where we have set
\begin{eqnarray*}
\zeta_1\widetilde\otimes\zeta_0^{\otimes q}&:=&\frac1{q+1}\sum_{j=0}^{q}\zeta_0^{\otimes j}\otimes\zeta_1\otimes\zeta_0^{\otimes(q-j)},\\
\zeta_1\widetilde\otimes_1\zeta_0^{\otimes q}&:=&\langle\zeta_1,\zeta_0\rangle_\Hf\,\zeta_0^{\otimes(q-1)}.
\end{eqnarray*}
Inserting this formula into~\eqref{eq:partial-res-nabla-Ip}, we find
\[\nabla^\st I_p(\zeta^{\otimes p})=-pI_p(\nabla\zeta\widetilde\otimes\zeta^{\otimes(p-1)})-p(p-1)I_{p-2}(\nabla\zeta\widetilde\otimes_1\zeta^{\otimes(p-1)}).\]
Since $\langle\zeta,\nabla\zeta\rangle_\Hf=\langle K\zeta,\nabla K\zeta\rangle_{\Ld^2(\R^d)}=0$ for real-valued $\zeta$, the second right-hand side term vanishes and we are led to
\[\nabla^\st I_p(\zeta^{\otimes p})=-pI_p(\nabla\zeta\widetilde\otimes\zeta^{\otimes(p-1)})=\Op(\nabla)I_p(\zeta^{\otimes p}).\]
In addition, this formula ensures that $\nabla^\st$ preserves the chaos decomposition.
\end{proof}

Given a self-adjoint operator $T$ on $\Ld^2(\R^d)$, the operator $\Op(T)$ on $\Ld^2(\Omega)$ is also self-adjoint and we may consider the corresponding unitary $C_0$-groups.
If $iT$ preserves the real part, then the group $\{e^{it\Op(T)}\}_{t\in\R}$ on $\Ld^2(\Omega)$ is shown to admit an explicit description.

\begin{lem}\label{lem:explic-eitT}
Let $T$ be essentially self-adjoint on $C^\infty_c(\R^d)$, and assume that the subset of real-valued functions $\Ld^2(\R^d;\R)$ is invariant under $\{e^{itT}\}_{t\in\R}$.
Then the operator $\Op(T)$ generates a unitary $C_0$-group $\{e^{it\Op(T)}\}_{t\in\R}$ on $\Ld^2(\Omega)$, which has the following explicit action: for all $\phi\in\Rc(\Omega)$, say $\phi=g(\Vc_0(\zeta_1),\ldots,\Vc_0(\zeta_n))$,
\[e^{it\Op(T)}\phi\,=\,g\big(\Vc_0(K^{-1}e^{-itT}K\zeta_1),\ldots,\Vc_0(K^{-1}e^{-itT}K\zeta_n)\big).\]
In particular, this entails that $e^{it\Op(T)}$ is multiplicative, that is, for all $\phi,\psi\in\Rc(\Omega)$,
\[e^{it\Op(T)}(\phi\psi)=(e^{it\Op(T)}\phi)(e^{it\Op(T)}\psi),\]
which implies that $\Op(T)$ is a derivation, that is, for all $\phi,\psi\in\Rc(\Omega)$,
\[\Op(T)(\phi \psi)=\psi\Op(T)\phi+\phi\Op(T)\psi.\qedhere\]
\end{lem}

\begin{proof}
Denote by $\{\widetilde U_t\}_{t\in\R}$ the group of operators defined on $\Rc(\Omega)$ as in the statement: for all $\phi\in\Rc(\Omega)$, say $\phi=g(\Vc_0(\zeta_1),\ldots,\Vc_0(\zeta_n))$ with $n\ge1$ and $\zeta_1,\ldots,\zeta_n\in C^\infty_c(\R^d;\R)$,
\[\widetilde U_t\phi\,:=\,g\big(\Vc_0(K^{-1}e^{-itT}K\zeta_1),\ldots,\Vc_0(K^{-1}e^{-itT}K\zeta_n)\big),\]
hence in particular, for all $\zeta\in C^\infty_c(\R^d;\R)$,
\begin{equation}\label{eq:redef-Ut-til}
\widetilde U_tI_p(\zeta^{\otimes p})=\widetilde U_tH_p(\Vc_0(\zeta))=H_p\big(\Vc_0(K^{-1}e^{-itT}K\zeta)\big)=I_p\big((K^{-1}e^{-itT}K\zeta)^{\otimes p}\big).
\end{equation}
This is well-defined since $e^{-itT}$ is assumed to preserve $\Ld^2(\R^d;\R)$.
(Note that $K^{-1}e^{-itT}K\zeta$ may of course no longer have any representative in $C^\infty_c(\R^d;\R)$ in its equivalence class in~$\Hf$.)
Noting that
\begin{multline*}
\big\langle\Vc_0(K^{-1}e^{-itT}K\zeta_j),\Vc_0(K^{-1}e^{-itT}K\zeta_l)\big\rangle_{\Ld^2(\Omega)}\,=\,\big\langle K^{-1}e^{-itT}K\zeta_j,K^{-1}e^{-itT}K\zeta_l\big\rangle_{\Hf}\\
\,=\,\big\langle e^{-itT}K\zeta_j,e^{-itT}K\zeta_l\big\rangle_{\Ld^2(\R^d)}\,=\,\langle K\zeta_j,K\zeta_l\rangle_{\Ld^2(\R^d)}\,=\,\langle \zeta_j,\zeta_l\rangle_{\Hf}\,=\,\langle\Vc_0(\zeta_j),\Vc_0(\zeta_l)\rangle_{\Ld^2(\Omega)},
\end{multline*}
and using again the assumption that $e^{-itT}$ preserves $\Ld^2(\R^d;\R)$,
we deduce that the (Gaussian) law of $\big(\Vc_0(K^{-1}e^{-itT}K\zeta_1),\ldots,\Vc_0(K^{-1}e^{-itT}K\zeta_n)\big)$ is invariant with respect to~$t$, hence for all $\phi\in\Rc(\Omega)$ and $t\in\R$,
\[\|\widetilde U_t\phi\|_{\Ld^2(\Omega)}=\|\phi\|_{\Ld^2(\Omega)}.\]
This allows to extend $\{\widetilde U_t\}_{t\in\R}$ by density as a unitary group on $\Ld^2(\Omega)$.
In addition, as $\{e^{-itT}\}_{t\in\R}$ is strongly continuous on $\Ld^2(\R^d)$, it is easily deduced that $\{\widetilde U_t\}_{t\in\R}$ is strongly continuous on $\Ld^2(\Omega)$. We denote by $i\widetilde T$ its skew-adjoint generator on $\Ld^2(\Omega)$.
Differentiating~\eqref{eq:redef-Ut-til} with respect to $t$ shows that the domain of $\widetilde{T}$ contains $\cal R(\Omega)$ and that $\widetilde T=\Op(T)$ on $\cal R(\Omega)$. Since $T$ is essentially self-adjoint on $C^\infty_c(\R^d)$, $\Op(T)$ is essentially self-adjoint on $\cal R(\Omega)$, and we conclude $\Op(T)=\widetilde T$, hence $\widetilde U_t=e^{it\Op(T)}$.
\end{proof}

In view of the application to Mourre's theory for Schrödinger operators, cf.~Section~\ref{sec:Mourre-intro}, we recall the definition of the unitary $C_0$-group of dilations $U_tg:=e^{td/2}g(e^t\cdot)$ on $\Ld^2(\R^d)$, and its generator $iA:=\frac1{2}\big(x\cdot\nabla+\nabla\cdot x\big)$.
We then define the self-adjoint operator
\begin{equation}\label{eq:defin-Ast}
A^\st:=\Op(A),\qquad\text{on $\Ld^2(\Omega)$,}
\end{equation}
and the associated unitary $C_0$-group $\{U_t^\st\}_{t\in\R}$ given by $U_t^\st:=e^{itA^\st}$.
Due to Lemma~\ref{lem:explic-eitT}, this satisfies, for all $\phi\in\Rc(\Omega)$, say $\phi=g(\Vc_0(\zeta_1),\ldots,\Vc_0(\zeta_n))$,
\begin{equation}\label{eq:Utst-action}
U_t^\st\phi\,=\,g\big(\Vc_0(K^{-1}U_{-t}K\zeta_1),\ldots,\Vc_0(K^{-1}U_{-t}K\zeta_n)\big),
\end{equation}
which entails in particular that the spaces $H^s(\Omega)$ are invariant under $\{U_t^{\st}\}_{t\in\R}$ for all $s\ge0$.

\subsection{Chaos decomposition of fibered operators}
While the unperturbed operators $\{H_{k,0}^\st\}_k$ preserve the chaos decomposition, cf.~Lemma~\ref{lem:nabla-Op}, we show that the random potential amounts to shifting the chaoses, thus playing the role of annihilation and creation operators on the Fock space decomposition.
This structure of the perturbed operators~$\{H_{k,\lambda}^\st\}_k$ can be viewed as drawing some surprising link between random Schrödinger operators and multi-particle quantum systems.

\begin{lem}\label{lem:V-decomp}
Assume for simplicity that $V=V_0$ is itself a Gaussian field.
Via the Wiener chaos decomposition~\eqref{eq:wiener-chaos}, for all $k\in\R^d$, the perturbed fibered operator $H_{k,\lambda}^\st=H_{k,0}^\st+\lambda V$ on $\Ld^2(\Omega)$ is unitarily equivalent to $T_k+\lambda(a+a^*)$ in terms of
\[T_k:=\bigoplus_{p=0}^\infty T_{k,p},\qquad a:=\bigoplus_{p=0}^\infty a_p,\qquad a^*:=\bigoplus_{p=0}^\infty a^*_p,\qquad\text{on}~~\bigoplus_{p=0}^\infty\Ld^2_\Sym((\R^d)^p),\]
in terms of
\[T_{k,p}:=-\Op^\circ_p(\nabla+\tfrac{ik}p)\cdot\Op^\circ_p(\nabla+\tfrac{ik}p)-|k|^2,\qquad\text{on}~~\Ld^2_\Sym((\R^d)^p),\]
and in terms of the annihilation and creation operators
\begin{eqnarray*}
a_p&:&\Ld^2_\Sym((\R^d)^{p+1})~\to~\Ld^2_\Sym((\R^d)^p),\\
a_p^*&:&\Ld^2_\Sym((\R^d)^p)~\to~\Ld^2_\Sym((\R^d)^{p+1}),
\end{eqnarray*}
which are defined as follows for all $u_{p+1}\in\Ld^2_\Sym((\R^d)^{p+1})$ and $u_{p}\in\Ld^2_\Sym((\R^d)^{p})$,
\begin{eqnarray*}
(a_pu_{p+1})(x_1,\ldots,x_p)&:=&\sum_{j=1}^{p+1}\int_{\R^d}\Cc_0^\circ(z)\,u_{p+1}(x_1,\ldots,x_{j-1},z,x_{j},\ldots,x_p)\,dz,\\
(a_p^*u_{p})(x_1,\ldots,x_{p+1})&:=&\frac1{p+1}\sum_{j=1}^{p+1}\Cc_0^\circ(x_j)\,u_{p}(x_1,\ldots,x_{j-1},x_{j+1},\ldots,x_{p+1}).
\end{eqnarray*}
For all $p$, the operators $a_p$ and $a_p^*$ are bounded and adjoint, with
\[\|a_p\|\,\le\,(p+1)^\frac12\,\expecm{|V_0|^2},\qquad\|a_p^*\|\,\le\,(p+1)^\frac12\,\expecm{|V_0|^2},\]
and thus the operators $a$ and $a^*$ are $\Lc^\frac12$-bounded and are adjoint in particular on $\Dm^{1,2}(\Omega)$. In addition, they satisfy the commutator relation $[a,a^*]=\expec{|V_0|^2}$ on $\Dm^{2,2}(\Omega)$.
\end{lem}

\begin{proof}
Given $\phi\in\Ld^2(\Omega)$, we consider its Wiener chaos expansion $\phi=\sum_{p=0}^\infty J_p(u_p)$, with $u_p \in\Ld^2_\Sym((\R^d)^p)$. In the case $V=V_0$, we can write $V=I_1(\delta_0)$, hence
\begin{equation*}
H_{k,\lambda}^\st\phi\,=\,\sum_{p=0}^\infty H_{k,0}J_p(u_p)+\lambda\sum_{p=0}^\infty I_1(\delta_0)J_p(u_p),
\end{equation*}
and the conclusion follows from Lemma~\ref{lem:nabla-Op}
and the product formula~\eqref{eq:prod-form}.
Finally, a direct computation ensures that $a$ and $a^*$ satisfy the usual properties of annihilation and creation operators,
\begin{gather*}
\langle u_{p},a_{p}u_{p+1}\rangle_{\Ld^2_\Sym((\R^d)^{p})}=\langle a_{p}^*u_{p},u_{p+1}\rangle_{\Ld^2_\Sym((\R^d)^{p+1})},\\
(a_{p}a_{p}^*-a_{p-1}^*a_{p-1})u_p=\Big(\int_{\R^d}(\Cc_0^\circ)^2\Big)u_p= \Cc_0(0)\,u_p= \expec{|V_0|^2}u_p.\qedhere
\end{gather*}
\end{proof}

\subsection{Some Mourre relations on $\Ld^2(\Omega)$}\label{sec:Mourre-rel}
Drawing inspiration from the construction of Mourre conjugates for Schrödinger operators on $\Ld^2(\R^d)$, cf.~Section~\ref{sec:Mourre-intro}, we show in item~(i) below that the generator of dilations $A^\st$ on $\Ld^2(\Omega)$ as constructed in Section~\ref{sec:oper-L2} is a conjugate for the stationary Laplacian $-\triangle^\st$.
Nevertheless, item~(iii) indicates that the perturbation by the random potential $V$ is never compatible in the sense of Mourre's theory for this conjugate operator, which prohibits to deduce any Mourre relation for perturbed operators of the form $-\triangle^\st+\lambda V$.
In spite of this, the incompatibility is shown to be comparable to the lack of boundedness of the underlying Gaussian field in the sense that it is bounded on any fixed chaos and $\Lc^{1/2}$-bounded on $\Ld^2(\Omega)$.
Finally, in item~(iv), we show that the action of the random potential~$V$ as described in Lemma~\ref{lem:V-decomp} on the Fock space allows to associate a natural conjugate. In other words, the stationary Laplacian~$-\triangle^\st$ describes diffusion on each chaos and the random potential $V$ describes shifts between chaoses: the transport properties of both parts are well understood and natural conjugates can be constructed for both, cf.~items~(i) and~(iv), but the construction of a conjugate for \mbox{$-\triangle^\st+\lambda V$} appears particularly difficult and is left as an open problem. In a semiclassical perspective, this is related to the construction of escape functions for the random acceleration model~\cite{Kesten-Papa-80,Komorowski-Ryzhik-06}.

\begin{prop}[Some Mourre relations]\label{prop:commut}$ $
\begin{enumerate}[(i)]
\item \emph{Conjugate operator for $-\triangle^\st$:}\\
The generator of dilations $A^\st$ on $\Ld^2(\Omega)$, cf.~\eqref{eq:defin-Ast}, satisfies
\[[-\triangle^\st,\tfrac1iA^\st]=2\,(-\triangle^\st),\]
and the domain $H^2(\Omega)$ of $-\triangle^\st$ is invariant under $\{U_t^\st=e^{itA^\st}\}_{t\in\R}$.
\item \emph{Conjugate operators for $\frac1i\nabla^\st_1$:}
\[[i\nabla_1^\st,\tfrac1iA^\st]=i\nabla_1^\st,\qquad
[i\nabla_1^\st,\Op(ix_1)]=\Lc.\]
\item \emph{Incompatibility of the perturbation:}\\
The commutator $[V,iA^\st]$ is well-defined and essentially self-adjoint on $\Rc(\Omega)$, but only $\Lc^{1/2}$-bounded provided $A\Cc_0^\circ\in\Ld^2(\R^d)$.
If~$V=V_0$ is itself Gaussian, then similarly $[V,\Op(ix_1)]$ is well-defined and essentially self-adjoint on $\Rc(\Omega)$, but only $\Lc^{1/2}$-bounded provided $x_1\Cc_0^\circ\in\Ld^2(\R^d)$.
\item \emph{Conjugate operator for the perturbation:}\\
If $V=V_0$ is itself Gaussian, then
there holds
\[[V,[\Lc,V]]=2\,\expec{|V|^2}.\qedhere\]
\end{enumerate}
\end{prop}

\begin{proof}
Since the operators $-\triangle^\st$, $i\nabla_1^\st$, $\Lc$, $A^\st$, $\Op(x_1)$ are essentially self-adjoint on $\Rc(\Omega)$, and since $\Rc(\Omega)$ is invariant under these operators, the commutators $[-\triangle^\st,\frac1iA^\st]$, $[i\nabla_1^\st,\frac1iA^\st]$, $[i\nabla_1^\st,\Op(ix_1)]$, $[V,iA^\st]$, $[V,\Op(ix_1)]$ are clearly well-defined on~$\Rc(\Omega)$ and are explicitly computed below on that linear subspace.
We split the proof into five steps.

\medskip
\step1 Proof of~(i).\\
Given $p\ge1$ and $u_p\in C^\infty_c(\R^d)^{\odot p}$, recalling the notation $U_t^\st:=e^{itA^\st}$, Lemmas~\ref{lem:nabla-Op} and~\ref{lem:explic-eitT} lead to
\[U_{-t}^\st \nabla^\st U_t^\st J_p(u_p)\,=\,J_p\big(U_{t}^{\otimes p}\Op_p^\circ(-\nabla) U_{-t}^{\otimes p} u_p\big),\]
and hence, by definition of $\Op_p^\circ(\nabla)$ and of $U_t$,
\[U_{-t}^\st\nabla^\st U_t^\st J_p(u_p)
\,=\,e^{-t}J_p\big(\!\Op_p^\circ(-\nabla)u_p\big)
\,=\,e^{-t}\nabla^\st J_p(u_p),\]
so that differentiating in $t$ yields
\[[\nabla^\st,iA^\st]J_p(u_p)=-\nabla^\st J_p(u_p),\]
and similarly
\[[\triangle^\st,iA^\st] J_p(u_p)=-2\triangle^\st J_p(u_p).\]

\medskip
\step2 Proof of~(ii).\\
The computation of the commutator $[i\nabla_1^\st,\frac1iA^\st]$ follows from Step~1, and it remains to compute the other one.
For $u_p\in C^\infty_c(\R^d)^{\odot p}$, Lemma~\ref{lem:nabla-Op} yields
\[[i\nabla^\st_1,\Op(ix_1)]J_p(u_p)=J_p\big([\Op_p^\circ(\nabla_1),\Op_p^\circ(x_1)]u_p\big),\]
while the definition of $\Op_p^\circ$ leads to $[\Op_p^\circ(\nabla_1),\Op_p^\circ(x_1)]=p$, hence
\[[i\nabla^\st_1,\Op(ix_1)]J_p(u_p)=pJ_p(u_p)=\Lc J_p(u_p).\]

\medskip
\step3 Proof of~(iv).\\
Given $V=V_0$, for $p\ge1$ and $u_p\in\Ld^2_\Sym((\R^d)^p)$, Lemma~\ref{lem:V-decomp} yields
\begin{equation}\label{eq:V-decomp-re}
VJ_p(u_p)=J_{p+1}(a_{p}^*u_p)+J_{p-1}(a_{p-1}u_p),
\end{equation}
so that the commutator with $\Lc$ takes the form,
\[[\Lc,V]J_p(u_p)=J_{p+1}(a_{p}^*u_p)-J_{p-1}(a_{p-1}u_p).\]
A direct computation of the commutator of these two operators then yields after straightforward simplifications,
\[[V,[\Lc,V]]J_p(u_p)=2\,J_p\big((a_{p}a_{p}^*-a_{p-1}^*a_{p-1})u_p\big),\]
hence, in view of Lemma~\ref{lem:V-decomp},
\[[V,[\Lc,V]]J_p(u_p)=2\,\expec{|V_0|^2}\!J_p(u_p).\]

\medskip
\step4 Proof of~(iii) for $[V,iA^\st]$.\\
In view of~\eqref{eq:Utst-action}, with $V=b(V_0)=b(J_1(K\delta_0))$, we find
\begin{equation*}
U_{-t}^\st VU_t^\st\phi=b\big(J_1(U_{t}K\delta_0)\big)\,\phi,
\end{equation*}
and differentiating in $t$ yields
\[[V,iA^\st]=b'(V_0)V_0',\qquad V_0':=J_1\big((iA)K\delta_0\big).\]
Noting that $iA$ preserves the real part and that
\[\langle V_0,V_0'\rangle_{\Ld^2(\Omega)}=\big\langle J_1(K\delta_0),J_1\big((iA)K\delta_0\big)\big\rangle_{\Ld^2(\Omega)}=\big\langle K\delta_0,(iA)K\delta_0\big\rangle_{\Ld^2(\R^d)}=0,\]
we deduce that $V_0$ and $V_0'$ are independent Gaussian random variables. Further note that $V_0'$ cannot be degenerate: indeed, by definition of $A$ as generator of dilations, $AK\delta_0$ can only vanish if $\Cc_0(x)=K\delta_0(x)\propto|x|^{-d/2}$, which is not compatible with $\Cc_0(0)=\expec{|V_0|^2}<\infty$. We may then deduce
\[\dfrac{\|[V,iA^\st](V'_0)^p\|}{\|(V'_0)^p\|_{\Ld^2(\Omega)}}=\frac{\|b'(V_0)(V'_0)^{p+1}\|_{\Ld^2(\Omega)}}{\|(V'_0)^p\|_{\Ld^2(\Omega)}}=\|b'(V_0)\|_{\Ld^2(\Omega)}\|V_0'\|_{\Ld^2(\Omega)}\sqrt{2p+1}\xrightarrow{p\uparrow\infty}\infty,\]
proving that $[V,iA^\st]$ is unbounded.

\medskip\noindent
Next, we show that $[V,iA^\st]$ is $\Lc^{1/2}$-bounded. For that purpose, for $\phi_p\in C^\infty_c(\R^d)^{\otimes p}$, we use the product formula~\eqref{eq:prod-form} to compute
\begin{equation*}
V_0'I_p(\phi_p)=I_{p+1}\big((K^{-1}(iA)K\delta_0)\widetilde\otimes\phi_p\big)+pI_{p-1}\big((K^{-1}(iA)K\delta_0)\widetilde\otimes_1\phi_p\big),
\end{equation*}
where the isomorphism~\eqref{eq:isom-L2-Wiener} allows to compute
\begin{eqnarray*}
\big\|I_{p+1}\big((K^{-1}(iA)K\delta_0)\widetilde\otimes\phi_p\big)\big\|_{\Ld^2(\Omega)}&=&\sqrt{(p+1)!}\,\big\|(K^{-1}(iA)K\delta_0)\widetilde\otimes\phi_p\big\|_{\Hf^{\otimes (p+1)}}\\
&\le&\sqrt{(p+1)!}\,\|K^{-1}(iA)K\delta_0\|_{\Hf}\|\phi_p\|_{\Hf^{\otimes p}}\\
&=&\sqrt{p+1}\,\|K^{-1}(iA)K\delta_0\|_{\Hf}\|I_p(\phi_p)\|_{\Ld^2(\Omega)},
\end{eqnarray*}
and similarly
\begin{eqnarray*}
\big\|pI_{p-1}\big((K^{-1}(iA)K\delta_0)\widetilde\otimes_1\phi_p\big)\big\|_{\Ld^2(\Omega)}&=&p\sqrt{(p-1)!}\,\big\|(K^{-1}(iA)K\delta_0)\widetilde\otimes_1\phi_p\big\|_{\Hf^{\otimes(p-1)}}\\
&\le& p\sqrt{(p-1)!}\,\|K^{-1}(iA)K\delta_0\|_{\Hf}\|\phi_p\|_{\Hf^{\otimes p}}\\
&=&\sqrt p\,\|K^{-1}(iA)K\delta_0\|_{\Hf}\|I_p(\phi_p)\|_{\Ld^2(\Omega)}.
\end{eqnarray*}
For $\phi\in\Ld^2(\Omega)$ in a finite union of chaoses, the Wiener chaos decomposition~\eqref{eq:wiener-chaos} then leads to
\begin{eqnarray*}
\|[V,iA^\st]\phi\|_{\Ld^2(\Omega)}\,\lesssim\, \|V_0'\phi\|_{\Ld^2(\Omega)}&\le&\|K^{-1}(iA)K\delta_0\|_{\Hf}\bigg(\sum_{p=0}^\infty (2p+1)\|I_p(\phi_p)\|_{\Ld^2(\Omega)}^2\bigg)^\frac12\\
&\lesssim&\|A\Cc_0^\circ\|_{\Ld^2(\R^d)}\|(\Lc+1)^{\frac12}\phi\|_{\Ld^2(\Omega)},
\end{eqnarray*}
and the claim follows.

\medskip
\step5 Proof of~(iii) for $[V,\Op(ix_1)]$.\\
In view of the Fock space decomposition~\eqref{eq:V-decomp-re} for multiplication by $V=V_0$, cf.~Lemma~\ref{lem:V-decomp}, a direct computation yields for all $p\ge1$ and $\phi_p\in C^\infty_c(\R^d)^{\otimes p}$,
\[[V,\Op(ix_1)]I_p(\phi_p)=-I_{p+1}\big((K^{-1}ix_1K\delta_0)\widetilde\otimes\phi_p\big)+pI_{p-1}\big((K^{-1}ix_1K\delta_0)\widetilde\otimes_1\phi_p\big),\]
and the conclusion follows as in~Step~4.
\end{proof}

\subsection{Mourre relations for fibered operators}\label{sec:Mourre-fiber}
This section is devoted to the construction of conjugates for the unperturbed fibered operators $\{H_{k,0}^\st\}_k$.
This appears to be surprisingly more involved than for $k=0$, as the group of dilations $\{U_t^\st\}_{t\in\R}$ is no longer adapted and must be suitably deformed.
Noting that bounds on iterated commutators are obtained similarly and that the dense subspaces $\Pc(\Omega)$ and $\Rc(\Omega)$ are exchangeable, the conclusion of Theorem~\ref{th:Mourre-pert} directly follows upon truncation.

\begin{theor}[Mourre relations for fibered operators]\label{th:Mourre}$ $
\begin{enumerate}[(i)]
\item \emph{Conjugate operator for $\{H_{k,0}^\st\}_k$:}\\
For all $k\in\R^d$, there exists a self-adjoint operator $C_k^\st$ on $\Ld^2(\Omega)$, essentially self-adjoint on $\Rc(\Omega)$, such that the commutator $[H_{k,0}^\st,\tfrac1iC_k^\st]$ is well-defined and essentially self-adjoint on $\Rc(\Omega)$, is~$H_{k,0}^\st$-bounded, and satisfies
\[\qquad[H_{k,0}^\st,\tfrac1iC_k^\st]\ge \Pi\big(H_{k,0}^\st+\tfrac34|k|^2\big)\Pi.\]
Hence, the fibered operator $H_{k,0}^\st$ satisfies a Mourre relation on $J_\e:=[\e-\frac34|k|^2,\infty)$ with respect to $C_k^\st$, for all $\e>0$,
\[\qquad\mathds1_{J_\e}(H_{k,0}^\st)\,[H_{k,0}^\st,\tfrac1iC_k^\st]\,\mathds1_{J_\e}(H_{k,0}^\st)\,\ge\,\e\,\mathds1_{J_\e}(H_{k,0}^\st)-\tfrac34|k|^2\E.\]
In addition, the domain $H^2(\Omega)$ of $H_{k,0}^\st$ is invariant under $\{e^{it C_k^\st}\}_{t\in\R}$.
\item
\emph{Incompatibility of the perturbation:}\\
If~$V=V_0$ is itself Gaussian,
then the commutator $[V,iC_k^\st]$ is well-defined and essentially self-adjoint on $\Rc(\Omega)$, but only $\Lc^{1/2}$-bounded provided that $A\Cc_0^\circ,x\Cc_0^\circ\in\Ld^2(\R^d)$.
\end{enumerate}
In addition,
the full commutator $[H_{k,\lambda}^\st,\frac1iC_k^\st]$ is well-defined and essentially self-adjoint on~$\Rc(\Omega)$ for all~$\lambda$, but it only satisfies the lower bound
\[[H_{k,\lambda}^\st,\tfrac1iC_k^\st]\ge H_{k,\lambda}^\st+\tfrac34|k|^2-C\lambda\Lc^\frac12-\tfrac34|k|^2\E,\]
which does not yield any Mourre relation.
\end{theor}

Before turning to the proof, we briefly underline the difficulty and explain the idea behind the construction.
Proposition~\ref{prop:commut}~(i)--(ii) leads to
\[[H_{k,0}^\st,\tfrac1iA^\st]=2H_{k/2,0}^\st,\]
which shows that the generator of dilations $A^\st$ should be properly modified.
Again drawing inspiration from the situation on the physical space, noting that on~$\Ld^2(\R^d)$ there holds
\begin{gather*}
[-(\nabla+ik)\cdot(\nabla+ik),iA_k]=2\big(-(\nabla+ik)\cdot(\nabla+ik)\big),
\end{gather*}
with $A_k:=\frac1{2i}\big(x\cdot(\nabla+ik)+(\nabla+ik)\cdot x\big)$,
we consider
\begin{equation}\label{eq:Akst}
A_k^\st:=\Op(A_k)=A^\st+\Op(k\cdot x),
\end{equation}
and a similar computation as in Proposition~\ref{prop:commut}~(i)--(ii) yields on $\Rc(\Omega)$,
\begin{gather*}
[H_{k,0}^\st,\tfrac1iA_k^\st]=2\big(-\triangle^\st-(1+\Lc)ik\cdot\nabla^\st+|k|^2\Lc\big).
\end{gather*}
(Recall that $\Lc$ and $\nabla^\st$ commute, cf.~Lemma~\ref{lem:nabla-Op}.)
In order to counter the apparition of factors $\Lc$ in this relation and
obtain a proper Mourre relation, a further modification of~$A_k^\st$ is thus needed.
More precisely, in the definition~\eqref{eq:Akst} of $A_k^\st$, the generator of dilations $A^\st$ is a suitable conjugate for the stationary Laplacian~$-\triangle^\st$, while $\Op(k\cdot x)$ is supposed to take into account the additional first-order contribution~$-2ik\cdot \nabla^\st$ in the fibered operator~$H_{k,0}^\st:=-\triangle^\st-2ik\cdot\nabla^\st$. The core of the problem then lies in the factor~$\Lc$ that appears in the commutator
\begin{equation}\label{eq:commut-L-too}
[ik\cdot\nabla^\st,i\Op(k\cdot x)]=|k|^2\Lc,
\end{equation}
which is related to the infinite dimensionality of the probability space.
The simplest way to solve this problem would be defining
\begin{equation}\label{eq:def-Bkst}
B_k^\st:=A^\st+\Lc^{-1/2}\Op(k\cdot x)\Lc^{-1/2},
\end{equation}
where we recall that $\Lc^{-1}:=\Pi\Lc^{-1}\Pi$ denotes the pseudo-inverse of $\Lc$. This indeed leads to the desired Mourre relation,
\begin{equation}\label{eq:Mourre-fibr-pre}
[H_{k,0}^\st,\tfrac1iB_k^\st]=2(H_{k,0}^\st+|k|^2)-2|k|^2\E.
\end{equation}
However, the perturbation $V$ behaves particularly badly with respect to this conjugate operator $B_k^\st$ in the sense that the commutator $[V,iB_k^\st]$ is not even bounded when restricted to any fixed Wiener chaos, thus excluding any meaningful use of such a relation.
While on the $p$th Wiener chaos $\Hc_p$ the operator $\Op(k\cdot x)$ amounts to the sum $\sum_{j=1}^pk\cdot x_j$, the choice~\eqref{eq:def-Bkst} consists of rather considering the algebraic mean $\frac1p\sum_{j=1}^pk\cdot x_j$. Another possible choice to avoid the factor $\Lc$ in the commutator~\eqref{eq:commut-L-too} is to use an $\ell^\infty$-norm of $\{k\cdot x_j\}_{j=1}^p$. We show in the following paragraphs that the latter choice has all the desired properties claimed in Theorem~\ref{th:Mourre}: it still yields a similar Mourre relation as in~\eqref{eq:Mourre-fibr-pre} and its commutator with the perturbation $V$ is $\Lc^{1/2}$-bounded.

\medskip
We construct the desired conjugate operator $C_k^\st$ via its action on the Fock space decomposition~\eqref{eq:wiener-chaos}.
As we are concerned with the suitable treatment of the first-order operator $ik\cdot\nabla^\st$, that is, the stationary derivative in the direction $k$, we set $z:=k\cdot x$ and first focus on the case of dimension $d=1$.
For all $p\ge0$, define the function $m_p:\R^p\to\R$,
\begin{eqnarray*}
m_p(z_1,\ldots,z_p)&:=&\textstyle\big(\!\max_j|z_j|\big)\sgn r_p(z_1,\ldots,z_p),\\
r_p(z_1,\ldots,z_p)&:=&\textstyle\max_j z_j+\min_j z_j,
\end{eqnarray*}
with the convention $\sgn(0)=0$. This function is clearly symmetric with respect to the variables $z_1,\dots,z_p$ and has the following main properties.

\begin{lem}
For all $p\ge0$, the function $m_p$ is well-defined and is continuous on $\R^p\setminus S_p$, where $S_p$ denotes the hypersurface
\[S_p:=r_p^{-1}\{0\}=\big\{z\in\R^p:\exists j\neq k\text{ such that $z_j=-z_k$ and $|z_j|=\textstyle\max_l|z_l|$}\big\}.\]
In addition, there exists a continuous function $g_p:S_p\to[0,\infty)$ such that
\begin{equation}\label{eq:comm-mp0}
[\Op_p^\circ(\partial),m_p]\,=\,1+g_p\,\delta_{S_p}\,\ge\,1.\qedhere
\end{equation}
\end{lem}

\begin{proof}
The continuity of $m_p$ is clear outside the zero locus $S_p$ of $r_p$, and we turn to the second part of the statement.
On $\R^p\setminus S_p$ there holds $m_p(z_1,\ldots,z_p)=z_j$ with $|z_j|=\max_i|z_i|$, hence $\sum_{j=1}^p\partial_jm_p=1$.
It remains to examine the jump of $m_p$ on $S_p$.
We claim that every line directed by the vector $(1,\dots,1)$ in $\R^p$ intersects the hypersurface~$S_p$ at a single point, and this would yield the conclusion.
Indeed, given a point $z:=(z_1,\dots,z_p)\in\R^p$, say $z_1=\min_j z_j$ and $z_2=\max_j z_j$, we can write $z=z'+s(1,\ldots,1)$ with $s=\frac12(z_1+z_2)$ and $z':=(z_1-s,\dots,z_p-s)\in S_p$, and $z'+t(1,\ldots,1)$ belongs to $S_p$ only if $t=s$.
\end{proof}

Next, in order to get a proper Mourre relation on the Fock space, we regularize the functions $\{m_p\}_p$ so as to
replace the Dirac part in~\eqref{eq:comm-mp0} by a positive bump function that is $p$-uniformly bounded on $\R^p$. For that purpose, it is not enough to regularize the sign function in the definition of $m_p$
in a fixed neighborhood of $S_p$, as the derivative would still produce an unbounded term due to the multiplication by $\max_j|z_j|$.
A suitable choice of the regularization is rather defined as follows.
First rewrite 
\[m_p(z_1,\ldots,z_p)=\textstyle\frac{1}{2}(\max_j z_j+\min_j z_j)+\frac12(\max_j z_j-\min_j z_j)\sgn(\max z_j+\min_j z_j),\]
where only the last sign function needs to be regularized.
Choose a smooth odd function $\chi:\R\to[-1,1]$ such that $\chi(s)=-1$ for $s\le -1$, $\chi(s)=1$ for $s\ge1$, $0\le\chi'\le2$ pointwise, $\chi(s)\le s$ for $-1\le s\le 0$, and $\chi(s)\ge s$ for $0\le s\le 1$.
We then set
\begin{equation*}
\widetilde m_p(z_1,\ldots,z_p):=\textstyle\frac12(\max_j z_j+\min_j z_j)+\frac12(\max_j z_j-\min_j z_j)\,\chi\big(\frac{\max_j z_j+\min_j z_j}{\max_jz_j-\min_jz_j}\big),
\end{equation*}
which is globally well-defined and continuous.
Note that
\begin{gather}
1\le\sum_{j=1}^p\partial_j\widetilde m_p\le3,\qquad\bigg|\Big(\sum_{j=1}^p\partial_{j}\Big)^r\widetilde m_p\bigg|\lesssim_{\chi,r}1,~~\text{for all $r\ge0$},\label{eq:prepre-mourre-mp}\\
|\widetilde m_p(z_1,\ldots,z_p)|\le |m_p(z_1,\dots,z_p)|=\textstyle\max_j|z_j|.\label{eq:prepre-trivial-bound}
\end{gather}
We also establish the following important property.
\begin{lem}\label{prop:m-function}
For all $p\ge0$, there holds for all $z,z_1,\dots,z_p\in\R$,
\begin{displaymath}
\big|\widetilde m_{p+1}(z,z_1,\dots,z_p)-\widetilde m_p(z_1,\dots,z_p)\big|\leq 2|z|.\qedhere
\end{displaymath}
\end{lem}

\begin{proof}
The conclusion follows from~\eqref{eq:prepre-trivial-bound} if $\max_j|z_j|\le |z|$, so that we can henceforth assume $\max_j|z_j|> |z|$.
By symmetry we can assume $z_1=\min_jz_j$ and $z_2=\max_jz_j$.
In the case $z_1\le z\le z_2$, we find
\[\widetilde m_{p+1}(z,z_1,\dots,z_p)=\widetilde m_p(z_1,\dots,z_p),\]
and the conclusion follows.
It remains to treat the case $z_1\le z_2\le z$, while the symmetric case $z\le z_1\le z_2$ is similar.
Given $z_1\le z_2\le z$, the assumption~$\max_j|z_j|>|z|$ implies $z_1<-|z|$. As $\frac{z_2+z_1}{z_2-z_1}\le \frac{z+z_1}{z-z_1}$, we compute
\begin{eqnarray*}
\lefteqn{\big|\widetilde m_{p+1}(z,z_1,\ldots,z_p)-\widetilde m_p(z_1,\ldots,z_p)\big|}\\
&\le&\textstyle\frac12(z-z_2)\Big(1+\chi\big(\frac{z_2+z_1}{z_2-z_1}\big)\Big)+\frac12(z-z_1)\Big(\chi\big(\frac{z+z_1}{z-z_1}\big)-\chi\big(\frac{z_2+z_1}{z_2-z_1}\big)\Big)\\
&\le&\textstyle\frac12(z-z_2)\Big(1+\chi\big(\frac{z_2+z_1}{z_2-z_1}\big)\Big)+\frac12(z-z_1)\Big(1+\chi\big(\frac{z+z_1}{z-z_1}\big)\Big).
\end{eqnarray*}
Noting that $\frac{z_2+z_1}{z_2-z_1}\le \frac{z+z_1}{z-z_1}\le0$ and that there holds $\chi(\frac{y+z_1}{y-z_1})=-1$ whenever $y\le0$,
the above becomes, in view of the properties of $\chi$,
\begin{eqnarray*}
\big|\widetilde m_{p+1}(z,z_1,\ldots,z_p)-\widetilde m_p(z_1,\ldots,z_p)\big|
&\le&\textstyle\frac12(z-z_2)\mathds1_{z_2\ge0}+\frac12(z-z_1)\big(1+\frac{z+z_1}{z-z_1}\big)\mathds1_{z\ge0}\\
&\le&\tfrac32|z|,
\end{eqnarray*}
as claimed.
\end{proof}

We now turn to the construction of the suitable conjugate operator for $ik\cdot \nabla^{\st}$. For all $p\ge0$, we define an operator $M_{k,p}$ on $\Ld^2_\Sym((\R^d)^p)$ as the multiplication by the function $(x_1,\ldots,x_p)\mapsto\widetilde m_p(k\cdot x_1,\ldots,k\cdot x_p)$, and we denote by $M_k=\bigoplus_{p=0}^\infty M_{k,p}$ the corresponding operator on the Fock space. Next, we define the operator $M_{k,p}^\st$ on the $p$th Wiener chaos~$\Hc_p$ by
\[M_{k,p}^\st J_p(u_p):=J_p(M_{k,p}u_p)=J_p\big(\widetilde m_p(k\cdot x_1,\ldots,k\cdot x_p)\,u_p\big),\]
and via the Wiener chaos decomposition~\eqref{eq:wiener-chaos} we set $M_k^\st:=\bigoplus_{p=0}^\infty M_{k,p}^\st$ on $\Ld^2(\Omega)$.
We then consider the following operator on $\Ld^2(\Omega)$,
\begin{equation}\label{eq:def-Ckst}
C_k^\st:=A^\st+\tfrac{1}{2}M_k^\st,
\end{equation}
which is clearly essentially self-adjoint on $\Rc(\Omega)$ given its action on Wiener chaoses; see also Lemma~\ref{lem:C_k^st} below.

\begin{rem}
The reader may wonder why this definition of $C_k^\st$ is chosen instead of $A^{\st}+M_k^\st$, which would seem more natural in view of~\eqref{eq:Akst}. The computation of the relevant commutators involves $[\nabla^\st,M_k^\st]$, hence the derivative $\sum_{j=1}^p\partial_j\widetilde m_p$, which in view of the regularization $\widetilde m_p$ of $m_p$ is not uniformly equal to $1$ but can vary in the whole interval $[1,3]$ (or at best in $[1,2+\delta]$ for some smaller $\delta>0$ if the cut-off function $\chi$ is chosen with $\chi'$ closer to $\mathds1_{[-1,1]}$). Due to this modification, symbols are deformed in the commutator computation, and the choice $A^\st+M_k^\st$ would fail at providing a Mourre relation close to~$0$. This is precisely corrected by the above choice~\eqref{eq:def-Ckst}.
\end{rem}

We first show that the operator $C_k^\st$ generates an explicit unitary $C_0$-group, which preserves $H^s(\Omega)$.

\begin{lem}\label{lem:C_k^st}
The operator $C_k^\st$ is essentially self-adjoint on $\Rc(\Omega)$ and its closure generates a unitary $C_0$-group $\{e^{itC_k^\st}\}_{t\in\R}$ on $\Ld^2(\Omega)$, which has the following explicit action on chaoses: for all $p\ge1$ and $u_p\in\Ld^2_\Sym((\R^d)^p)$,
\[e^{itC_k^\st}J_p(u_p)\,=\,J_p(U_{t,k,p}u_p),\]
where $U_{t,k,p}$ is defined by
\[(U_{t,k,p}u_p)(x_1,\ldots,x_p):=e^{t\frac {dp}2}\exp\Big(\frac i2\int_0^t \widetilde m_p(k\cdot e^sx_1,\ldots,k\cdot e^sx_p)\,ds\Big)\,u_p(e^tx_1,\ldots,e^tx_p).\]
In particular, in view of~\eqref{eq:prepre-mourre-mp}, for all $s\ge0$, the subspace $H^s(\Omega)$ is invariant under this group action $\{e^{itC_k^\st}\}_{t\in\R}$.
\end{lem}

\begin{proof}
In view of the chaos decomposition~\eqref{eq:wiener-chaos}, it suffices to check that 
for all $p\ge1$ the family $\{U_{t,k,p}\}_{t\in\R}$ defines a unitary $C_0$-group on $\Ld^2_\Sym((\R^d)^p)$ and that its self-adjoint generator is given by $C_{k,p}:=\Op_p(A)+\frac12M_{k,p}$ on $C^\infty_c(\R^d)^{\odot p}$.
First note that the family $\{U_{t,k,p}\}_{t\in\R}$ clearly defines a unitary group on $\Ld^2_\Sym((\R^d)^p)$. Next,
for all $\epsilon>0$, we decompose
\begin{multline*}
\frac1{i\e}\big(U_{\e,k,p}u_p-u_p)(x_1,\ldots,x_p)
\,=\,\frac1{i\e}\Big(e^{\e\frac{dp}2}u_p(e^\e x_1,\ldots,e^\e x_p)-u_p(x_1,\ldots,x_p)\Big)\\
+\frac1{i\e}\bigg(\exp\Big(\frac i2\int_0^\e \widetilde m_p(k\cdot e^sx_1,\ldots,k\cdot e^sx_p)\,ds\Big)-1\bigg)e^{\e\frac{dp}2}\,u_p(e^\e x_1,\ldots,e^\e x_p).
\end{multline*}
As $\e\downarrow0$, for $u_p\in C^\infty_c(\R^d)^{\odot p}$, the first right-hand side term converges to $\Op_p^\circ(A)u_p$ in~$\Ld^2_\Sym((\R^d)^p)$, while the second one converges to $\frac12M_{k,p}u_p$, and the claim easily follows.
\end{proof}

Next, we show that $C_k^\st$ is a conjugate for the fibered operator $H_{k,0}^\st$ away from the bottom of the spectrum. This completes the proof of Theorem~\ref{th:Mourre}(i). Choosing the cut-off function $\chi$ with $\chi'$ closer to $\mathds1_{[-1,1]}$, and suitably increasing the factor $\frac12$ in definition~\eqref{eq:def-Ckst}, the term~$\frac34|k|^2$ in the lower bound~\eqref{eq:Hk0-Mourre-pre} below could be improved into almost $\frac89|k|^2$, but the present construction does not allow to reach a value any closer to $|k|^2$.

\begin{lem}
The commutator $[H_{k,0}^\st,\frac1iC_k^\st]$ is well-defined and essentially self-adjoint on~$\Rc(\Omega)$, is $H_{k,0}^\st$-bounded, and satisfies the lower bound
\begin{equation}\label{eq:Hk0-Mourre-pre}
[H_{k,0}^\st,\tfrac1iC_k^\st]\,\ge\,\textstyle\Pi\big(H_{k,0}^\st+\frac34|k|^2\big)\Pi,
\end{equation}
which entails the following Mourre relation on $J_\e:=[\e-\frac34|k|^2,\infty)$, for all $\e>0$,
\[\mathds1_{J_\e}(H_{k,0}^\st)[H_{k,0}^\st,\tfrac1iC_k^\st]\mathds1_{J_\e}(H_{k,0}^\st)\,\ge\,\e\mathds1_{J_\e}(H_{k,0}^\st)-\tfrac34|k|^2\E.\qedhere\]
\end{lem}

\begin{proof}
For all $p\ge0$, we define the operator $M'_{k,p}$ on $\Ld^2_\Sym((\R^d)^p)$ as the multiplication by the function $(x_1,\ldots,x_p)\mapsto (\sum_{j=1}^p\partial_j\widetilde m_p)(k\cdot x_1,\ldots,k\cdot x_p)$, we denote by $M_{k,p}^{\prime\st}$ the corresponding operator defined on the $p$th Wiener chaos $\Hc_p$ by $M_{k,p}^{\prime\st}J_p(u_p)=J_p(M'_{k,p}u_p)$, and we set $M_k^{\prime\st}:=\bigoplus_{p=0}^\infty M_{k,p}^{\prime\st}$ on $\Ld^2(\Omega)$. A direct computation on Wiener chaoses yields
\[[\nabla^\st,M_{k}^\st]=-kM_{k}^{\prime\st}.\]
Combining this with Proposition~\ref{prop:commut}~(i)--(ii), we easily find on $\Rc(\Omega)$,
\begin{eqnarray}
[H_{k,0}^\st,\tfrac1iC_k^\st]&=&\textstyle2(-\triangle^\st)-2ik\cdot\nabla^\st-ik\cdot\frac12(\nabla^\st M_k^{\prime\st}+M_k^{\prime\st} \nabla^\st)+|k|^2M_k^{\prime\st}\nonumber\\
&=&\textstyle2(H_{k,0}^\st+|k|^2)\nonumber\\
&&-ik\cdot\tfrac12\big((\nabla^\st+ik) (M_k^{\prime\st}-2)+(M_k^{\prime\st}-2) (\nabla^\st+ik)\big),\label{eq:commut-HC}
\end{eqnarray}
where the right-hand side is well-defined and symmetric on $\Rc(\Omega)$. We split the rest of the proof into two steps.

\medskip
\step1 Proof of the lower bound~\eqref{eq:Hk0-Mourre-pre} on $\Rc(\Omega)$.\\
Note that constants belong to the kernel of the commutator $[H_{k,0}^\st,\tfrac1iC_k^\st]$.
The above expression~\eqref{eq:commut-HC} for the latter yields for all $\phi\in\Rc(\Omega)$ with $\expec\phi=0$,
\[\big\langle \phi,[H_{k,0}^\st,\tfrac1iC_k^\st]\phi\big\rangle_{\Ld^2(\Omega)}\,\ge\,2\|(\nabla^\st+ik)\phi\|_{\Ld^2(\Omega)}^2-|k|\|(\nabla^\st+ik)\phi\|_{\Ld^2(\Omega)}\|(M_k^{\prime\st}-2)\phi\|_{\Ld^2(\Omega)}.\]
The bound~\eqref{eq:prepre-mourre-mp} implies $1\le M_k^{\prime\st}\le3$ on $\Ld^2(\Omega)\ominus\C$, hence
\[\big\langle \phi,[H_{k,0}^\st,\tfrac1iC_k^\st]\phi\big\rangle_{\Ld^2(\Omega)}\,\ge\,2\|(\nabla^\st+ik)\phi\|_{\Ld^2(\Omega)}^2-|k|\|(\nabla^\st+ik)\phi\|_{\Ld^2(\Omega)}\|\phi\|_{\Ld^2(\Omega)},\]
and we are led to
\begin{eqnarray*}
\big\langle \phi,[H_{k,0}^\st,\tfrac1iC_k^\st]\phi\big\rangle_{\Ld^2(\Omega)}&\ge&\|(\nabla^\st+ik)\phi\|_{\Ld^2(\Omega)}^2-\tfrac14|k|^2\|\phi\|_{\Ld^2(\Omega)}^2\\
&=&\big\langle \phi,(H_{k,0}^\st+\tfrac34|k|^2)\phi\big\rangle_{\Ld^2(\Omega)},
\end{eqnarray*}
that is, the lower bound~\eqref{eq:Hk0-Mourre-pre}.

\medskip
\step2 Proof that the commutator $[H_{k,0}^\st,\tfrac1iC_k^\st]$ is essentially self-adjoint on $\Rc(\Omega)$ and that its closure is $H_{k,0}^\st$-bounded and self-adjoint on its domain $H^2(\Omega)$.

\medskip\noindent
For all $p\ge0$, we define the operator $M''_{k,p}$ on $\Ld^2_\Sym((\R^d)^p)$ as the multiplication by the function $(x_1,\ldots,x_p)\mapsto (\sum_{j,l=1}^p\partial_{jl}^2\widetilde m_p)(k\cdot x_1,\ldots,k\cdot x_p)$, we denote by $M_{k,p}^{\prime\prime\st}$ the corresponding operator defined on the $p$th Wiener chaos $\Hc_p$ by $M_{k,p}^{\prime\prime\st} J_p(u_p)=J_p(M''_{k,p}u_p)$, and we set $M_k^{\prime\prime\st}:=\bigoplus_{p=0}^\infty M_{k,p}^{\prime\prime\st}$ on $\Ld^2(\Omega)$.
A direct computation on Wiener chaoses yields
\[[\nabla^\st,M_k^{\prime\st}]=-kM_{k}^{\prime\prime\st},\]
hence the expression~\eqref{eq:commut-HC} can be rewritten as
\begin{equation}\label{eq:rewrite-commut-second}
[H_{k,0}^\st,\tfrac1iC_k^\st]
\,=\,\textstyle2(H_{k,0}^\st+|k|^2)-ik\cdot(M_k^{\prime\st}-2) (\nabla^{\st}+ik)
+\frac i2|k|^2M_k^{\prime\prime\st}.
\end{equation}
Note that the bound~\eqref{eq:prepre-mourre-mp} ensures that $M_k^{\prime\st}-2$ is bounded by $2$ and that $M_k^{\prime\prime\st}$ is bounded on $\Ld^2(\Omega)$, hence
\begin{eqnarray*}
\lefteqn{\big\|\big(-ik\cdot(M_k^{\prime\st}-2) (\nabla^{\st}+ik)
+\tfrac i2|k|^2M_k^{\prime\prime\st}\big)\phi\big\|_{\Ld^2(\Omega)}}\\
&\lesssim&|k|\|(\nabla^\st+ik)\phi\|_{\Ld^2(\Omega)}+|k|^2\|\phi\|_{\Ld^2(\Omega)}\\
&=&|k|\|(H_{k,0}^\st+|k|^2)^\frac12\phi\|_{\Ld^2(\Omega)}+|k|^2\|\phi\|_{\Ld^2(\Omega)}.
\end{eqnarray*}
Together with~\eqref{eq:rewrite-commut-second}, this shows that the commutator $[H_{k,0}^\st,\tfrac1iC_k^\st]$ is an infinitesimal perturbation of $2H_{k,0}^\st$, and the conclusion follows from the Kato-Rellich theorem.
\end{proof}

We turn to the proof of Theorem~\ref{th:Mourre}(ii), that is, the incompatibility of the perturbation~$V$ with respect to the above-constructed conjugate operator $C_k^\st$. In view of Proposition~\ref{prop:commut}(iii), it remains to establish the following.

\begin{lem}\label{prop:com-L1/2-bounded}
If~$V=V_0$ is itself Gaussian, the commutator $[V,iM_k^{\st}]$ is well-defined and essentially self-adjoint on $\Rc(\Omega)$, but is only $\cal L^{1/2}$-bounded provided that $x\Cc_0^\circ\in\Ld^2(\R^d)$.
\end{lem}

\begin{proof}
In view of the Fock space decomposition~\eqref{eq:V-decomp-re} for multiplication by $V=V_0$,
cf.~Lemma~\ref{lem:V-decomp}, we find
\begin{multline*}
[V,M_k^\st]J_p(u_p)
\,=\,J_{p+1}\bigg(\frac1{p+1}\sum_{j=1}^{p+1}\Cc_0^\circ(x_j)\,u_p(x_1,\ldots,x_{j-1},x_{j+1},\ldots,x_{p+1})\\
\times\big(\widetilde m_p(k\cdot x_1,\ldots,k\cdot x_{j-1},k\cdot x_{j+1},\ldots,k\cdot x_{p+1})-\widetilde m_{p+1}(k\cdot x_1,\ldots,k\cdot x_{p+1})\big)\bigg)\\
\hspace{-3cm}+pJ_{p-1}\bigg(\int_{\R^d}\Cc_0^\circ(z)\,u_p(x_1,\ldots,x_{p-1},z)\\
\times\big(\widetilde m_p(k\cdot x_1,\ldots,k\cdot x_{p-1},k\cdot z)-\widetilde m_{p-1}(k\cdot x_1,\ldots,k\cdot x_{p-1})\big)dz\bigg).
\end{multline*}
In view of Lemma~\ref{prop:m-function} and of the Wiener chaos decomposition~\eqref{eq:wiener-chaos}, arguing similarly as in the proof of Proposition~\ref{prop:commut}(iii), we easily deduce for all $\phi\in\Ld^2(\Omega)$,
\[\|[V,M_k^\st]\phi\|_{\Ld^2(\Omega)}
\,\lesssim\,|k|\Big(\int_{\R^d}|x|^2|\Cc_0^\circ(x)|^2dx\Big)^\frac12\|(\Lc+1)^\frac12\phi\|_{\Ld^2(\Omega)},\]
and the conclusion follows.
\end{proof}

Finally, we argue that, while well-defined and symmetric on $\Rc(\Omega)$, the full commutator~$[H_{k,\lambda}^\st,\frac1iC_k^\st]$ is essentially self-adjoint.

\begin{lem}\label{lem:commut-full-sa}
If~$V=V_0$ is itself Gaussian, the commutator~$[H_{k,\lambda}^\st,\frac1iC_k^\st]$ is well-defined and essentially self-adjoint on $\Rc(\Omega)$.
\end{lem}

This is obtained as a particular case of the following abstract result, which is a convenient reformulation of Nelson's theorem~\cite{Nelson-72}. Note that this result also ensures that $H_{k,\lambda}^\st$ is essentially self-adjoint on $\Rc(\Omega)$ when $V=V_0$ is Gaussian; a more general criterion for essential self-adjointness of $H_{k,\lambda}^\st$ in case of an unbounded potential without particular Gaussian structure is included in Appendix~\ref{sec:self-adj}.

\begin{prop}\label{prop:nelson}
Let $H_1$ and $L$ be self-adjoint operators on their respective domains $D(H_1)$ and $D(L)$ on a Hilbert space $\Hc$, and let $H_2$ be a symmetric operator defined on some dense linear subspace $\Dc\subset \Hc$, such that
\begin{enumerate}[\qquad$\bullet$]
\item $H_1$ and $L$ are nonnegative and commute strongly, hence $H_1+L$ is self-adjoint on $D(H_1)\cap D(L)$;
\item $2H_2$ is a Kato perturbation of $H_1+L$, that is, $\Dc$ is a core for $H_1+L$ and  there is $\alpha<1$ and $C\ge1$ such that there holds, for all $u\in\Dc$,
\[2\|H_2u\|_\Hc\le \alpha\|(H_1+L)u\|_{\Hc}+C\|u\|_\Hc;\]
\item $\pm i[H_2,L]\lesssim H_1+L+1$, that is, for all $u\in\Dc$,
\[\qquad\big|\langle H_2u,Lu\rangle_\Hc-\langle Lu,H_2u\rangle_\Hc\big|\,\lesssim\,\langle u,(H_1+L+1)u\rangle_\Hc.\]
\end{enumerate}
Then the operator $H:=H_1+H_2$ is essentially self-adjoint on $\Dc$.
\end{prop}
 
\begin{proof}
We split the proof into two steps.

\medskip
\step1 Proof that for $\kappa$ large enough the operator $H^\kappa:=H+L+\kappa$ is essentially self-adjoint on $\Dc$ and satisfies $H^\kappa\ge1$.

\medskip\noindent
As $H_1$ and $L$ are self-adjoint on their respective domains and nonnegative and as they commute strongly, their sum $H_1^\kappa:=H_1+L+\kappa$ is self-adjoint on $D(H_1)\cap D(L)$, cf.~\cite[Lemma~4.15.1]{Putnam-67}. Next, as $H_2$ is a Kato perturbation of $H_1+L$, it is also a Kato perturbation of $H_1^\kappa$ for all $\kappa\ge0$,
hence the Kato-Rellich theorem entails that $H^\kappa:=H_1^\kappa+H_2$ is essentially self-adjoint on $\Dc$.
In addition, as $H_2$ is a Kato perturbation of $H_1+L$ and as $H_1+L$ is nonnegative, it is easily deduced that $H_1+H_2+L$ is bounded from below, hence $H^\kappa\ge1$ for $\kappa$ large enough.

\medskip
\step2 Conclusion.\\
In view of Nelson's theorem~\cite{Nelson-72} in form of~\cite[Corollary 1.1]{Faris-Lavine-74}, together with the result of Step~1, the conclusion follows provided that we can check the following two additional properties, for $\kappa$ large enough,
\begin{enumerate}[(i)]
\item $H$ is $H^\kappa$-bounded on $\Dc$;
\item $\pm i[H,H^\kappa]\lesssim_\kappa H^\kappa$, that is, for all $u\in\Dc$,
\[\big|\langle Hu,H^\kappa u\rangle_\Hc-\langle H^\kappa u,Hu\rangle_\Hc\big|\,\lesssim_\kappa\,\langle u,H^\kappa u\rangle_\Hc.\]
\end{enumerate}
We start with the proof of condition~(i).
On the one hand, since $H_1$ and $L$ are nonnegative and commute strongly, we can deduce $\|Lu\|_\Hc\le\|H_1^\kappa u\|_\Hc$ for all $u\in D(H_1)\cap D(L)$.
On the other hand, since $H_2$ is a Kato perturbation of $H_1+L$, we find for all $u\in \Dc$,
\begin{equation*}
\|H_2u\|_\Hc\le \alpha\|(H_1+L)u\|_\Hc+C\|u\|_\Hc\le \alpha\|H_1^\kappa u\|_\Hc+C\|u\|_\Hc,
\end{equation*}
hence
\[\|H_1^\kappa u\|_\Hc\le\|H^\kappa u\|_\Hc+\|H_2 u\|_\Hc\le\|H^\kappa u\|_\Hc+ \alpha\|H_1^\kappa u\|_\Hc+C\|u\|_\Hc,\]
which leads to
\[\|H_1^\kappa u\|_\Hc\le\tfrac1{1-\alpha}\|H^\kappa u\|_\Hc+\tfrac C{1-\alpha}\|u\|_\Hc.\]
Combined with the above, this yields
\[\|Lu\|_\Hc\le\|H_1^\kappa u\|_\Hc\lesssim \|H^\kappa u\|_\Hc+\|u\|_\Hc,\]
hence in particular,
\begin{eqnarray*}
\|H u\|_\Hc\,\le\,\|H^\kappa u\|_\Hc+\|Lu\|_\Hc+\kappa\|u\|_\Hc
\,\lesssim\,\|H^\kappa u\|_\Hc+(1+\kappa)\|u\|_\Hc,
\end{eqnarray*}
that is,~(i).
It remains to establish condition~(ii). As $H_1$ and $L$ commute, we can write for all $u\in \Dc$,
\[\langle Hu,H^\kappa u\rangle_\Hc-\langle H^\kappa u,Hu\rangle_\Hc=\langle Hu,L u\rangle_\Hc-\langle L u,Hu\rangle_\Hc=\langle H_2u,L u\rangle_\Hc-\langle L u,H_2u\rangle_\Hc,\]
hence, by assumption,
\[\big|\langle Hu,H^\kappa u\rangle_\Hc-\langle H^\kappa u,Hu\rangle_\Hc\big|\lesssim\langle u,(H_1+L+1)u\rangle_\Hc.\]
Again noting as in Step~1 that the nonnegativity of $H_1+L$ and the fact that $2H_2$ is a Kato perturbation of $H_1+L$ imply that $(H_1+L)+2H_2+\kappa$ is nonnegative for $\kappa$ large enough, the claim~(ii) follows.
\end{proof}

With the above abstract result at hand, we quickly indicate how Lemma~\ref{lem:commut-full-sa} is a simple consequence.

\begin{proof}[Proof of Lemma~\ref{lem:commut-full-sa}]
In view of~\eqref{eq:commut-HC}, we can decompose $[H_{k,\lambda}^\st,\tfrac1iC_k^\st]=H_1+H_2$ on~$\Rc(\Omega)$,
in terms of
\begin{eqnarray*}
H_1&:=&2(H_{k,0}^\st+|k|^2),\\
H_2&:=&H_{2,1}+H_{2,2},\\
H_{2,1}&:=&-ik\cdot\tfrac12\big((\nabla^\st+ik) (M_k^{\prime\st}-2)+(M_k^{\prime\st}-2) (\nabla^\st+ik)\big),\\
H_{2,2}&:=&-\lambda\big([V,iA^\st]+\tfrac12[V,iM_k^\st]\big).
\end{eqnarray*}
We shall appeal to Proposition~\ref{prop:nelson} with $L:=\Lc$, and it remains to check the different assumptions.
First, $H_1=-2(\nabla^\st+ik)\cdot(\nabla^\st+ik)$ and $\Lc$ are both essentially self-adjoint on $\Rc(\Omega)$, as discussed in Sections~\ref{sec:FB-fibr} and~\ref{chap:Mall}, respectively, they are clearly nonnegative, and Lemma~\ref{lem:nabla-Op} ensures that they commute strongly.
Also note that $H_1$ and $\Lc$ leave the linear subspace $\Rc(\Omega)$ invariant and that spectral projections of $\Lc$ also leave $\Rc(\Omega)$ invariant.
Using projections onto a finite number of chaoses, one then easily sees that $H_1+L$ is essentially self-adjoint on $\Rc(\Omega)$.
Next, we may rewrite as in~\eqref{eq:rewrite-commut-second},
\begin{equation}\label{eq:rewr-H21}
H_{2,1}\,=\,\textstyle-ik\cdot(M_k^{\prime\st}-2) (\nabla^{\st}+ik)
+\frac i2|k|^2M_k^{\prime\prime\st},
\end{equation}
hence the boundedness of $M_k^{\prime\st}$ and $M_k^{\prime\prime\st}$ leads to
\begin{equation*}
\|H_{2,1}\phi\|_{\Ld^2(\Omega)}\,\lesssim\,\textstyle|k|\|(\nabla^{\st}+ik)\phi\|_{\Ld^2(\Omega)}
+|k|^2\|\phi\|_{\Ld^2(\Omega)}
\,\lesssim\,\textstyle|k|\|H_1^\frac12\phi\|_{\Ld^2(\Omega)}+|k|^2\|\phi\|_{\Ld^2(\Omega)},
\end{equation*}
showing that $H_{2,1}$ is $H_1^{1/2}$-bounded.
Proposition~\ref{prop:commut}(iii) and Lemma~\ref{prop:com-L1/2-bounded} also ensure that $[V,iA^\st]$ and $[V,iM_k^\st]$ are $\Lc^{1/2}$-bounded. This proves that $H_2$ is $(H_1+\Lc)^{1/2}$-bounded, hence $(H_1+\Lc)$-infinitesimal.
Finally, it remains to consider the commutator $[H_2,\Lc]$. Since by definition $M_k^{\prime\st}$ and $M_k^{\prime\prime\st}$ preserve chaoses, identity~\eqref{eq:rewr-H21} implies $[H_{2,1},\Lc]=0$. In view of the explicit description of $[V,iA^\st]$ and $[V,iM_k^\st]$ in Proposition~\ref{prop:commut}(iii) and in Lemma~\ref{prop:com-L1/2-bounded}, respectively, as these commutators have a similar structure as $V$ itself on the Wiener chaos decomposition, a similar computation as in the proof of Proposition~\ref{prop:commut}(iv) easily shows that $[H_{2,2},\Lc]$ is $\Lc^{1/2}$-bounded, and the conclusion follows.
\end{proof}

\subsection{Consequences of Mourre's relations}\label{sec:consequ-Mourre}
This section is devoted to the proof of Corollary~\ref{th:mourre-partial}.
Let~$V=V_0$ be a stationary Gaussian field.
Given some $L_0>0$, consider the projection $Q_\lambda := 1_{[0,(L_0\lambda)^{-2}]}(\Lc)$ onto Wiener chaoses of order~$\le(L_0\lambda)^{-2}$. We split the proof into three steps.

\medskip
\step{1} Preliminary on spectral trunctations: For all $\phi\in\Ld^2(\Omega)$, $L\ge1$, and $h\in\Ld^\infty(\R)$ supported in $\R\setminus\frac12[-\frac1L,\frac1L]$, there holds
\begin{equation}\label{eq:truncate-sp}
\big|\expecm{\bar\phi\,e^{-itH_{k,\lambda}^\st}h(H_{k,\lambda}^\st)1}\!\big|\,\lesssim\,\lambda L\,
\|h\|_{\Ld^\infty(\R)}\|\phi\|_{\Ld^2(\Omega)},
\end{equation}
where the factor $\lambda L$ can be replaced by $(\lambda L)^2$ if $\phi=1$.

\medskip\noindent
Choose $h_0\in C^\infty(\R)$ with $h_0(y)=1$ for $|y|\ge\frac1{2L}$ and $h_0(y)=0$ for $|y|\le\frac1{4L}$ such that $h_0\ge0$ and $|h'_0|\le8L$ pointwise.
By definition, we find $|h|\le\|h\|_{\Ld^\infty(\R)}|h_0|$, hence
\[\|h(H_{k,\lambda}^\st)1\|_{\Ld^2(\Omega)}\le\|h\|_{\Ld^\infty(\R)}\|h_0(H_{k,\lambda}^\st)1\|_{\Ld^2(\Omega)}.\]
Using spectral calculus with $h_0(0)=0$ and $H_{k,\lambda}^\st1=\lambda V$, this leads to
\begin{eqnarray*}
\|h(H_{k,\lambda}^\st)1\|_{\Ld^2(\Omega)}&\le&\|h\|_{\Ld^\infty(\R)}\Big\|\int_0^1h_0'(sH_{k,\lambda}^\st)H_{k,\lambda}^\st1\,ds\Big\|_{\Ld^2(\Omega)}\\
&\le&8\,\lambda L\,\|h\|_{\Ld^\infty(\R)}\|V\|_{\Ld^2(\Omega)}.
\end{eqnarray*}
If $\phi$ is nonconstant, the conclusion directly follows from the Cauchy--Schwarz inequality. If $\phi=1$, decomposing $h$ into its positive and negative parts $h_+$ and $h_-$, we can rather estimate
\begin{eqnarray*}
\big|\expecm{e^{-itH_{k,\lambda}^\st}h(H_{k,\lambda}^\st)1}\!\big|&\le&\big\|(h_+)^\frac12(H_{k,\lambda}^\st)1\big\|_{\Ld^2(\Omega)}^2+\big\|(h_-)^\frac12(H_{k,\lambda}^\st)1\big\|_{\Ld^2(\Omega)}^2\\
&\le&128\,(\lambda L)^2\|h\|_{\Ld^\infty(\R)}\|V\|_{\Ld^2(\Omega)}^2.
\end{eqnarray*}

\medskip
\step2 Proof that for all $k\in\R^d\setminus\{0\}$ the flow for the truncated fibered Schrödinger operator $Q_\lambda H_{k,\lambda}^\st Q_\lambda$ satisfies for all $s\ge0$,
\begin{equation}\label{eq:preconcm-res}
\lim_{\lambda\downarrow0}\Big|\expecm{e^{-i\lambda^{-2}s\,Q_\lambda H_{k,\lambda}^\st Q_\lambda}1}-e^{-s(\alpha_k+i\beta_k)}\Big|\,=\,0.
\end{equation}
We wish to apply Theorem~\ref{th:CGH} in form of Remark~\ref{rem:CGH} with $H_\lambda^\circ:=Q_\lambda H_{k,0}^\st Q_\lambda$, $W_\lambda^\circ:=Q_\lambda VQ_\lambda$, $E_0:=0$, $\psi_0:=1$, and it suffices to check the different assumptions.
First, as the conjugate operator $C_k^\st$ commutes with the Wiener chaos decomposition, Theorem~\ref{th:Mourre}(i) ensures that the truncated operator $Q_\lambda H_{k,0}^\st Q_\lambda$ restricted to $\Ld^2(\Omega)\ominus\C$ satisfies a strict Mourre relation on $J_\e:=[\e-\frac34|k|^2,\infty)$ for all $\e>0$: we find
\[\big[ Q_\lambda H_{k,0}^\st Q_\lambda,\tfrac1iC_k^\st\big]\ge \Pi Q_\lambda\big( H_{k,0}^\st +\tfrac34|k|^2\big)Q_\lambda\Pi,\]
hence, on $\Ld^2(\Omega)\ominus\C$,
\[\mathds1_{J_\e}(Q_\lambda H_{k,0}^\st Q_\lambda)\big[ Q_\lambda H_{k,0}^\st Q_\lambda,\tfrac1iC_k^\st\big]\mathds1_{J_\e}(Q_\lambda H_{k,0}^\st Q_\lambda)\ge \e\,\mathds1_{J_\e}( Q_\lambda H_{k,0}^\st Q_\lambda).\]
Similar computations as in the proof of Theorem~\ref{th:Mourre}(i) show the $H_{k,0}^\st$-boundedness of iterated commutators $\operatorname{ad}^k_{C_k^\st}(H_{k,0}^\st)$,
hence the $\lambda$-uniform $Q_\lambda H_{k,0}^\st Q_\lambda$-boundedness of
\[\operatorname{ad}^k_{C_k^\st}(Q_\lambda H_{k,0}^\st Q_\lambda)=Q_\lambda\operatorname{ad}^k_{C_k^\st}( H_{k,0}^\st )\,Q_\lambda.\]
In addition, as the domain of $H_{k,0}^\st$ is $H^2(\Omega)$, as Theorem~\ref{th:Mourre}(i) ensures that $H^2(\Omega)$ is invariant under the unitary group generated by $C_k^\st$, and as the latter further preserves the chaos decomposition,
we deduce that the domain of $Q_\lambda H_{k,0}^\st Q_\lambda$ is also invariant under the unitary group generated by $C_k^\st$.

\medskip\noindent
It remains to check assumptions on the perturbation $Q_\lambda VQ_\lambda$. Note that $\Pi(Q_\lambda VQ_\lambda1)=V$, which clearly satisfies $\|\langle C_k^\st\rangle^6V\|_{\Ld^2(\Omega)}\lesssim1$. Further, iterating the proof of Theorem~\ref{th:Mourre}(ii), we find that for all $k\ge0$ iterated commutators $\operatorname{ad}^k_{C_k^\st}(V)$ are $\Lc^{1/2}$-bounded, hence for all~$\phi\in\Ld^2(\Omega)$,
\[\big\|\operatorname{ad}^k_{C_k^\st}(\lambda Q_\lambda VQ_\lambda)\phi\big\|_{\Ld^2(\Omega)}\lesssim\lambda\big\|Q_\lambda\Lc^{1/2}\phi\big\|_{\Ld^2(\Omega)}+\lambda\|\phi\|_{\Ld^2(\Omega)},\]
which entails, by definition of $Q_\lambda$,
\[\big\|\operatorname{ad}^k_{C_k^\st}(\lambda Q_\lambda VQ_\lambda)\phi\big\|_{\Ld^2(\Omega)}\lesssim\big(\lambda+\tfrac1{L_0}\big)\|\phi\|_{\Ld^2(\Omega)}.\]
Choosing $L_0\simeq1$ large enough, we may then apply Theorem~\ref{th:CGH} in form of Remark~\ref{rem:CGH},
to the effect of
\[\expec{e^{-iQ_\lambda H_{k,\lambda}^\st Q_\lambda t}\,\mathds1_{J_\e}(Q_\lambda H_{k,\lambda}^\st Q_\lambda)\,1}=e^{-\lambda^2t(\alpha_k+i\beta_k)}+o(1).\]
In view of Step~1, the spectral truncation $\mathds1_{J_\e}(Q_\lambda H_{k,\lambda}^\st Q_\lambda)$ can be removed up to a further $O(\lambda^2)$ error, and the claim~\eqref{eq:preconcm-res} follows.

\medskip
\step3 Conclusion.\\
In view of the result of Step~1, it remains to prove for all $0\le s\le (e^2\Cc_0(0)^\frac12L_0)^{-1}$,
\begin{equation}\label{eq:remains-dec}
\lim_{\lambda\downarrow0}\Big\|e^{-i\lambda^{-2}s\,Q_\lambda H_{k,\lambda}^\st Q_\lambda}1-e^{-i\lambda^{-2}s H_{k,\lambda}^\st }1\Big\|_{\Ld^2(\Omega)}\,=\,0,
\end{equation}
while the conclusion of Corollary~\ref{th:mourre-partial} then follows from the fibration~\eqref{eq:flow-decomp0}.
Set for abbreviation $u_{k,\lambda}^t:=e^{-itH_{k,\lambda}^\st}1$ and $\tilde u_{k,\lambda}^t:=e^{-itQ_\lambda H_{k,\lambda}^\st Q_\lambda}1$. Since the flow $u_{k,\lambda}$ satisfies the equation $i\partial_tu_{k,\lambda}=(H_{k,0}^\st+\lambda V)u_{k,\lambda}$, an iterative use of Duhamel's formula allows to decompose, for all~$N\ge1$,
\[u_{k,\lambda}^t\,=\,\sum_{p=0}^N(-i\lambda)^pu_{k}^{p;t}+(-i\lambda)^{N+1}E_{k,\lambda}^{N;t},\]
in terms of
\begin{eqnarray*}
u_{k}^{p;t}&:=&\int_{(\R^+)^{p+1}}\delta\Big(t-\sum_{j=1}^{p+1}s_j\Big)\, e^{-is_1H_{k,0}^\st}Ve^{-is_2H_{k,0}^\st}\ldots Ve^{-is_pH_{k,0}^\st}V\,ds_1\ldots ds_{p+1},\\
E_{k,\lambda}^{N;t}&:=&\int_{(\R^+)^{N+2}}\delta\Big(t-\sum_{j=1}^{N+2}s_j\Big)\, e^{-is_1H_{k,\lambda}^\st}Ve^{-is_2H_{k,0}^\st}\ldots Ve^{-is_{N+1}H_{k,0}^\st}V\,ds_1\ldots ds_{N+2}.
\end{eqnarray*}
Noting that $u_k^{p;t}\in\Hc_{p}$ for all $p\ge0$, we deduce that the truncation error can be represented as follows, for any $N\le (L_0\lambda)^{-2}$,
\begin{equation}\label{eq:trunc-err0}
u_{k,\lambda}^t-\tilde u_{k,\lambda}^t=(-i\lambda)^{N+1}(E_{k,\lambda}^{N;t}-\tilde E_{k,\lambda}^{N;t}),
\end{equation}
where $\tilde E_{k,\lambda}^{N;t}$ is defined similarly as $E_{k,\lambda}^{N;t}$ with $V$ and $H_{k,\lambda}^\st$ replaced by $Q_\lambda V Q_\lambda$ and $Q_\lambda H_{k,\lambda}^\st Q_\lambda$, respectively.
A direct estimate yields
\begin{equation*}
\|E_{k,\lambda}^{N;t}\|_{\Ld^2(\Omega)}\,\le\,\int_{(\R^+)^{N+2}}\delta\Big(t-\sum_{j=1}^{N+2}s_j\Big)\,\big\|Ve^{-is_2H_{k,0}^\st}\ldots Ve^{-is_{N+1}H_{k,0}^\st}V\big\|_{\Ld^2(\Omega)}\,ds_1\ldots ds_{N+2},
\end{equation*}
hence, noting that $V$ is bounded by $\Cc_0(0)^\frac12(2p+1)^\frac12$ on $\bigcup_{n\le p}\Hc_n$, cf.~Lemma~\ref{lem:V-decomp},
\begin{eqnarray*}
\|E_{k,\lambda}^{N;t}\|_{\Ld^2(\Omega)}&\le&\Cc_0(0)^{\frac12(N+1)}(2N+1)!!^\frac12\int_{(\R^+)^{N+2}}\delta\Big(t-\sum_{j=1}^{N+2}s_j\Big)\,ds_1\ldots ds_{N+2}\\
&\le&\frac{(2N+1)!!^\frac12}{(N+1)!} (\Cc_0(0)^\frac12t)^{N+1}\,\le\,\frac{(e\Cc_0(0)^\frac12t)^{N+1}}{(N+1)!^\frac12}.
\end{eqnarray*}
Similarly estimating $\tilde E_{k,\lambda}^{N;t}$ and inserting this into~\eqref{eq:trunc-err0}, we find for all $N\le (L_0\lambda)^{-2}$,
\begin{equation*}
\|u_{k,\lambda}^t-\tilde u_{k,\lambda}^t\|_{\Ld^2(\Omega)}\le\frac{2(e\Cc_0(0)^\frac12\lambda t)^{N+1}}{(N+1)!^\frac12}.
\end{equation*}
Setting $t=\lambda^{-2}s$ and choosing $N=\lfloor(L_0\lambda)^{-2}\rfloor$, we easily deduce for $s\le (e^2\Cc_0(0)^\frac12 L_0)^{-1}$,
\begin{equation*}
\|u_{k,\lambda}^{\lambda^{-2}s}-\tilde u_{k,\lambda}^{\lambda^{-2}s}\|_{\Ld^2(\Omega)}\lesssim \lambda^{-1}e^{-\frac12(L_0\lambda)^{-2}},
\end{equation*}
and the claim~\eqref{eq:remains-dec} follows.
\qed

%%%%%%%%%%%%%%%%%%%%%
%%%%%%%%%%%%%%%%%%%%%
%%%%%%%%%%%%%%%%%%%%%

\medskip
\section{Exact resonance conjectures and consequences}\label{sec:resconj}
This section is devoted to the proof of Corollary~\ref{th:res-spectr} and Proposition~\ref{prop:pert-res} as consequences of the resonance conjectures~\ref{C1} and~\ref{C2}.

\subsection{Resonant-mode expansion}\label{sec:expdec-conj-pr}
We start with the proof of Corollary~\ref{th:res-spectr}.
For $\phi\in \Pc(\Omega)$, the Floquet--Bloch fibration~\eqref{eq:flow-decomp0} takes the form
\[\expec{\bar\phi\, u_\lambda^t(x)}
\,=\,\int_{\R^d}\widehat u^\circ(k)\,e^{ik\cdot x-it|k|^2}\,\expecm{\overline{\phi(\tau_x\cdot)}e^{-itH_{k,\lambda}^\st}1}\,\dbar k,\]
hence it suffices analyze $\expecm{\bar\phi\,e^{-itH_{k,\lambda}^\st}1}$ for fixed $k$ and $\phi\in \Pc(\Omega)$.
We split the proof into three steps, separately establishing items~(i) and~(ii).

\medskip
\step1 Meromorphic extension of the spectral measure: Under~\ref{C1}, for all $\phi\in \Pc(\Omega)$, the spectral measure $\mu_{k,\lambda}^{\phi,1}$ is analytic on $[-\frac1M,\frac1M]$ and admits a local meromorphic extension $\nu_{k,\lambda}^{\phi,1}$ on the complex neighborhood $\frac1MB$,
\begin{equation}\label{eq:meas-ext-0}
\nu_{k,\lambda}^{\phi,1}(z)\,=\,\frac{1}{2i\pi}\bigg(\frac{\langle\phi,\Pi_{k,\lambda}1\rangle_{\Pc(\Omega),\Pc'(\Omega)}}{z_{k,\lambda}-z}-\frac{\langle\phi,\Pi_{k,\lambda}^*1\rangle_{\Pc(\Omega),\Pc'(\Omega)}}{\overline{z_{k,\lambda}}-z}\bigg)
+\frac1{2i\pi}\big(\zeta_{k,\lambda}^{\phi,1}(z)-\overline{\zeta_{k,\lambda}^{1,\phi}}(z)\big),
\end{equation}
which can alternatively be expressed as
\begin{eqnarray}\label{eq:meas-ext-1}
\lefteqn{\nu_{k,\lambda}^{\phi,1}(z)\,=\,\frac\lambda z\nu_{k,\lambda}^{\phi,V}(z)}\\
&=&\frac{\lambda}{2i\pi z}\bigg(\frac{\langle\phi,\Pi_{k,\lambda}V\rangle_{\Pc(\Omega),\Pc'(\Omega)}}{z_{k,\lambda}-z}-\frac{\langle\phi,\Pi_{k,\lambda}^*V\rangle_{\Pc(\Omega),\Pc'(\Omega)}}{\overline{z_{k,\lambda}}-z}\bigg)
+\frac\lambda{2i\pi z}\big(\zeta_{k,\lambda}^{\phi,V}(z)-\overline{\zeta_{k,\lambda}^{V,\phi}}(z)\big),\nonumber
\end{eqnarray}
and moreover, in case $\phi=1$,
\begin{eqnarray}\label{eq:meas-ext-11}
\nu_{k,\lambda}^{1,1}(z)&=&\frac{\lambda^2}{z^2}\nu_{k,\lambda}^{V,V}(z)\\
&=&\frac{\lambda^2}{2i\pi z^2}\bigg(\frac{\langle V,\Pi_{k,\lambda}V\rangle_{\Pc(\Omega),\Pc'(\Omega)}}{z_{k,\lambda}-z}-\frac{\langle V,\Pi_{k,\lambda}^*V\rangle_{\Pc(\Omega),\Pc'(\Omega)}}{\overline{z_{k,\lambda}}-z}\bigg)
+\frac{\lambda^2}{\pi z^2}\big(\Im\zeta_{k,\lambda}^{V,V}\big)(z),\nonumber
\end{eqnarray}
where we write for abbreviation
\begin{eqnarray}\label{eq:not-Pi}
\langle\phi',\Pi_{k,\lambda}\phi\rangle_{\Pc(\Omega),\Pc'(\Omega)}&=&\overline{\langle\Psi_{k,\lambda}^+,\phi'\rangle_{\Pc'(\Omega),\Pc(\Omega)}}\langle\Psi_{k,\lambda}^-,\phi\rangle_{\Pc'(\Omega),\Pc(\Omega)},\\
\langle\phi',\Pi^*_{k,\lambda}\phi\rangle_{\Pc(\Omega),\Pc'(\Omega)}&=&\overline{\langle\Psi_{k,\lambda}^-,\phi'\rangle_{\Pc'(\Omega),\Pc(\Omega)}}\langle\Psi_{k,\lambda}^+,\phi\rangle_{\Pc'(\Omega),\Pc(\Omega)}.\nonumber
\end{eqnarray}

\medskip\noindent
Indeed, Stone's formula together with Conjecture~\ref{C1} yields for $|y|<\frac1M$,
\begin{eqnarray*}
\lefteqn{\mu_{k,\lambda}^{\phi',\phi}(y)
\,=\,\lim_{\e\downarrow0}\frac{1}{2i\pi}\Big(\big\langle\phi',(H^\st_{k,\lambda}-y-i\e)^{-1}\phi\big\rangle_{\Ld^2(\Omega)}-\big\langle\phi',(H^\st_{k,\lambda}-y+i\e)^{-1}\phi\big\rangle_{\Ld^2(\Omega)}\Big)}\\
&=&\lim_{\e\downarrow0}\frac{1}{2i\pi}\Big(\big\langle\phi',(H^\st_{k,\lambda}-y-i\e)^{-1}\phi\big\rangle_{\Ld^2(\Omega)}-\overline{\big\langle\phi,(H^\st_{k,\lambda}-y-i\e)^{-1}\phi'\big\rangle_{\Ld^2(\Omega)}}\Big)\\
&=&\frac{1}{2i\pi}\bigg(\frac{\langle\phi',\Pi_{k,\lambda}\phi\rangle_{\Pc(\Omega),\Pc'(\Omega)}}{z_{k,\lambda}-y}-\frac{\langle\phi',\Pi_{k,\lambda}^*\phi\rangle_{\Pc(\Omega),\Pc'(\Omega)}}{\overline{z_{k,\lambda}}-y}\bigg)+\frac{1}{2i\pi}\big(\zeta_{k,\lambda}^{\phi',\phi}(y)-\overline{\zeta_{k,\lambda}^{\phi,\phi'}}(y)\big),
\end{eqnarray*}
and~\eqref{eq:meas-ext-0} follows. Identities~\eqref{eq:meas-ext-1} and~\eqref{eq:meas-ext-11} are obvious consequences as $H_{k,\lambda}^\st1=\lambda V$.

\medskip
\step2 Proof of~(i): Under~\ref{C1}, for all $k\in K$, $0\le\lambda<\lambda_0$, $\phi\in \Pc(\Omega)$, $n\ge0$, and $g\in C^\infty_c(\R)$ supported in $[-\frac1M,\frac1M]$ with $g=1$ in $\frac12[-\frac1M,\frac1M]$, there holds
\[\Big|\expecm{\bar\phi\,e^{-itH_{k,\lambda}^\st}g(H_{k,\lambda}^\st)1}-e^{-itz_{k,\lambda}}\langle\phi,\Pi_{k,\lambda}1\rangle_{\Pc(\Omega),\Pc'(\Omega)}\Big|\,\lesssim_{n,\phi,g,M}\,\lambda(1+t)^{-n},\]
where the factor $\lambda$ in the right-hand side can be replaced by $\lambda^2$ in case $\phi=1$. Combined with~\eqref{eq:truncate-sp} to remove the spectral truncation, this indeed yields~(i).

\medskip\noindent
Starting from
\begin{equation*}
\expecm{\bar\phi\,e^{-itH_{k,\lambda}^\st}g(H_{k,\lambda}^\st)1}
\,=\,\int_{[-\frac1M,\frac1M]} e^{-ity}g(y)\,\mu_{k,\lambda}^{\phi,1}(y)\,dy,
\end{equation*}
and using formula~\eqref{eq:meas-ext-0} for the meromorphic extension $\nu_{k,\lambda}^{\phi,1}$ of the spectral measure $\mu_{k,\lambda}^{\phi,1}$,
we obtain by contour deformation,
\begin{equation*}
\expecm{\bar\phi\,e^{-itH_{k,\lambda}^\st}g(H_{k,\lambda}^\st)1}
\,=\,e^{-itz_{k,\lambda}}\langle\phi,\Pi_{k,\lambda}1\rangle_{\Pc(\Omega),\Pc'(\Omega)}+\int_\gamma e^{-itz}g(\Re z)\,\nu_{k,\lambda}^{\phi,1}(z)\,dz,
\end{equation*}
where the smooth path $\gamma$ is a deformation of the real interval $[-\frac1M,\frac1M]$ in the lower half-plane such that $\gamma$ remains on the real axis on $[-\frac1M,\frac1M]\setminus\frac12[-\frac1M,\frac1M]$ while the part on $\frac12 [-\frac1M,\frac1M]$ is deformed into a path in $\{z:\Im z\le0,|\Re z|\le\frac1{2M}\}\bigcap\frac1MB$ that stays pointwise at a distance $\frac1{4M}$ from the origin.
Using the identity $e^{-itz}\,=\,(1+t)^{-1}(1+i\tfrac{d}{dz})e^{-itz}$ and integrating by parts, we find iteratively for all $n\ge0$,
\begin{multline*}
\Big|\expecm{\bar\phi\,e^{-itH_{k,\lambda}^\st}g(H_{k,\lambda}^\st)1}-e^{-itz_{k,\lambda}}\langle\phi,\Pi_{k,\lambda}1\rangle_{\Pc(\Omega),\Pc'(\Omega)}\Big|\\
\,\le\,(1+t)^{-n}\int_\gamma \big|(1-i\tfrac d{dz})^n\big(g(\Re z)\,\nu_{k,\lambda}^{\phi,1}(z)\big)\big|\,d|z|.
\end{multline*}
As the remainder $\zeta_{k,\lambda}^{\phi,1}$ is holomorphic on $\frac1MB$ and has continuous dependence on $\lambda$ for $0\le\lambda<\lambda_0$, cf.~\ref{C1}, we deduce that all its derivatives are bounded on $\frac1{2M}B$ uniformly with respect to $\lambda$.
Hence, it follows from~\eqref{eq:meas-ext-1} that all derivatives of $\nu_{k,\lambda}^{\phi,1}$ are uniformly bounded by $O(\lambda)$ on the path $\gamma$, and thus
\begin{equation}\label{eq:concl-step3}
\Big|\expecm{\bar\phi\,e^{-itH_{k,\lambda}^\st}g(H_{k,\lambda}^\st)1}-e^{-itz_{k,\lambda}}\langle\phi,\Pi_{k,\lambda}1\rangle_{\Pc(\Omega),\Pc'(\Omega)}\Big|\,\lesssim_{n,\phi,g,M}\,\lambda(1+t)^{-n},
\end{equation}
where in view of~\eqref{eq:meas-ext-11} the factor $\lambda$ can be replaced by $\lambda^2$ in case $\phi=1$.

\medskip
\step3 Proof of~(ii): Under~\ref{C2}, for all $k\in K$, $0\le \lambda<\lambda_0$, and $\phi\in \Pc(\Omega)$,
\[\Big|\expecm{\bar\phi\,e^{-itH_{k,\lambda}^\st}1}-e^{-itz_{k,\lambda}}\langle\phi,\Pi_{k,\lambda}1\rangle_{\Pc(\Omega),\Pc'(\Omega)}\Big|\,\lesssim_{\phi}\,\lambda e^{-\frac t{8M}\lambda^\rho},\]
where the factor $\lambda$ in the right-hand side can be replaced by $\lambda^2$ in case $\phi=1$.

\medskip\noindent
Similarly as in Step~2, applying formula~\eqref{eq:meas-ext-0} and contour deformation, we find
\begin{equation*}
\expecm{\bar\phi\,e^{-itH_{k,\lambda}^\st}1}
\,=\,e^{-itz_{k,\lambda}}\langle\phi,\Pi_{k,\lambda}1\rangle_{\Pc(\Omega),\Pc'(\Omega)}+\int_\gamma e^{-itz}\nu_{k,\lambda}^{\phi,1}(z)\,dz,
\end{equation*}
where the smooth path $\gamma$ is
given by $\gamma:=\gamma_3^-\cup\gamma_2^-\cup\gamma_1^-\cup\gamma_0\cup\gamma_1^+\cup\gamma_2^+\cup\gamma_3^+$ with
\begin{enumerate}[$\bullet$]
\item $\gamma_0:=\gamma\cap\{z:|\Re z|\le\frac1{2M}\}$ connects $\frac1{2M}(-1-i\lambda^\rho)$ and $\frac1{2M}(1-i\lambda^\rho)$, does not exit the ball of radius $\frac1{M}$, and always stays at a distance $\frac1{4M}$ from the origin;
\item $\gamma_1^-:=-\frac{i}{2M}\lambda^\rho+[-L,-\frac1{2M}]$ and $\gamma_1^+:=-\frac{i}{2M}\lambda^\rho+[\frac1{2M},L]$;
\item $\gamma_2^\pm$ connects $\pm L-\frac{i}{2M}\lambda^\rho$ and $\pm(L+1)$ without crossing the real axis;
\item $\gamma_3^-:=(-\infty,-L-1]$ and $\gamma_3^+:=[L+1,\infty)$;
\end{enumerate}
where $L\ge2$ is to be fixed later.
Inserting formula~\eqref{eq:meas-ext-1} for $\nu_{k,\lambda}^{\phi,1}$, using the uniform bound assumed to hold on $\gamma$, cf.~\ref{C2},
and setting $\gamma_3:=\gamma_3^-\cup\gamma_3^+$, the above turns into
\begin{multline*}
\Big|\expecm{\bar\phi\,e^{-itH_{k,\lambda}^\st}g(H_{k,\lambda}^\st)1}
-e^{-itz_{k,\lambda}}\langle\phi,\Pi_{k,\lambda}1\rangle_{\Pc(\Omega),\Pc'(\Omega)}\Big|\\
\,\lesssim_{\phi,g}\,\lambda^{1-M}\int_{\gamma\setminus\gamma_3}\frac{|e^{-itz}|}{|z|}\,d|z|+\lambda\,\Big|\int_{\gamma_3}\frac{e^{-ity}}{y}\mu_{k,\lambda}^{\phi,V}(y)\,dy\Big|.
\end{multline*}
Hence, by definition of $\gamma$,
\begin{multline*}
\Big|\expecm{\bar\phi\,e^{-itH_{k,\lambda}^\st}g(H_{k,\lambda}^\st)1}
-e^{-itz_{k,\lambda}}\langle\phi,\Pi_{k,\lambda}1\rangle_{\Pc(\Omega),\Pc'(\Omega)}\Big|\\
\,\lesssim_{\phi,g}\,\lambda^{1-M}L^{-1}+\lambda^{1-M} e^{-\frac t{2M}\lambda^\rho}\bigg(1+\int_\frac1M^L\frac{dy}y\bigg)
\,\lesssim_{\phi,g}\,\lambda^{1-M}\Big(L^{-1}+e^{-\frac t{2M}\lambda^\rho}\log L\Big).
\end{multline*}
Optimization in $L\ge2$ yields the bound $\lambda^{1-M}e^{-\frac{t}{4M}\lambda^\rho}$. Interpolating this with the result~\eqref{eq:concl-step3} of Step~3 under~\ref{C1}, the conclusion follows.
\qed

\subsection{Computing resonances}\label{sec:RS}
We turn to the proof of Proposition~\ref{prop:pert-res}.
Assuming that for all $\phi,\phi'\in\Pc(\Omega)$ the map
\[[0,\lambda_0)\to \C\times\Pc'(\Omega)\times \Pc'(\Omega)\times\Ld^\infty_\loc(\tfrac1MB):\lambda\mapsto\big(z_{k,\lambda},\Psi_{k,\lambda}^+,\Psi_{k,\lambda}^-,\zeta_{k,\lambda}^{\phi,\phi'}\big)\]
is of class $C^2$, we iteratively compute the first two derivatives,
\begin{eqnarray*}
\big(z_{k,0},\Psi_{k,0}^{+},\Psi_{k,0}^{-}\big)&:=&\big(z_{k,\lambda},\Psi_{k,\lambda}^+,\Psi_{k,\lambda}^-\big)\big|_{\lambda=0},\\
\big(z_{k,0}',\Psi_{k,0}^{+\prime},\Psi_{k,0}^{-\prime}\big)&:=&\tfrac{d}{d\lambda}\big(z_{k,\lambda},\Psi_{k,\lambda}^+,\Psi_{k,\lambda}^-\big)\big|_{\lambda=0},\\
\big(z_{k,0}'',\Psi_{k,0}^{+\prime\prime},\Psi_{k,0}^{-\prime\prime}\big)&:=&\big(\tfrac{d}{d\lambda}\big)^2\big(z_{k,\lambda},\Psi_{k,\lambda}^+,\Psi_{k,\lambda}^-\big)\big|_{\lambda=0}.
\end{eqnarray*}
Note that the resonant and co-resonant states $(\Psi^+_{k,\lambda},\Psi^-_{k,\lambda})$ are only defined up to multiplication by $(a_\lambda,\bar a_\lambda^{-1})$ for any complex-valued function $\lambda\mapsto a_\lambda$, cf.~\eqref{eq:dec-res}.
When differentiating, this gauge invariance implies that $(\Psi^+_{k,0},\Psi^-_{k,0})$ is only defined up to multiplication by $(\alpha_0,\bar \alpha_0^{-1})$ for any $\alpha_0\in\C$, next $(\Psi^{+\prime}_{k,0},\Psi^{-\prime}_{k,0})$ is defined up to addition of $(\alpha_1\Psi^{+}_{k,0},-\bar \alpha_1\Psi^-_{k,0})$ for any $\alpha_1\in\C$, and next $(\Psi^{+\prime\prime}_{k,0},\Psi^{-\prime\prime}_{k,0})$ is defined up to addition of $(\alpha_2\Psi^{+}_{k,0},-\bar \alpha_2\Psi^-_{k,0})$ for any $\alpha_2\in\C$.

\medskip\noindent
We first compute $z_{k,0},\Psi_{k,0}^+,\Psi_{k,0}^-$.
The resonance conjecture~\ref{C1} yields for $\Im z>0$,
\begin{equation}\label{eq:LRC-re}
\big\langle{\phi'},(H_{k,\lambda}^\st-z)^{-1}\phi\big\rangle_{\Ld^2(\Omega)}=\frac{\overline{\langle\Psi_{k,\lambda}^+,\phi'\rangle_{\Pc'(\Omega),\Pc(\Omega)}}\langle\Psi_{k,\lambda}^-,\phi\rangle_{\Pc'(\Omega,\Pc(\Omega)}}{z_{k,\lambda}-z}+\zeta_{k,\lambda}^{\phi',\phi}(z),
\end{equation}
with $\zeta_{k,\lambda}^{\phi',\phi}$ holomorphic on $\{z:\Im z>0\}\bigcup\frac1MB$.
Setting $\lambda=0$ and $\phi'=1$, we find for~$\Im z>0$,
\[-\frac1z\expec{\phi}=\expec{(H_{k,0}^\st-z)^{-1}\phi}=\frac{\overline{\langle\Psi_{k,0}^+,1\rangle_{\Pc'(\Omega),\Pc(\Omega)}}\langle\Psi_{k,0}^-,\phi\rangle_{\Pc'(\Omega,\Pc(\Omega)}}{z_{k,0}-z}+\zeta_{k,0}^{1,\phi}(z),\]
and similarly, exchanging the roles of $\phi$ and $\phi'$,
\[-\frac1{\bar z}\expec{\phi}=\frac{\overline{\langle\Psi_{k,0}^-,1\rangle_{\Pc'(\Omega),\Pc(\Omega)}}\langle\Psi_{k,0}^+,\phi\rangle_{\Pc'(\Omega,\Pc(\Omega)}}{\bar z_{k,0}-\bar z}+\overline{\zeta_{k,0}^{\phi,1}(z)}.\]
For $z\to0$, we deduce
\[z_{k,0}=0,\qquad\Psi_{k,0}^+=\alpha,\qquad\Psi_{k,0}^-=\bar \alpha^{-1},\qquad\text{for some $\alpha\in\C$}.\]
By gauge symmetry, as explained above, we can e.g.\@ choose $\alpha=1$,
\begin{equation}\label{eq:res-RS-0}
z_{k,0}=0,\qquad\Psi_{k,0}^+=\Psi_{k,0}^-=1.
\end{equation}

\medskip\noindent
Next, we compute $z_{k,0}',\Psi^{+\prime}_{k,0},\Psi^{-\prime}_{k,0}$.
Differentiating identity~\eqref{eq:LRC-re} at $\lambda=0$, using~\eqref{eq:res-RS-0}, and choosing $\phi'=1$,
we find for $\Im z>0$,
\begin{multline*}
\expec{V(H_{k,0}^\st-z)^{-1}\phi}\\
=-\frac{1}{z}z_{k,0}'\expec\phi
-\overline{\langle\Psi_{k,0}^{+\prime},1\rangle_{\Pc'(\Omega),\Pc(\Omega)}}\expec{\phi}
-\langle\Psi_{k,0}^{-\prime},\phi\rangle_{\Pc'(\Omega),\Pc(\Omega)}
+z\partial_\lambda\zeta_{k,\lambda}^{1,\phi}(z)\big|_{\lambda=0},
\end{multline*}
and similarly, exchanging the roles of $\phi$ and $\phi'$,
\begin{multline*}
\expec{V(H_{k,0}^\st-\bar z)^{-1}\phi}\\
=-\frac1{\bar z}\bar z_{k,0}'\expecm{\phi}
-\overline{\langle\Psi_{k,0}^{-\prime},1\rangle_{\Pc'(\Omega),\Pc(\Omega)}}\expecm{\phi}
-\langle\Psi_{k,0}^{+\prime},\phi\rangle_{\Pc'(\Omega),\Pc(\Omega)}
+\overline{z\partial_\lambda\zeta_{k,\lambda}^{\phi,1}(z)}\big|_{\lambda=0}.
\end{multline*}
Choosing $z=i\e$ with $\e\downarrow0$, we easily deduce $z_{k,0}'=0$ and
\begin{eqnarray*}
\expec{V(H_{k,0}^\st\mp i0)^{-1}\phi}
=-\overline{\langle\Psi_{k,0}^{\pm\prime},1\rangle_{\Pc'(\Omega),\Pc(\Omega)}}\,\expec{\phi}
-\langle\Psi_{k,0}^{\mp\prime},\phi\rangle_{\Pc'(\Omega),\Pc(\Omega)}.
\end{eqnarray*}
Noting that the left-hand side vanishes for $\phi=1$, we are led to
\begin{eqnarray*}
\langle\Psi_{k,0}^{-\prime},\phi\rangle_{\Pc'(\Omega),\Pc(\Omega)}&=&-\expecm{V(H_{k,0}^\st-i0)^{-1}\phi}+\beta\expec\phi,\\
\langle\Psi_{k,0}^{+\prime},\phi\rangle_{\Pc'(\Omega),\Pc(\Omega)}&=&-\expecm{V(H_{k,0}^\st+i0)^{-1}\phi}-\bar\beta\expec\phi,
\end{eqnarray*}
for some $\beta\in\C$. By gauge symmetry, as explained above, we can e.g.\@ choose $\beta=0$,
\begin{equation}\label{eq:res-RS-1}
z_{k,0}'=0,\qquad\Psi_{k,0}^{\pm\prime}=-(H_{k,0}^\st\mp i0)^{-1}V.
\end{equation}

\medskip\noindent
Finally, we turn to the second derivatives $z_{k,0}'',\Psi_{k,0}^{+\prime\prime},\Psi_{k,0}^{-\prime\prime}$. Differentiating identity~\eqref{eq:LRC-re} twice at $\lambda=0$, using~\eqref{eq:res-RS-0} and~\eqref{eq:res-RS-1}, and choosing $\phi'=1$, we find for $\Im z>0$,
\begin{multline*}
2\,\expecm{V(H_{k,0}^\st-z)^{-1}V(H_{k,0}^\st-z)^{-1}\phi}\\
=\frac{1}{z}z_{k,0}''\,\expec{\phi}
+\overline{\langle\Psi_{k,0}^{+\prime\prime},1\rangle_{\Pc'(\Omega),\Pc(\Omega)}}\expec{\phi}
+\langle\Psi_{k,0}^{-\prime\prime},\phi\rangle_{\Pc'(\Omega),\Pc(\Omega)}
-z\partial_\lambda^2\zeta_{k,\lambda}^{1,\phi}(z)\big|_{\lambda=0},
\end{multline*}
and similarly, exchanging the roles of $\phi$ and $\phi'$,
\begin{multline*}
2\,\expecm{V(H_{k,0}^\st-\bar z)^{-1}V(H_{k,0}^\st-\bar z)^{-1}\phi}\\
=\frac{1}{\bar z}\bar z_{k,0}''\expecm{\phi}
+\overline{\langle\Psi_{k,0}^{-\prime\prime},1\rangle_{\Pc'(\Omega),\Pc(\Omega)}}\expec{\phi}
+\langle\Psi_{k,0}^{+\prime\prime},\phi\rangle_{\Pc'(\Omega),\Pc(\Omega)}
-\overline{z\partial_\lambda^2\zeta_{k,\lambda}^{\phi,1}(z)}\big|_{\lambda=0}.
\end{multline*}
Choosing $z=i\e$ with $\e\downarrow0$, and using again the gauge invariance, we easily deduce
\begin{eqnarray*}
z_{k,0}''&=&-2\,\expecm{V(H_{k,0}^\st-i0)^{-1}V},\\
\Psi_{k,0}^{\pm\prime\prime}&=&2(H_{k,0}^\st\mp i0)^{-1}\Pi V(H_{k,0}^\st\mp i0)^{-1}V,
\end{eqnarray*}
in terms of the projection $\Pi\phi:=\phi-\expec\phi$ onto $\Ld^2(\Omega)\ominus\C$. This completes the proof.
\qed

%%%%%%%%%%%%%%%%%%%%%
%%%%%%%%%%%%%%%%%%%%%
%%%%%%%%%%%%%%%%%%%%%

\medskip
\section{An illustrative toy model}\label{sec:toy}
In this last section, we display a toy model that shares many spectral features of Schrödinger operators, but is explicitly solvable and allows for a rigorous study of its spectrum and resonances, illustrating the relevance of the resonance conjectures~\ref{C1} and~\ref{C2}.
More precisely, we replace the free Schrödinger operator $H_0=-\triangle$ by
\[\widetilde H_0:=\tfrac1i\nabla_1:=\tfrac1i\ee_1\cdot\nabla,\]
and we consider the corresponding perturbed operator
\[\widetilde H_\lambda:=\tfrac1i\nabla_1+\lambda V\]
on \mbox{$\Ld^2(\R^d\times\Omega)$}.
Via the Floquet--Bloch fibration~\eqref{eq:fibration}, this operator is decomposed as
\begin{equation}\label{eq:fibration-toy}
\big(\widetilde H_\lambda, \Ld^2(\R^d\times\Omega)\big)\,=\, \int_{\R^d}^{\oplus} \big(\widetilde H_{\lambda}^\st+k_1,\Ld^2(\Omega)\big)\,\mathfrak e_k\,\dbar k,\qquad\mathfrak e_k(x):=e^{ik\cdot x},
\end{equation}
in terms of the following (centered) fibered operator on the stationary space $\Ld^2(\Omega)$,
\[\widetilde H_{\lambda}^\st\,:=\,\widetilde H_{0}^\st+\lambda V,\qquad \widetilde H_{0}^\st\,:=\,\tfrac1i\nabla^\st_1\,:=\,\tfrac1i\ee_1\cdot\nabla^\st.\]
For this toy model, we establish the following detailed spectral properties, which are in perfect analogy with the expected situation in the Schrödinger case. Note however that the energy transport remains ballistic, cf.~item~(v) below, in stark contrast with the quantum diffusion in the Schrödinger case: this could be related to the fact that the centered fibered operator $\widetilde H_\lambda^\st$ in this toy model is not deformed under the fibration parameter $k$.

\begin{theor}[Toy model]\label{th:toy}
Assume for simplicity that $V=V_0$ is a stationary Gaussian field on with covariance $\Cc_0\in C^\infty_c(\R^d)$.
\begin{enumerate}[(i)]
\item \emph{Spectral decomposition of $\widetilde H_0^\st$:}\\
The eigenvalue at $0$ is simple (with eigenspace $\C$) and
\[\qquad\sigma_\pp(\widetilde H_0^\st)\,=\,\{0\},\qquad\sigma_\sg(\widetilde H_0^\st)\,=\,\varnothing,\qquad\sigma(\widetilde H_0^\st)\,=\,\sigma_\ac(\widetilde H_0^\st)\,=\,\R.\]
\item \emph{Spectral decomposition of $\widetilde H_\lambda^\st$:}\\
For $\lambda>0$, the eigenvalue at $0$ is fully absorbed in the absolutely continuous spectrum,
\[\qquad\sigma_\pp(\widetilde H_\lambda^\st)\,=\,\sigma_\sg(\widetilde H_\lambda^\st)\,=\,\varnothing,\qquad\sigma(\widetilde H_\lambda^\st)\,=\,\sigma_\ac(\widetilde H_\lambda^\st)\,=\,\R.\]
\item \emph{Fibered resonances:}\\
For all $\lambda>0$, the resolvent $z\mapsto (\widetilde H^\st_\lambda-z)^{-1}$ defined on $\Im z>0$ as a family of operators $\Pc(\Omega)\to\Pc'(\Omega)$ can be extended meromorphically to the whole complex plane with a unique simple pole at
\[\qquad z=-i\lambda^2 \alpha_\circ,\qquad \alpha_\circ:=\int_{0}^\infty\Cc_0(se_1)\,ds.\]
In other words, for all $\phi,\phi'\in \Pc(\Omega)$, we can write for $\Im z>0$,
\[\qquad\expecm{\bar\phi'(\widetilde H^\st_\lambda-z)^{-1}\phi}\,=\,\frac{\overline{\langle\widetilde\Psi_\lambda^+,\phi'\rangle_{\Pc'(\Omega),\Pc(\Omega)}}\langle\widetilde\Psi_\lambda^-,\phi\rangle_{\Pc'(\Omega),\Pc(\Omega)}}{-i\lambda^2\alpha_\circ-z}+\widetilde\zeta_\lambda^{\phi',\phi}(z),\]
where the remainder $\widetilde\zeta_\lambda^{\phi',\phi}$ is entire, has a continuous dependence on $\lambda\ge0$, and satisfies the uniform bound
\[\qquad|\widetilde\zeta_\lambda^{\phi',\phi}(z)|\,\lesssim_{\phi,\phi'}\,
\begin{cases}
1,&\text{if neither $\phi$ nor $\phi'$ is constant};\\
\lambda,&\text{if $\phi$ or $\phi'$ is constant};\\
\lambda^2,&\text{if $\phi$ and $\phi'$ are constant};
\end{cases}\]
and where the so-called resonant and co-resonant states $\widetilde\Psi_\lambda^+,\widetilde\Psi_\lambda^-\in \Pc'(\Omega)$ are explicitly defined, cf.~Remark~\ref{rem:def-Gamma}(a) below.
\item \emph{Continuous resonant spectrum:}\\
For $\lambda>0$ small enough, the resolvent $z\mapsto(\widetilde H_\lambda-z)^{-1}$
defined on $\Im z>0$ as a family of operators $\Ld^2(\R^d)\otimes\Pc(\Omega)\to\Ld^2(\R^d)\otimes\Pc'(\Omega)$ can be extended holomorphically to $\Im z>-\lambda^2\alpha_\circ$,
and we denote the extension by $R_{\lambda}(z)$.
For $\phi\in \Pc(\Omega)$,
this extension has the following discontinuity, as $\Im z\downarrow-\lambda^2\alpha_\circ$,
\begin{align*}
\qquad\sup_{\|g\|_{\Ld^2(\R^d)}=1}\big\langle \phi g,R_{\lambda}(z)\phi g\big\rangle_{\Ld^2(\R^d\times\Omega)}\,\sim\,i\frac{\overline{\langle\widetilde\Psi_\lambda^+,\phi\rangle_{\Pc'(\Omega),\Pc(\Omega)}}\langle\widetilde\Psi_\lambda^-,\phi\rangle_{\Pc'(\Omega),\Pc(\Omega)}}{\lambda^2\alpha_\circ+\Im z}.
\end{align*}
Next, viewed as a family of operators $\Ld^2_\comp(\R^d)\otimes\Pc(\Omega)\to\Ld^2_\loc(\R^d)\otimes\Pc'(\Omega)$, the extended resolvent $R_{\lambda}$ can be further extended to all $\C$ as an entire function.
\item \emph{Ballistic transport:}\\
For all $u^\circ\in\Ld^2(\R^d)$ with $\|xu^\circ\|_{\Ld^2(\R^d)}<\infty$, there holds
\[\qquad\lim_{t\uparrow\infty}\frac1t\,\expec{\|xu_\lambda^t\|_{\Ld^2(\R^d)}^2}^\frac12\,=\,\|u^\circ\|_{\Ld^2(\R^d)}.\qedhere\]
\end{enumerate}
\end{theor}

\begin{rems}$ $\label{rem:def-Gamma}
\begin{enumerate}[(a)]
\item\emph{Explicit formula for resonant state:}\\
Up to a gauge transformation, the resonant and co-resonant states $\widetilde\Psi_\lambda^+,\widetilde\Psi_\lambda^-\in \Pc'(\Omega)$ in item~(iii) take the form
\[\widetilde\Psi_\lambda^\pm\,=\,e^{\frac12\lambda^2\int_0^\infty s\,\Cc_0(s\ee_1)\,ds}\,\widetilde\Psi_{\lambda}^{\circ,\pm},\]
where $\widetilde\Psi_{\lambda}^{\circ,\pm} \in \Pc'(\Omega)$ is formally defined as
\begin{equation}\label{eq:def0-Gamma}
\widetilde\Psi_{\lambda}^{\circ,\pm}~~\text{``$=$''}~~\frac{e^{\pm\frac\lambda i\int_0^\infty V(\mp s\ee_1,\cdot)\,ds}}{\expecm{e^{\pm\frac\lambda i\int_0^\infty V(\mp s\ee_1,\cdot)\,ds}}}.
\end{equation}
More precisely, the action of $\widetilde\Psi_{\lambda}^{\circ,\pm}$ on $\Pc(\Omega)$ is defined inductively on monomials of increasing degree: we set $\langle\widetilde\Psi_{\lambda}^{\circ,\pm},1\rangle_{\Pc'(\Omega),\Pc(\Omega)}=1$, and for all $n\ge1$ and $x_1,\ldots,x_n\in\R^d$,
\begingroup\allowdisplaybreaks
\begin{multline}\label{eq:def-Gamma}
\qquad\bigg\langle\widetilde\Psi_{\lambda}^{\circ,\pm},\prod_{j=1}^nV(x_j,\cdot)\bigg\rangle_{\Pc'(\Omega),\Pc(\Omega)}\,=\,\sum_{l=2}^n\Cc_0(x_1-x_l)\,\bigg\langle\widetilde\Psi_{\lambda}^{\circ,\pm},\prod_{2\le j\le n\atop j\ne l}V(x_j,\cdot)\bigg\rangle_{\Pc'(\Omega),\Pc(\Omega)}\\
\mp\frac\lambda i\Big(\int_{0}^\infty\Cc_0(x_1\pm se_1)\,ds\Big)\,\bigg\langle\widetilde\Psi_{\lambda}^{\circ,\pm},\prod_{j=2}^n V(x_j,\cdot)\bigg\rangle_{\Pc'(\Omega),\Pc(\Omega)},
\end{multline}
\endgroup
while the formal representation~\eqref{eq:def0-Gamma} is understood in view of Wick's theorem.
In particular, there holds $\langle\widetilde\Psi_\lambda^\pm,\phi\rangle_{\Pc'(\Omega),\Pc(\Omega)}=\expec{\phi}\pm i\lambda\int_0^\infty\expec{V(\mp se_1,\cdot)\,\phi}ds+O_\phi(\lambda^2)$.
\item\emph{Mourre's approach:}\\
In view of Proposition~\ref{prop:commut}(ii)--(iii), with the notation of Section~\ref{sec:oper-L2},
the commutator $[\widetilde H_\lambda,\tfrac1i\Op(x_1)]$ is $\Lc$-bounded and satisfies the lower bound
\[[\widetilde H_\lambda^\st,\tfrac1i\Op(x_1)]\,\ge\,\Lc-C\lambda\Lc^{\frac12}.\]
In other words, $\tilde H_\lambda^\st$ satisfies a Mourre relation with conjugate $\Op(x_1)$ ``up to $\Lc$''.
Much spectral information can be inferred from such a property, and in particular another proof of Theorem~\ref{th:toy} above could essentially be deduced.
This Mourre approach would be particularly useful in the discrete setting, that is, for the discrete operator $\widetilde H_\lambda^\st=\frac1i \nabla_1+\lambda V$ on $\ell^2(\Z^d)$ with an i.i.d.\@ Gaussian field $V$ on $\Z^d$: the proof below can indeed not be adapted to that case as the flow is not explicit.
\qedhere
\end{enumerate}
\end{rems}

\begin{proof}[Proof of Theorem~\ref{th:toy}]$ $
Item~(i) is obtained similarly as in the Schrödinger case, cf.~Section~\ref{sec:spectrum-Hk0}, and the proof is omitted. Item~(ii) is a direct consequence of~(iii). It remains to establish items~(iii), (iv), and~(v). Without loss of generality, we assume that $\Cc_0$ is supported in $B$.

\medskip
\step1 Proof of~(iii).\\
It suffices to show for all $\phi,\phi'\in \Pc(\Omega)$ and $t\ge0$,
\begin{align}\label{eq:explicit-evol-toy}
\Big|\expecm{\bar\phi'\,e^{-it\widetilde H_{\lambda}^\st}\phi}-\overline{\langle\widetilde\Psi_\lambda^+,\phi'\rangle_{\Pc'(\Omega),\Pc(\Omega)}}\langle\widetilde\Psi_\lambda^-,\phi\rangle_{\Pc'(\Omega),\Pc(\Omega)}\,e^{-t\lambda^2\alpha_\circ}\Big|\,\lesssim_{\phi,\phi'}\,\mathds1_{|t|\le C_{\phi,\phi'}},
\end{align}
where we gain a factor $\lambda$ in the right-hand side if $\phi$ or $\phi'$ is constant, and a factor $\lambda^2$ if both are constants.
Indeed, for $\Im z>0$, we can write
\[\expecm{\bar\phi'(\widetilde H_\lambda^\st-z)^{-1}\phi}\,=\,i\int_0^\infty e^{itz}\,\expecm{\bar\phi'\,e^{-it\widetilde H_\lambda^\st}\phi}\,dt,\]
so that the conclusion~(iii) would follow from~\eqref{eq:explicit-evol-toy} after integration.
By linearity, it suffices to establish~\eqref{eq:explicit-evol-toy} for monomials
\[\phi=\prod_{j=1}^nV(x_j,\cdot),\qquad\phi'=\prod_{j=1}^mV(y_j,\cdot).\]
Noting that the fibered evolution $\phi_\lambda^t:=e^{-it\widetilde H_\lambda^\st}\phi$ can be explicitly computed as
\[\phi_\lambda^t(\omega)\,=\,\phi(-t\ee_1,\omega)\,\psi_\lambda^t(\omega),\qquad \psi_\lambda^t(\omega)\,:=\,\exp\Big(\frac \lambda i \int_0^tV(-s\ee_1,\omega)\,ds\Big),\]
we find
\[\expecm{\phi'\,e^{-it\widetilde H_{\lambda}^\st}\phi}=\E\bigg[{\Big(\prod_{j=1}^mV(y_j,\cdot)\Big)\Big(\prod_{j=1}^nV(x_j-t\ee_1,\cdot)\Big)\,\psi_\lambda^t}\bigg].\]
By Wick's formula, for $m\ge1$, we compute
\begin{multline*}
\expecm{\phi'\,e^{-it\widetilde H_{\lambda}^\st}\phi}=\sum_{l=2}^m\Cc_0(y_1-y_l)\,\E\bigg[{\Big(\prod_{2\le j\le m\atop j\ne l}V(y_j,\cdot)\Big)\Big(\prod_{j=1}^nV(x_j-t\ee_1,\cdot)\Big)\,\psi_\lambda^t}\bigg]\\
+\sum_{l=2}^n\Cc_0(y_1-x_l+t\ee_1)\,\E\bigg[{\Big(\prod_{j=2}^mV(y_j,\cdot)\Big)\Big(\prod_{1\le j\le n\atop j\ne l}V(x_j-t\ee_1,\cdot)\Big)\,\psi_\lambda^t}\bigg]\\
+\frac\lambda i\Big(\int_{0}^{t}\Cc_0(y_1+s\ee_1)\,ds\Big)\,\E\bigg[{\Big(\prod_{j=2}^mV(y_j,\cdot)\Big)\Big(\prod_{j=1}^nV(x_j-te_1,\cdot)\Big)\,\psi_\lambda^t}\bigg].
\end{multline*}
Since by assumption
\[|\Cc_0(y_1-x_l+te_1)|\lesssim\mathds1_{|t|\le 1+|y_1|+|x_l|},\]
and similarly $\big|\int_{t}^\infty\Cc_0(y_1+s\ee_1)\,ds\big|\lesssim\mathds1_{|t|\le 1+|y_1|}$, we deduce
\begin{multline*}
\expecm{\phi'\,e^{-it\widetilde H_{\lambda}^\st}\phi}=\sum_{l=2}^m\Cc_0(y_1-y_l)\,\E\bigg[{\Big(\prod_{2\le j\le m\atop j\ne l}V(y_j,\cdot)\Big)\Big(\prod_{j=1}^nV(x_j-t\ee_1,\cdot)\Big)\,\psi_\lambda^t}\bigg]\\
+\frac\lambda i\Big(\int_{0}^{\infty}\Cc_0(y_1+s\ee_1)\,ds\Big)\,\E\bigg[{\Big(\prod_{j=2}^mV(y_j,\cdot)\Big)\Big(\prod_{j=1}^{n}V(x_j-t\ee_1,\cdot)\Big)\,\psi_\lambda^t}\bigg]
\,+\,O_{\phi,\phi'}(\mathds1_{|t|\le C_{\phi,\phi'}}).
\end{multline*}
We recognize here the inductive definition~\eqref{eq:def-Gamma} of the resonant state $\widetilde\Psi_\lambda^+$, so that the above becomes
\[\expecm{\phi'\,e^{-it\widetilde H_{\lambda}^\st}\phi}=\overline{\langle\widetilde\Psi_{\lambda,\circ}^+,\phi'\rangle_{\Pc'(\Omega),\Pc(\Omega)}}\,\E\bigg[{\Big(\prod_{j=1}^nV(x_j-t\ee_1,\cdot)\Big)\,\psi_\lambda^t}\bigg]+\,O_{\phi,\phi'}(\mathds1_{|t|\le C_{\phi,\phi'}}),\]
and a similar computation yields
\[\expecm{\phi'\,e^{-it\widetilde H_{\lambda}^\st}\phi}=\overline{\langle\widetilde\Psi_{\lambda,\circ}^+,\phi'\rangle_{\Pc'(\Omega),\Pc(\Omega)}}\langle\widetilde\Psi_{\lambda,\circ}^-,\phi\rangle_{\Pc'(\Omega),\Pc(\Omega)}\,\expec{\psi_\lambda^t}+\,O_{\phi,\phi'}(\mathds1_{|t|\le C_{\phi,\phi'}}).\]
Finally, since $\int_0^tV(-s\ee_1,\cdot)\,ds$ is Gaussian, we compute
\begin{eqnarray*}
\expecm{\psi_\lambda^t}\,=\,e^{-\frac12\lambda^2\var{\int_0^tV(-s\ee_1,\cdot)ds}}&=&e^{-\lambda^2\int_0^t(t-s)\Cc_0(s\ee_1)ds}\\
&=&e^{-\lambda^2t\int_0^\infty\Cc_0(s\ee_1)ds}e^{\lambda^2\int_0^\infty s\,\Cc_0(s\ee_1)ds}+O_{M}(\lambda^2\mathds1_{|t|\le1}),
\end{eqnarray*}
and the conclusion~\eqref{eq:explicit-evol-toy} follows.

\medskip
\step2 Proof of~(iv).\\
For $g,g'\in\Ld^2(\R^d)$ and $\phi,\phi'\in \Pc(\Omega)$, the Floquet--Bloch fibration~\eqref{eq:fibration-toy} yields
\[\big\langle \phi'g',(\widetilde H_\lambda-z)^{-1}(\phi g)\big\rangle_{\Ld^2(\R^d\times\Omega)}\,=\,\int_{\R^d} \overline{\widehat g'(k)}\widehat g(k)\, \big\langle\phi',(\widetilde H_{\lambda}^\st+k_1-z)^{-1}\phi\big\rangle_{\Ld^2(\Omega)}\,\dbar k,\]
and thus, inserting the result of item~(iii), for $\Im z>0$,
\begin{multline*}
\big\langle \phi'g',(\widetilde H_\lambda-z)^{-1}(\phi g)\big\rangle_{\Ld^2(\R^d\times\Omega)}
\,=\,\int_{\R^d} \overline{\widehat g'(k)}\widehat g(k)\,\widetilde\zeta_\lambda^{\phi',\phi}(z-k_1)\,\dbar k\\
+\overline{\langle\widetilde\Psi_\lambda^+,\phi'\rangle_{\Pc'(\Omega),\Pc(\Omega)}}\langle\widetilde\Psi_\lambda^-,\phi\rangle_{\Pc'(\Omega),\Pc(\Omega)}\int_{\R^d} \frac{\overline{\widehat g'(k)}\widehat g(k)}{k_1-i\lambda^2\alpha_\circ-z}\dbar k,
\end{multline*}
where the first right-hand side term is entire and the second is holomorphic on $\Im z>-\lambda^2\alpha_\circ$.
Next, for $g,g'\in \Ld^2_\comp(\R^d)$, the Fourier transforms $\widehat g,\widehat g'$ are entire functions, which allows to extend the second right-hand side term holomorphically to the whole complex plane. Indeed, for $y\in\R$, the Sokhotski--Plemelj formula yields
\[\lim_{\e\downarrow0}\int_{\R^d}\frac{\overline{\widehat g'(k)}\widehat g(k)}{k_1-y\mp i\e}\dbar k\,=\,\pv\int_{\R^d}\frac{\overline{\widehat g'(k)}\widehat g(k)}{k_1-y}\dbar k\pm i\pi\int_{\R^{d-1}}\overline{\widehat g'(y,k')}\widehat g(y,k')\,\dbar k',\]
so that the function $T_\lambda(z)$ defined for $\Im z>-\lambda^2\alpha_\circ$ by
\[T_\lambda(z)=\int_{\R^d} \frac{\overline{\widehat g'(k)}\widehat g(k)}{k_1-i\lambda^2\alpha_\circ-z}\dbar k,\]
and defined for $\Im z<-\lambda^2\alpha_\circ$ by
\[T_\lambda(z)=\int_{\R^d} \frac{\overline{\widehat g'(k)}\widehat g(k)}{k_1-i\lambda^2\alpha_\circ-z}\dbar k+2i\pi\int_{\R^{d-1}}\overline{\widehat g'(z+i\lambda^2\alpha_\circ,k')}\widehat g(z+i\lambda^2\alpha_\circ,k')\,\dbar k',\]
is entire. This proves~(iv).

\medskip
\step3 Proof of~(v).\\
The flow can be explicitly integrated,
\[u_\lambda^t(x)=u^\circ(x-t\ee_1)\,e^{-i\lambda\int_0^tV(x-s\ee_1)ds},\]
and is seen to satisfy ballistic transport,
\[\frac1t\,\expec{\|xu_\lambda^t\|_{\Ld^2(\R^d)}^2}^\frac12\,=\,\frac1t\,\expec{\int_{\R^d}|x|^2|u^\circ(x-t\ee_1)|^2dx}^\frac12\,\xrightarrow{t\uparrow\infty}\,\|u^\circ\|_{\Ld^2(\R^d)}.\qedhere\]
\end{proof}

%%%%%%%%%%%%%%%%%%%%%
%%%%%%%%%%%%%%%%%%%%%
%%%%%%%%%%%%%%%%%%%%%

\medskip
\appendix
\section{Self-adjointness with unbounded potentials}\label{sec:self-adj}
For a {bounded} random potential $V$,
the Schrödinger operator $H_\lambda=-\triangle+\lambda V$ on \mbox{$\Ld^2(\R^d\times\Omega)$} is clearly self-adjoint on $\Ld^2(\Omega;H^2(\R^d))$
and the fibered operators $\{H_{k,\lambda}^\st\}_k$ are self-adjoint on $H^2(\Omega)$ just as for~$\lambda=0$.
The present appendix is concerned with the corresponding self-adjointness statement in the unbounded setting.
More precisely, we show that essential self-adjointness still holds provided that $V\in\Ld^2(\Omega)$ has negative part $V_-\in\Ld^p(\Omega)$ for some $p>\frac d2$. This condition is essentially optimal and applies in particular to the case when $V=V_0$ is a stationary Gaussian field. (Note however that this Gaussian case is much simpler in view of Malliavin calculus and can be obtained as a consequence of Nelson's theorem in form of Proposition~\ref{prop:nelson} with $L=\Lc$.)

\medskip
A random potential $V\in\Ld^2(\Omega)$ defines (densely) a multiplicative operator on \mbox{$\Ld^2(\R^d\times\Omega)$}.
If $V$ is not uniformly bounded, this operator is unbounded, so that the self-adjointness of $-\triangle+\lambda V$ is a subtle question and may fail, cf.~\cite{Kato-72}.
Whenever realizations $V_\omega=V(\cdot,\omega)$ are quadratically controlled from below, in the sense of $V(x,\omega)\ge -M(\omega)\,(1+|x|^2)$ for some random variable $M\in\Ld^2(\Omega)$, the Faris-Lavine argument~\cite{Faris-Lavine-74} ensures that the Schr\"{o}dinger operator \mbox{$-\triangle+\lambda V$} on $\Ld^2(\R^d\times\Omega)$ is essentially self-adjoint on $C^\infty_c(\R^d;\Ld^\infty(\Omega))$.
By a Borel-Cantelli argument, the quadratic lower bound holds whenever
the negative part~$V_-$ belongs to $\Ld^p(\Omega)$ for some~$p>\frac d2$ and satisfies more precisely $\sup_{|x|\le1}V_-(x,\cdot)\in\Ld^p(\Omega)$.
In this setting, since $-\triangle+\lambda V$ is essentially self-adjoint on $C^\infty_c(\R^d;\Ld^\infty(\Omega))$, we may repeat the direct integral decomposition~\eqref{eq:fibration-re+} and the fibered operators $H^\st_{k,\lambda}$ on $\Ld^2(\Omega)$ are necessarily essentially self-adjoint on $H^2\cap\Ld^\infty(\Omega)$ for almost all $k\in\R^d$, e.g.~\cite[p.280]{Reed-Simon-78}.
In order to conclude for all $k$, some continuity would be needed, which typically requires smoothness of $V$. In order to avoid such spurious assumptions, we provide another argument below. While the usual Faris-Lavine argument is of no use in the stationary space $\Ld^2(\Omega)$, we draw inspiration from an earlier work by Kato~\cite{Kato-72}.

\begin{theor}[Essential self-adjointness]\label{th:app-sa}
Assume that the potential $V\in\Ld^2(\Omega)$ satisfies $\sup_{|x|\le1} V_-(\tau_x\cdot)\in\Ld^{p}(\Omega)$ for some $p>\frac d2$. Then for all $\lambda\ge0$ and $k\in\R^d$ the operator $H_{k,\lambda}^\st$ is essentially self-adjoint on $H^2\cap\Ld^\infty(\Omega)$.
\end{theor}

\begin{proof}
Let $k\in\R^d$ and $\lambda\ge0$ be fixed.
For $V\in\Ld^2(\Omega)$, we note that the operator $H_{k,\lambda}^\st$ is well-defined on the whole of $\Ld^2(\Omega)$ with values in the space $\Dc':=\Ld^1(\Omega)+H^{-2}(\Omega)$ (cf.~Section~\ref{sec:stat-calc} for notation), and it is obviously continuous $\Ld^2(\Omega)\to\Dc'$.
Let $\dot H_{k,\lambda}^\st$ denote the restriction of $H_{k,\lambda}^\st$ with domain $\Dc:=H^2\cap\Ld^\infty(\Omega)$. Since $H_{k,\lambda}^\st\Dc\subset\Ld^2(\Omega)$, the operator $\dot H_{k,\lambda}^\st$ can be viewed as a densely defined operator on $\Ld^2(\Omega)$ and it is clearly symmetric. Its adjoint $(\dot H_{k,\lambda}^\st)^*$ is easily seen as the restriction of $H_{k,\lambda}^\st$ to $\Ld^2(\Omega)$, that is, defined whenever $u\in \Ld^2(\Omega)$ and $H_{k,\lambda}^\st\in \Ld^2(\Omega)$.
In this context, the following conditions are equivalent:
\begin{enumerate}
\item[(E1)] $\dot H_{k,\lambda}^\st$ is essentially self-adjoint.
\item[(E2)] There exist two complex numbers $z_\pm$ with $\Im z_\pm\gtrless0$ such that $H_{k,\lambda}^\st+z_\pm$ is an injection of $\Ld^2(\Omega)$ into $\Dc'$.
\item[(E3)] The restriction of $H_{k,\lambda}^\st$ to $\Ld^2(\Omega)$ is the strong closure of $\dot H_{k,\lambda}^\st$, that is, for all $\phi\in\Ld^2(\Omega)$ with $ H_{k,\lambda}^\st\phi\in\Ld^2(\Omega)$ there exists a sequence $(\phi_n)_n\subset\Dc$ such that $\phi_n\to\phi$ and $H_{k,\lambda}^\st\phi_n\to H_{k,\lambda}^\st\phi$ in $\Ld^2(\Omega)$.
\end{enumerate}
We proceed by truncation: we define the truncated operator $H_{k,\lambda;R}^\st:=H_{k,0}^\st+\lambda V\mathds1_{V\ge-R}$ for $R\ge1$, and we split the proof into four steps.

\medskip
\step1
Proof that $H_{k,\lambda;R}^\st$ is essentially self-adjoint.\\
Assume that $\phi\in\Ld^2(\Omega)$ satisfies $(H_{k,\lambda;R}^\st+z)\phi=0$ in $\Dc'$.
Applying the differential inequality of~\cite[Lemma~A]{Kato-72} in the form
\[-\triangle^\st|\phi|\,\le\,-\Re\big(\tfrac{\bar \phi}{|\phi|}(\nabla^\st+ik)\cdot(\nabla^\st+ik)\phi\big)\,=\,\Re\big(\tfrac{\bar \phi}{|\phi|}(H_{k,0}^\st+|k|^2)\phi\big),\]
we deduce
\begin{eqnarray*}
-\triangle^\st|\phi|\,\le\,\Re\big(|k|^2-z-\lambda V\mathds1_{V\ge-R}\big)|\phi|
\,\le\,(|k|^2+\lambda R-\Re z)|\phi|.
\end{eqnarray*}
For $\Re z$ large enough, we have $c:=\Re z-|k|^2-\lambda R\ge1$ and 
\begin{eqnarray*}
(c-\triangle^\st)|\phi|&\le&0.
\end{eqnarray*}
Since the operator $-\triangle^\st$ is nonnegative, this implies $\phi=0$.
Hence, the operator $H_{k,\lambda;R}^\st+z$ is an injection of $\Ld^2(\Omega)$ into $\Dc'$, and the claim follows from the equivalence between (E1) and (E2).

\medskip
\step2
For all $\alpha\ge0$ and $R\ge1$, there exists a cut-off function $\chi_R^\alpha\in\Ld^2(\Omega;[0,1])$ with the following properties:
\begin{enumerate}[(i)]
\item $\chi_R^\alpha=0$ on $E_R:=\{\omega:\exists y\in 4B~\text{such that}~V(y,\omega)\le-R\}$;
\item $\chi_R^\alpha=1$ outside $E_R^\alpha:=\{\omega:\exists\, y\in(R^\alpha+7)B~\text{with}~V(y,\omega)\le-R\}$;
\item $|\nabla^\st\chi_R^\alpha|\lesssim R^{-\alpha}$, $|(\nabla^\st)^2\chi_R^\alpha|\lesssim R^{-\alpha}$;
\item there exists $\e>0$ such that $\chi_R^\alpha\to1$ almost surely as $R\uparrow\infty$ whenever $\alpha\le\frac12+\e$.
\end{enumerate}
Define $E_{R,\omega}:=\{x\in\R^d:\exists y\in B_4(x)~\text{such that}~V(y,\omega)\le-R\}$. Choose $\rho\in C^\infty_c(\R^d)$ with $\int_{\R^d}\rho=1$, $\rho\ge0$, $\rho=0$ outside $2B$, $|\nabla\rho|\lesssim1$, and $|\nabla^2\rho|\lesssim1$, and choose an even function $\chi\in C^\infty(\R)$ with $\chi(0)=0$, $\chi(s)=1$ for $|s|\ge1$, $|\chi'|\lesssim1$, and $|\chi''|\lesssim1$. We then construct the stationary function
\[\chi_R^\alpha(x,\omega):=\chi\Big(\frac1{R^\alpha}\int_{\R^d}\rho(x-y)\,\operatorname{dist}(y,E_{R,\omega}+2B)\,dy\Big).\]
Properties~(i)--(iii) easily follow for this choice. We turn to~(iv). The definition of $\chi_R^\alpha$, a union bound, and Markov's inequality yield
\begin{eqnarray*}
\pr{\chi_R^\alpha<1}&\le&\p\Big\{\omega:\inf_{y\in 2B}\operatorname{dist}(y,E_{R,\omega}+2B)\le R^\alpha\big\}\\
&\le&\p\big\{\omega:\operatorname{dist}(0,E_{R,\omega})\le R^\alpha+4\big\}\\
&\le&\p\big\{\omega:\exists y\in B_{R^\alpha+8},\,V(y,\omega)\le -R\big\}\\
&\lesssim& R^{\alpha d}\,\prm{\textstyle\inf_BV\le-R}\\
&\le& R^{\alpha d-p}\,\expecm{\textstyle(\inf_BV)_-^p}.
\end{eqnarray*}
Since by assumption $\expecm{(\textstyle\inf_BV)_-^p}<\infty$ for $p=\frac d2+\e$ for some $\e>0$, we deduce $\chi_R^\alpha\to1$ in measure whenever $\alpha<\frac12+\frac \e d$.
In order to establish almost sure convergence, we similarly compute, noting that $E_{R,\omega}$ is decreasing in $R$,
\begin{eqnarray*}
\pr{\chi_R^\alpha\not\to1}&\le&\lim_{R_0\uparrow\infty}\p\Big\{\omega:\exists R\ge R_0~\text{such that}~\operatorname{dist}(0,E_{R,\omega})\le R^\alpha+4\Big\}\\
&\le&\lim_{R_0\uparrow\infty}\sum_{n=0}^\infty\p\Big\{\omega:\operatorname{dist}(0,E_{2^nR_0,\omega})\le (2^{n+1}R_0)^\alpha+4\Big\}\\
&\lesssim&\lim_{R_0\uparrow\infty}\sum_{n=0}^\infty(2^nR_0)^{\alpha d-p}\,\expecm{(\textstyle\inf_BV)_-^p},
\end{eqnarray*}
and almost sure convergence $\chi_R^\alpha\to1$ follows under the same condition on $\alpha$.

\medskip
\step3 G\aa rding inequalities:
\begin{enumerate}[\quad(G1)]
\item For all $\phi\in\Ld^2(\Omega)$ with $H_{k,\lambda;R}^\st\phi\in H^{-1}(\Omega)$, there holds $\phi\in H^1(\Omega)$ and
\[\|\nabla^\st \phi\|_{\Ld^2(\Omega)}^2\,\le\,4\|H_{k,\lambda;R}^\st\phi\|_{H^{-1}(\Omega)}^2+(1+8|k|^2+4\lambda R)\|\phi\|_{\Ld^2(\Omega)}^2.\]
\item For all $\phi\in\Ld^2(\Omega)$ with $H_{k,\lambda}^\st\phi\in \Ld^2(\Omega)$, there holds
\[\|\mathds1_{\Omega\setminus E_R}\nabla^\st\phi\|_{\Ld^2(\Omega)}^2\,\lesssim\,\|H_{k,\lambda}^\st \phi\|_{\Ld^2(\Omega)}^2+(1+|k|^2+\lambda R)\|\phi\|_{\Ld^2(\Omega)}^2,\]
where as in Step~2 we have set $E_R:=\{\omega:\exists\, y\in4B~\text{with}~V(y,\omega)\le-R\}$.
\end{enumerate}
By density, it suffices to argue for $\phi\in\Dc$.
We start with the G\aa rding inequality~(G1) for the truncated operator~$H_{k,\lambda;R}^\st$.
For $\phi\in\Dc$, we compute
\begin{eqnarray*}
\Re\langle\phi,H_{k,\lambda;R}^\st\phi\rangle_{\Ld^2(\Omega)}&\ge&\Re\langle\phi,H_{k,0}^\st\phi\rangle_{\Ld^2(\Omega)}-\lambda R\|\phi\|_{\Ld^2(\Omega)}^2\\
&\ge&\frac12\|\nabla^\st\phi\|_{\Ld^2(\Omega)}^2-(2|k|^2+\lambda R)\|\phi\|_{\Ld^2(\Omega)}^2,
\end{eqnarray*}
hence,
\begin{equation*}
\frac12\|\nabla^\st\phi\|_{\Ld^2(\Omega)}^2
\,\le\,\|H_{k,\lambda;R}^\st\phi\|_{H^{-1}(\Omega)}^2+\frac14\|\phi\|_{H^1(\Omega)}^2+(2|k|^2+\lambda R)\|\phi\|_{\Ld^2(\Omega)}^2,
\end{equation*}
and the claim~(G1) follows.

\medskip\noindent
We turn to~(G2).
Similarly as in Step~2, we may construct a cut-off function $\chi_R'$ with the following properties
\begin{enumerate}[(i')]
\item $\chi_R'=0$ on $E_R':=\{\omega:V(\omega)\le-R\}$;
\item $\chi_R'=1$ outside $E_R$;
\item $|\nabla^\st\chi_R'|\lesssim1$, $|(\nabla^\st)^2\chi_R'|\lesssim1$.
\end{enumerate}
Noting that $H_{k,\lambda}^\st(\phi\chi_R')=H_{k,\lambda;R}^\st(\phi\chi_R')$, the result~(G1) yields
\[\|\nabla^\st(\phi\chi_R')\|_{\Ld^2(\Omega)}^2\,\le\,4\|H_{k,\lambda}^\st(\phi\chi_R')\|_{H^{-1}(\Omega)}^2+(1+8|k|^2+4\lambda R)\|\phi\|_{\Ld^2(\Omega)}^2.\]
Computing
\[H_{k,\lambda}^\st(\phi\chi_R')\,=\,\chi_R'H_{k,\lambda}^\st\phi+\phi H_{k,0}^\st \chi_R'-2\nabla^\st\chi_R'\cdot\nabla^\st\phi,\]
and noting that $|H_{k,0}^\st\chi_R'|\lesssim1+|k|$, we deduce
\[\|\nabla^\st(\phi\chi_R')\|_{\Ld^2(\Omega)}^2\,\lesssim\,\|H_{k,\lambda}^\st \phi\|_{\Ld^2(\Omega)}^2+\|\nabla^\st\chi_R'\cdot\nabla^\st \phi\|_{H^{-1}(\Omega)}^2+(1+|k|^2+\lambda R)\|\phi\|_{\Ld^2(\Omega)}^2.\]
Since for $\phi'\in H^1(\Omega)$ integration by parts yields
\[\big|\big\langle \phi',\nabla^\st\chi_R'\cdot\nabla^\st \phi\big\rangle_{\Ld^2(\Omega)}\big|\,=\,\big|\big\langle \nabla^\st \phi'\cdot \nabla^\st\chi_R'+\phi'\triangle^\st\chi_R'\,,\,\phi\big\rangle_{\Ld^2(\Omega)}\big|\,\lesssim\,\|\phi'\|_{H^1(\Omega)}\|\phi\|_{\Ld^2(\Omega)},\]
we find
\[\|\nabla^\st\chi_R'\cdot\nabla^\st \phi\|_{H^{-1}(\Omega)}\,\lesssim\,\|\phi\|_{\Ld^2(\Omega)},\]
hence,
\[\|\nabla^\st(\phi\chi_R')\|_{\Ld^2(\Omega)}^2\,\le\,\|H_{k,\lambda}^\st\phi\|_{\Ld^2(\Omega)}^2+(1+|k|^2+\lambda R)\|\phi\|_{\Ld^2(\Omega)}^2.\]
Since $\chi_R'=1$ outside $E_R$, the claim~(G2) follows.

\medskip
\step4 Conclusion.
\\
Let $\phi\in\Ld^2(\Omega)$ with $H_{k,\lambda}^\st\phi\in\Ld^2(\Omega)$.
In view of the equivalence between properties~(E1) and~(E3), it suffices to construct a sequence $(\phi_n)_n\subset\Dc$ such that $\phi_n\to\phi$ and $H_{k,\lambda}^\st\phi_n\to H_{k,\lambda}^\st\phi$ in $\Ld^2(\Omega)$ as $n\uparrow\infty$.
We argue by truncation.
Let $\chi_R^\alpha$ be the cut-off function defined in Step~2 and choose $\alpha>\frac12$ such that property~(iv) is satisfied. We show that for all $R\ge1$ there exists a sequence $(\phi_{n,R})_n\subset\Dc$ such that
\begin {samepage}
\begin{enumerate}[\qquad$\bullet$]
\item $\phi_{n,R}\chi_R^\alpha\to \phi(\chi_R^\alpha)^2$ and $H_{k,\lambda}^\st(\phi_{n,R}\chi_R^\alpha)\to H_{k,\lambda}^\st (\phi(\chi_R^\alpha)^2)$ in $\Ld^2(\Omega)$ as~$n\uparrow\infty$;
\item $\phi(\chi_R^\alpha)^2\to\phi$ and $H_{k,\lambda}^\st (\phi(\chi_R^\alpha)^2)\to H_{k,\lambda}^\st\phi$ in $\Ld^2(\Omega)$ as $R\uparrow\infty$.
\end{enumerate}
\end{samepage}
We split the proof into three further substeps.

\medskip
\substep{4.1} Proof that for all $R\ge1$ there exists a sequence $(\phi_{n,R})_n\subset\Dc$ such that $\phi_{n,R}\to \phi\chi_R^\alpha$ and $H_{k,\lambda;R}^\st \phi_{n,R}\to H_{k,\lambda}^\st(\phi\chi_R^\alpha)$ in $\Ld^2(\Omega)$ as $n\uparrow\infty$.

\medskip\noindent
For all $R\ge1$, since by Step~1 the operator $H_{k,\lambda;R}^\st$ is essentially self-adjoint, the equivalence between properties~(E1) and~(E3) implies that there exists a sequence $(\phi_{n,R})_n$ such that $\phi_{n,R}\to \phi\chi_R^\alpha$ and $H_{k,\lambda;R}^\st\phi_{n,R}\to H_{k,\lambda;R}^\st(\phi\chi_R^\alpha)$ in $\Ld^2(\Omega)$ as $n\uparrow\infty$.
By definition of $\chi_R^\alpha$, there holds $H_{k,\lambda}^\st(\phi\chi_R^\alpha)=H_{k,\lambda;R}^\st(\phi\chi_R^\alpha)$, and the claim follows.

\medskip
\substep{4.2} Proof that for all $R\ge1$ there holds $\phi_{n,R}\chi_R^\alpha\to \phi(\chi_R^\alpha)^2$ and $H_{k,\lambda}^\st(\phi_{n,R}\chi_R^\alpha)\to H_{k,\lambda}^\st(\phi(\chi_R^\alpha)^2)$ in $\Ld^2(\Omega)$ as $n\uparrow\infty$.

\medskip\noindent
We start from the identity
\begin{multline*}
H_{k,\lambda}^\st\big(\phi_{n,R}\chi_R^\alpha\big)-H_{k,\lambda}^\st\big(\phi(\chi_R^\alpha)^2\big)=\chi_R^\alpha H_{k,\lambda}^\st (\phi_{n,R}-\phi\chi_R^\alpha)+(\phi_{n,R}-\phi\chi_R^\alpha)H_{k,0}^\st\chi_R^\alpha\\
-2\nabla^\st\chi_R^\alpha\cdot \nabla^\st (\phi_{n,R}-\phi\chi_R^\alpha),
\end{multline*}
and note that the convergence of $(\phi_{n,R})_n$ (cf.~Substep~4.1) implies $\phi_{n,R}\chi_R^\alpha\to \phi(\chi_R^\alpha)^2$ and
\[\limsup_{n\uparrow\infty}\big\|H_{k,\lambda}^\st\big(\phi_{n,R}\chi_R^\alpha\big)-H_{k,\lambda}^\st\big(\phi(\chi_R^\alpha)^2\big)\big\|_{\Ld^2(\Omega)}\,\lesssim\,
\|\nabla^\st (\phi_{n,R}-\phi\chi_R^\alpha)\|_{\Ld^2(\Omega)}.\]
Combining this with the G\aa rding inequality~(G1) of Step~3 and with the convergence properties of $(\phi_{n,R})_n$, the claim follows.

\medskip
\substep{4.3} Proof that $\phi(\chi_R^\alpha)^2\to\phi$ and $H_{k,\lambda}^\st(\phi(\chi_R^\alpha)^2)\to H_{k,\lambda}^\st\phi$ in $\Ld^2(\Omega)$ as $R\uparrow\infty$.

\medskip\noindent
We start from the identity
\[H_{k,\lambda}^\st\big(\phi(\chi_R^\alpha)^2\big)=(\chi_R^\alpha)^2 H_{k,\lambda}^\st\phi+\phi H_{k,0}^\st(\chi_R^\alpha)^2-2\nabla^\st(\chi_R^\alpha)^2\cdot \nabla^\st\phi,\]
and note that the properties of $\chi_R^\alpha$ and the dominated convergence theorem lead to $\phi(\chi_R^\alpha)^2\to \phi$, $(\chi_R^\alpha)^2 H_{k,\lambda}^\st\phi\to H_{k,\lambda}^\st\phi$, and $\phi H_{k,0}^\st(\chi_R^\alpha)^2\to0$ in $\Ld^2(\Omega)$, hence
\[\limsup_{R\uparrow\infty}\big\|H_{k,\lambda}^\st\big(\phi(\chi_R^\alpha)^2\big)-H_{k,\lambda}^\st\phi\big\|_{\Ld^2(\Omega)}\,\lesssim\,\limsup_{R\uparrow\infty}R^{-\alpha}\|\mathds1_{\Omega\setminus E_R}\nabla^\st\phi\|_{\Ld^2(\Omega)}.\]
Combining this with the G\aa rding inequality~(G2) of Step~3 yields
\[\limsup_{R\uparrow\infty}\big\|H_{k,\lambda}^\st\big(\phi(\chi_R^\alpha)^2\big)-H_{k,\lambda}^\st\phi\big\|_{\Ld^2(\Omega)}\,\lesssim_{k,\lambda}\,\limsup_{R\uparrow\infty}R^{-\alpha}\big(\|H_{k,\lambda}^\st\phi\|_{\Ld^2(\Omega)}+R^\frac12\|\phi\|_{\Ld^2(\Omega)}\big)\,=\,0,\]
and the claim follows.
\end{proof}

%%%%%%%%%%%%%%%%%%%%%
%%%%%%%%%%%%%%%%%%%%%
%%%%%%%%%%%%%%%%%%%%%

\section*{Acknowledgements}
The authors wish to thank Antoine Gloria, Sylvain Golenia, Felipe Hernandez, Laure Saint-Raymond, Johannes Sjöstrand, and Martin Vogel for motivating discussions at different stages of this work.
Financial support is acknowledged from the CNRS-Momentum program.

\section*{Data availability statement}
Data sharing is not applicable to this article as no new data were created or analyzed in this study.

\bigskip
\bibliographystyle{plain}

\def\cprime{$'$} \def\cprime{$'$} \def\cprime{$'$}

\end{document}